\documentclass[reqno]{amsart}

\usepackage{amsmath, amssymb, amsfonts, amsthm}

\theoremstyle{plain}
\newtheorem{theorem}{Theorem}[section] 

\newtheorem{corollary}[theorem]{Corollary}
\newtheorem{lemma}[theorem]{Lemma}
\newtheorem{proposition}[theorem]{Proposition}
\newtheorem{definition}[theorem]{Definition}

\newtheorem{assumption}[theorem]{Assumption}
\newtheorem*{assumpintro}{Assumption}

\theoremstyle{remark}
\newtheorem{remark}{Remark}[section]

\usepackage{graphicx}
\usepackage{bm}
\usepackage{stmaryrd}

\newcommand\C{\mathbb{C}}
\renewcommand\H{\mathbb{H}}
\newcommand\R{\mathbb{R}}

\newcommand\E{\mathbb{E}}

\newcommand\veps{\varepsilon}

\newcommand\cE{\mathcal{E}}
\newcommand\cF{\mathcal{F}}
\newcommand\cX{\mathcal{X}}
\newcommand\cR{\mathcal{R}}
\newcommand\cS{\mathcal{S}}
\newcommand\cT{\mathcal{T}}
\newcommand\cO{\mathcal{O}}
\newcommand\cQ{\mathcal{Q}}
\newcommand\cZ{\mathcal{Z}}

\newcommand\rI{\mathrm{I}}

\newcommand\rX{\mathrm{X}}

\newcommand\pa{\partial}
\newcommand\opa{\overline{\partial}{}}
\newcommand\paS{\partial_\cS}
\newcommand\opaS{\overline{\partial}{}_\cS}
\newcommand\pao{\partial_\omega}
\newcommand\opao{\overline{\partial}{}_\omega}
\newcommand\DS{\Delta_\cS}

\newcommand\Area{\operatorname{Area}}
\newcommand\crad{\operatorname{crad}}
\newcommand\cst{\operatorname{cst}}
\newcommand\diag{\operatorname{diag}}
\newcommand\diam{\operatorname{diam}}
\newcommand\dist{\operatorname{dist}}
\def\hm{\operatorname{hm}}
\renewcommand\Re{\operatorname{Re}}
\renewcommand\Im{\operatorname{Im}}
\newcommand\Id{\operatorname{Id}}
\newcommand\iso{\operatorname{iso}}
\newcommand\Int{\operatorname{Int}}

\newcommand\osc{\operatorname{osc}}
\renewcommand\Pr[2]{\operatorname{Pr}[#1\,;#2]}

\newcommand\dm{\diamond}
\newcommand\Dm{\diamondsuit}

\newcommand{\frb}{\mathfrak{b}}
\newcommand{\frw}{\mathfrak{w}}

\def\Od{\Omega^\delta}
\def\intr{{\operatorname{int}(\eta)}}

\def\LipKd{{\mbox{\textsc{Lip(}$\kappa$\textsc{,}$\delta$\textsc{)}}}}
\def\ExpFat{{{\mbox{\textsc{Exp-Fat(}$\delta$\textsc{)}}}}}
\def\ExpFatPrime{{{\mbox{\textsc{Exp-Fat(}$\delta,\delta'$\textsc{)}}}}}
\newcommand\Unif{{\mbox{\textsc{Unif(}$\delta$\textsc{)}}}}
\newcommand\Qflat{{\mbox{\textsc{Flat(}$\delta$\textsc{)}}}}

\newcommand{\stoptocwriting}{%
  \addtocontents{toc}{\protect\setcounter{tocdepth}{-5}}}
\newcommand{\resumetocwriting}{%
  \addtocontents{toc}{\protect\setcounter{tocdepth}{\arabic{tocdepth}}}}

\begin{document}

\title{Ising model and s-embeddings of planar graphs}

\author[Dmitry Chelkak]{Dmitry Chelkak$^\mathrm{a,b,c}$}

\thanks{\textsc{${}^\mathrm{A}$ ENS--MHI chair funded by MHI. D\'epartement de math\'ematiques et applications, \'Ecole Normale Sup\'erieure, CNRS, PSL University, 45 rue d'Ulm, 75005 Paris, France}}

\thanks{{\textsc{${}^\mathrm{B}$ \emph{Current address:} Department of Mathematics, University of Michigan, Ann Arbor, MI 48109-1043, USA}}}

\thanks{\textsc{${}^\mathrm{C}$} \emph{On leave from} \textsc{St.~Petersburg Dept.~of Steklov Mathematical Institute RAS, Fontanka 27, 191023 St.~Petersburg, Russia}}

\thanks{\emph{E-mail:} \texttt{dchelkak@umich.edu}}

\date{\today}

\begin{abstract}
We discuss the notion of s-embeddings~$\mathcal{S}=\mathcal{S}_\mathcal{X}$ of planar graphs carrying a nearest-neighbor Ising model. The construction of~$\mathcal{S}_\mathcal{X}$ is based upon a choice of a global complex-valued solution~$\mathcal{X}$ of the propagation equation for Kadanoff--Ceva fermions. Each choice of~$\mathcal{X}$  provides an interpretation of all other fermionic observables as s-holomorphic functions on~$\mathcal{S}_\mathcal{X}$. We set up a general framework for the analysis of such functions on s-embeddings~$\mathcal{S}^\delta$ with~$\delta\to 0$. Throughout this analysis, a key role is played by the functions~$\mathcal{Q}^\delta$ associated with~$\mathcal{S}^\delta$, the so-called origami maps in the bipartite dimer model terminology. In particular, we give an interpretation of the mean curvature of the limit of discrete surfaces $(\mathcal{S}^\delta;\mathcal{Q}^\delta)$ viewed in the Minkowski space~$\mathbb R^{2,1}$ as the mass in the Dirac equation describing the continuous limit of the model.

We then focus on the simplest situation when $\mathcal{S}^\delta$ have uniformly bounded lengths/angles and $\mathcal{Q}^\delta=O(\delta)$; as a particular case this includes all critical Ising models on {doubly periodic} graphs via their canonical s-embeddings. In this setup we prove RSW-type crossing estimates for the random cluster representation of the model and the convergence of basic fermionic observables. The proof relies upon a new strategy as compared to the already existing literature, {and} also provides a quantitative estimate on the speed of convergence.
\end{abstract}



\maketitle

\newpage

\tableofcontents

\newpage

\section{Introduction, main results and perspectives}
\setcounter{equation}{0}

\subsection{General context} The Ising model of a ferromagnet, introduced by Lenz in 1920, recently celebrated its centenary; probably, this is one of the most studied models in statistical mechanics. The \emph{planar} Ising model (i.e., the 2D model with nearest-neighbor interactions) gives rise to surprisingly rich structures of correlation functions; we refer the reader to monographs~\cite{Friedli-Velenik,McCoy-Wu,Mussardo,Palmer}, lecture notes~\cite{duminil-smirnov}, as well as to the introductions of the papers~\cite{AGG-20,ChSmi2} and references therein for more background on the subject, discussed from a variety of perspectives.

In this paper we consider the planar Ising model without the magnetic field, which is known to be exactly solvable on any graph and at any temperature: the partition function can be written as the Pffafian of a certain matrix and the entries of the inverse matrix -- known under the name \emph{fermionic observables} -- satisfy a simple \emph{propagation equation}; we refer the reader to the paper~\cite{CCK} for more details on various combinatorial formalisms used to study the planar Ising model during its long history.

We prefer to work with the ferromagnetic Ising model defined on \emph{faces} of a planar graph~$G$; the partition function is given by
\begin{equation}
\label{eq:intro-Zcirc}
\textstyle \cZ(G)\ =\ \sum_{\sigma:G^\circ\to\{\pm 1\}}\exp\big[\,\beta\sum_{e\in E(G)}J_e\sigma_{v^\circ_-(e)}\sigma_{v^\circ_+(e)}\,\big],
\end{equation}
where~$G^\circ$ and~$E(G)$ denote the dual graph and the set of edges of~$G$, respectively; $\beta=1/kT$ is the inverse temperature, $J_e>0$ are interaction constants assigned to the edges of~$G$, and~$v^\circ_{\pm}(e)\in G^\circ$ denote the two faces adjacent to~$e\in E(G)$. Passing to the \emph{domain walls} representation ({also known as the} low-temperature expansion) of the model, one can rewrite~$\cZ(G)$ as
\[
\textstyle \cZ(G)\ =\ 2\prod_{e\in E(G)}(x(e))^{-1/2}\times\,\sum_{C\in\cE(G)}\prod_{e\in C}x(e),\qquad x(e):=\exp[-2\beta J_e].
\]
where~$\cE(G)$ denotes the set of all even subgraphs of~$G$. Throughout this paper we identify edges of~$G$ with faces~$z(e)$ of the graph~$\Lambda(G):=G\cup G^\circ$ and denote 
\begin{equation}
\label{eq:x=tan-theta} \theta_{z(e)}\ :=\ 2\arctan x(e)\ \in\ (0,\tfrac{1}{2}\pi).
\end{equation}
Let us emphasize that we do \emph{not} fix an embedding of~$G$ into~$\C$ at this point, thus~$\theta_e$ is nothing more than another (abstract, i.e., not geometric) \emph{parametrization} of the interaction constants~$J_e$.

We are mostly interested in the situation when~$(G,x)$ is a big weighted graph carrying critical or near-critical Ising weights~$x(e)$ though we do \emph{not} precise the exact sense of this condition and mention it here only to give a proper perspective. For instance, the reader can think about the setup in which~$G$ is a subgraph of an infinite {doubly periodic} graph, in this case the criticality condition on the collection of weights~$x(e)$ is known explicitly~\cite{Li,cimasoni-duminil}. However, this periodic setup is \emph{not} the main motivation for our paper and should be viewed as a very particular case though our results are new even in this situation and, in particular, answer a question posed by Duminil-Copin and Smirnov in their lecture notes on the conformal invariance of 2D lattice models; see~\cite[Question~8.5]{duminil-smirnov}.

Given a weighted graph~$(G,x)$ we aim to embed it into the complex plane~$\C$ (actually, we construct both an embedding of $\Lambda(G):=G\cup G^\circ$ and of its dual graph~$\Dm(G)$) in a way allowing to analyze (subsequential) limits of fermionic observables in the same spirit as in the seminal work of Smirnov~\cite{Smi-ICM06,Smi-I} on the critical Ising model on~$\mathbb{Z}^2$. A possible analogy could be Tutte's barycentric embeddings which, among other things, provide a framework to study the convergence of discrete harmonic functions to continuous ones. Let us emphasize that \emph{s-embeddings} introduced in our paper are \emph{not} directly related to these barycentric embeddings; in fact both can be viewed as special cases of a more general construction called \emph{t-embeddings} or \emph{Coloumb gauges} appearing in the bipartite dimer model context~\cite{KLRR,CLR1}. The notion of s-embeddings was first announced in~\cite{Ch-ICM18} and motivated a closely related research project~\cite{CLR1,CLR2}. We benefit from this interplay and partly rely upon results obtained in~\cite{CLR1} in a more general context.

One of the most important (from our perspective) motivations to study special embeddings of irregular weighted planar graphs into~$\C$ is the conjectural convergence of critical random maps carrying a lattice model to an object called \emph{Liouville Quantum Gravity}; e.g. see~\cite{duplantier-sheffield-LQG,garban-Bourbaki} or~\cite{rhodes-vargas-lectures} and references therein. Though this discussion goes far beyond the scope of this paper, let us briefly mention several ideas of this kind that appeared in the literature: circle packing embeddings (see~\cite{nachmias-lectures} and references therein), embeddings as piecewise flat Riemann surfaces (e.g., see~\cite{gill-rohde} or~\cite[Section~1.7.1]{garban-Bourbaki}), embeddings by square tilings for UST-decorated random maps~\cite{smirnov-boykiy-private}, Tutte's embeddings~\cite{gwynne-miller-sheffield,gwynne-holden-sun-survey} and circle packings~\cite{gurel-jerison-nachmias} for the mating-of-trees approach to the LQG, the so-called Cardy embeddings~\cite{holden-sun-Cemb} introduced recently in the pure gravity context, etc. We believe that adding s-embeddings to this list might be a proper path for the analysis of random maps carrying the Ising model.

Besides generalizing the results of~\cite{Smi-I,ChSmi2} to a large family of deterministic graphs with `flat' functions~$\cQ^\delta=O(\delta)$ {and suggesting} a general framework to attack problems coming from the random maps context, this paper also contains the following contributions to the subject:
\begin{itemize}
\item A new `constructive' strategy of the proof of Theorem~\ref{thm:FK-conv}, which avoids compactness arguments in the Carath\'eodory topology on the set of planar domains and gives a quantitative estimate of the speed of convergence; see Section~\ref{sub:proof-FK} for details. Due to a general framework recently developed in~{\cite{richards-thesis}} this also allows to control the speed of convergence of FK-Ising interfaces to SLE(16/3) curves; 
    {see~\cite{richards-thesis} for more details.}
\smallskip
\item In the non-flat case $\cQ^\delta\not\to 0$ as~$\delta\to 0$, we reveal the importance of embeddings of planar graphs carrying a nearest-neighbor Ising model into the \emph{Minkowski space}~$\R^{2,1}$ and not simply into the complex plane. As briefly discussed in Section~\ref{sub:shol-limits}, this leads to a rather unexpected interpretation of the \emph{mass} appearing in the effective description of the model in the small mesh size limit as the \emph{mean curvature} of the corresponding surface in~$\R^{2,1}$.
\end{itemize}
The latter observation can be used in deterministic setups (e.g., see a recent paper~\cite{CIM-massive} where such an interpretation of a near-critical model on isoradial grids is discussed) and leads to interesting questions on limits of random surfaces in~$\R^{2,1}$ obtained via s-embeddings of appropriate random maps carrying the Ising model.

\subsection{Fermionic observables, s-embeddings and s-holomorphicity} Appearances of fermionic observables in the planar Ising model can be traced back to the seminal work of Kaufman and Onsager. Below we rely upon the classical spin-disorders formalism of Kadanoff and Ceva~\cite{kadanoff-ceva}, which is briefly reviewed in Section~\ref{sub:notation} below; see also~\cite{CCK} for a discussion of links of this formalism with other combinatorial techniques such as Kac--Ward matrices or dimer representations.
The notation used in this paper agrees with that of~\cite[Section~3]{CCK} and with~\cite{Ch-ICM18}.

Let
\begin{equation}
\label{eq:KC-chi-def}
\chi_c:=\mu_{v^\bullet(c)}\sigma_{v^\circ(c)},
\end{equation}
where~$v^\bullet(c)\in G=G^\bullet$ and~$v^\circ(c)\in G^\circ$ are incident primal and dual vertices of the planar graph~$G$  and~$\mu_{v^\bullet}$ is the so-called \emph{disorder variable} associated to~$v^\bullet$; we briefly recall its definition in Section~\ref{sub:notation}. The {Kadanoff--Ceva fermions}~\eqref{eq:KC-chi-def} are indexed by edges 
of the graph~$\Lambda(G):=G^\circ\cup G^\bullet$ or, equivalently, by vertices of the medial graph~$\Upsilon(G)$ of~$\Lambda(G)$. Given a set \mbox{$\varpi=\{v_1^\bullet,\ldots,v_{m-1}^\bullet,v_1^\circ,\ldots,v_{n-1}^\circ\}\subset\Lambda(G)$} with~$n,m$ even, one can consider the expectation
\begin{equation}
\label{eq:KC-fermions}
X_{\varpi}(c):=\E[\,\chi_c\cO_\varpi[\mu,\sigma]\,],\qquad \cO_\varpi[\mu,\sigma]:=\mu_{v_1^\bullet}\ldots\mu_{v_{m-1}^\bullet}\sigma_{v_1^\circ}\ldots\sigma_{v_{n-1}^\circ},
\end{equation}
below we call such expectations \emph{Kadanoff--Ceva fermionic observables}. Though $X_\varpi(c)$ is a priori defined only up to the sign, one can avoid making non-canonical choices by passing to a \emph{double cover}~$\Upsilon^\times_\varpi(G)$ of the graph~$\Upsilon(G)$ (e.g., see~\cite[Fig.~4]{CIM-massive}) that branches over all faces of~$\Upsilon(G)$ except those from the set~$\varpi$.

\smallskip

\begin{figure}
\hskip 0.02\textwidth \begin{minipage}{0.44\textwidth}
\includegraphics[clip, trim=4.5cm 16.8cm 10.2cm 5.4cm, width=\textwidth]{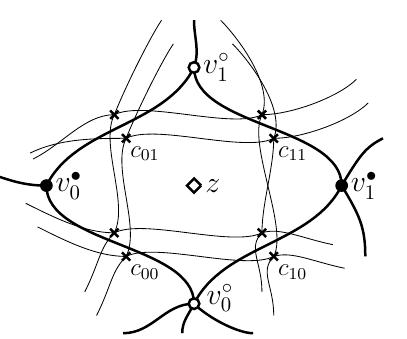}
\end{minipage}\hskip 0.08\textwidth \begin{minipage}{0.44\textwidth}
\includegraphics[clip, trim=12.4cm 16.8cm 2.3cm 5.4cm, width=\textwidth]{Sdef_tikz.pdf}
\end{minipage}

\caption{The notation for vertices {and `corners'} near~$z\in\Dm(G)$, a piece of the double cover~$\Upsilon^\times(G)$ {(left) and a piece of the corresponding bipartite graph~$\Upsilon^\bullet(G)\cup\Upsilon^\circ(G)$ (right).} Sometimes we fix a choice of the lifts of the `corners'~$c_{pq}$ to the double cover so that~$c_{11}$\,---\,$c_{01}$\,---\,$c_{00}$\,---\,$c_{10}$ are connected on~$\Upsilon^\times(G)$.\label{fig:notation}}
\end{figure}

Let~$z\in\Dm(G)$ and~$v_0^\bullet,v_0^\circ,v_1^\bullet,v_1^\circ\in\Lambda(G)$ be the vertices of the quad~$z$ listed counterclockwise, see {Fig.~\ref{fig:notation}} for the notation. It is well known (see~\cite[Section~4.3]{Mercat-CMP}, \cite[Remark~4]{Smi-ICM10} and~\cite[Section~3.5]{CCK} for historical comments) that Kadanoff--Ceva fermionic observables satisfy a very simple three-term propagation equation with coefficients determined by the Ising interaction parameter~\eqref{eq:x=tan-theta}: 
\begin{equation}
\label{eq:3-terms}
X(c_{pq})=X(c_{p,1-q})\cos\theta_z+X(c_{1-p,q})\sin\theta_z,
\end{equation}
where the corner~$c_{pq}\in\Upsilon^\times_\varpi(G)$ corresponds to the edge~$(v^\bullet_p v^\circ_q)$ of~$\Lambda(G)$ and we assume that the lifts of~$c_{pq}$, $c_{p,1-q}$ and of $c_{1-p,q}$ from the medial graph~$\Upsilon(G)$ to its double cover (see {Fig.~\ref{fig:notation}}) are chosen so that they remain connected on~$\Upsilon^\times_\varpi(G)$. Note that all solutions of the propagation equation~\eqref{eq:3-terms} are \emph{spinors} on~$\Upsilon^\times_\varpi(G)$, i.e., their values at two lifts of the same~$c\in\Upsilon(G)$ differ by the sign.

\smallskip

Let a {proper} embedding~$\cS:\Lambda(G)\to\C$ of a planar graph~$\Lambda(G)$ into the complex plane be fixed and denote
\begin{equation} \label{eq:def-eta}
\eta_c:=\varsigma\cdot \exp\big[-\tfrac{i}{2}\arg(\cS(v^\bullet(c))-\cS(v^\circ(c)))\big],\qquad \varsigma:=e^{i\frac{\pi}{4}},
\end{equation}
where a global prefactor $\varsigma\in\mathbb{T}=\{z\in\C:|z|=1\}$ can be chosen arbitrarily; we choose the value~$e^{i\frac\pi 4}$ in order to keep the notation of this paper consistent with~\cite{ChSmi2}. As above, one can avoid an ambiguity in the values of square roots in the definition~\eqref{eq:def-eta} by passing to the double cover $\Upsilon^\times(G)$. Clearly, the products
\begin{equation}
\label{eq:etaX-product}
\eta_c X_\varpi(c):\Upsilon_\varpi(G)\to\C
\end{equation}
are spinors defined on the double cover~$\Upsilon_\varpi(G)$ that branches \emph{only} over~$\varpi$.

\smallskip

Assume now that we work with the critical Ising model on~$\mathbb{Z}^2$ or, more generally, with the critical Z-invariant model on isoradial {grids}; see~\cite{ChSmi2}. In this context, essentially the same objects as~\eqref{eq:KC-fermions} are sometimes called \emph{Smirnov's fermionic observables}. The correspondence reads as follows (e.g., see~\cite[Lemma~3.4]{ChSmi2}):
\begin{itemize}
\item in the setup of the \emph{critical Z-invariant} Ising model on an \emph{isoradial {grids,}} a real-valued spinor~$X_\varpi$ satisfies the propagation equation~\eqref{eq:3-terms} around a quad~$z\in\Dm(G)$ if and only if there exists a number~$F_\varpi(z)\in\C$ such that
    \begin{equation}
    \label{eq:Pr=Pr-iso}
    \eta_c X_\varpi(c)\ =\ \Pr{F_\varpi(z)}{\eta_c\R}\,:=\,\tfrac{1}{2}\big(\,F_\varpi(z)+\eta_c^2\cdot \overline{F_\varpi(z)}\,\big).
    \end{equation}
\end{itemize}
Note that~$F_\varpi:\Dm_\varpi(G)\to\C$ has the same branching structure as~\eqref{eq:etaX-product}: this is a spinor defined on the double cover~$\Dm_\varpi(G)$ of~$\Dm(G)$ that branches over~$\varpi$. The equation~\eqref{eq:Pr=Pr-iso}, reformulated in terms of the values~$F_\varpi(z)$, is nothing but the definition of \emph{s-holomorphic functions} on an isoradial grid; see~\cite[Definition~3.1]{ChSmi2}.

\smallskip

Though a correspondence similar to~\eqref{eq:Pr=Pr-iso} can be established in more general contexts (e.g., see~\cite{beffara-duminil,park} for the massive perturbation of the Ising model on the square grid or~\cite[Section~3.6]{CCK}), let us emphasize the following fundamental difference between Kadanoff--Ceva fermionic observables and s-holomorphic functions:
\begin{itemize}
\item Kadanoff--Ceva observables~$X_\varpi$ are defined on double covers~$\Upsilon^\times_\varpi(G)$ of an \emph{abstract} planar graph~$\Upsilon(G)$ while
\item Smirnov's observables~\cite{Smi-ICM06,Smi-I} and their spinor generalizations~$F_\varpi$ require to fix an \emph{embedding}~$\cS$ of the graphs under consideration.
\end{itemize}

\smallskip

We now move to a slightly informal definition of \emph{s-embeddings} of planar graphs carrying an Ising model, a more detailed presentation is given in Section~\ref{sub:semb-definition}.
\begin{definition}\label{def:intro-semb}
Let~$\cX:\Upsilon^\times(G)\to\C$ be a \emph{complex-valued} solution of the propagation equation~\eqref{eq:3-terms}, which we call a \emph{Dirac spinor}. We say that~$\cS=\cS_\cX:\Lambda(G)\to\C$ is an \emph{s-embedding} (associated with~$\cX$) of the weighted graph~$(G,x)$ if
\begin{equation}
\label{eq:Sdef}
\cS_\cX(v^\bullet(c))-\cS_\cX(v^\circ(c))\ =\ (\cX(c))^2\ \ \text{for all\ \ $c\in\Upsilon(G)$}.
\end{equation}
\end{definition}

\begin{remark}
(i) The consistency of~\eqref{eq:Sdef} around a quad~$z\in\Dm(G)$, i.e., the identity
\[
(\cX(c_{10}(z)))^2+(\cX(c_{01}(z)))^2=(\cX(c_{00}(z)))^2+(\cX(c_{11}(z)))^2
\]
can be easily derived from~\eqref{eq:3-terms}; see {Fig.~\ref{fig:notation}} for the notation.

\smallskip

\noindent (ii) In the `standard' context of isoradial {grids} (already embedded into~$\C$ so that all quads~$z\in\Dm(G)$ are \emph{rhombi} with the sides of length~$\delta$), the function~$\eta_c$ given by~\eqref{eq:def-eta} solves the propagation equation~\eqref{eq:3-terms} and thus can be chosen as a Dirac spinor~$\cX(c):=\varsigma\,\delta^{1/2}\overline{\eta}_c$ in Definition~\ref{def:intro-semb}. With this choice, the s-embedding~$\cS_\cX$ coincides with the original isoradial grid. We discuss more particular cases of \mbox{s-embeddings} in Section~\ref{sub:semb-definition}. In particular, {each doubly periodic} graph carrying a critical Ising model admits a canonical {doubly periodic} s-embedding.
\end{remark}

One can also easily see from~\eqref{eq:3-terms} that
\[
|\cX(c_{10}(z))|^2+|\cX(c_{01}(z))|^2=|\cX(c_{00}(z))|^2+|\cX(c_{11}(z))|^2,
\]
which means that~$(\cS(v_0^\bullet)\cS(v_0^\circ)\cS(v_1^\bullet)\cS(v_1^\circ))$ is a \emph{tangential} (though not necessarily convex) quadrilateral in the complex plane. Let a function $\cQ^\delta:\Lambda(G)\to\R$ be defined (up to a global additive constant) by the identity
\[
\cQ^\delta(v^\bullet(c))-\cQ^\delta(v^\circ(c))\ :=\ |\cS^\delta(v^\bullet(c))-\cS^\delta(v^\circ(c))|\,.
\]
{(Note that~$\cQ=0$ on~$G^\circ$ and~$\cQ=\delta$ on~$G^\bullet$ if~$\cS$ is a rhombic lattice of mesh~$\delta$.)}

{
\begin{remark} \label{rem:remintro-R21etc}
Clearly, a multiplication of the Dirac spinor~$\cX$ by a global factor~$e^{is}$, $s\in\R$, results in the rotation~$\cS\mapsto e^{2is}\cS$ of the s-embedding. Less trivially, if one replaces~$\cX$ by~$\cosh(t)\cX+\sinh(t)\overline{\cX}$, $t\in\R$, then~$\cS$ and~$\cQ$ change as
\[
(\Re\cS\,;\,\Im\cS\,;\,\cQ)\ \mapsto\ (\cosh(2t)\Re\cS+\sinh(2t)\cQ\,;\,\Im\cS\,;\,\sinh(2t)\Re\cS+\cosh(2t)\cQ),
\]
which is an isometry in the Minkowski space~$\R^{2,1}$; see Fig.~\ref{fig:Z2-emb}.
This means that a natural viewpoint on s-embeddings is to consider them not simply as tilings~$\cS$ of the complex plane but rather as discrete surfaces $(\cS;\cQ)$ in $\C\times\R\cong \R^{2,1}$. Though in this paper we are mostly interested in the `flat' setup when~$\cQ^\delta\to 0$ as~$\delta\to 0$, the reader can find a discussion of the general case in Section~\ref{sub:shol-limits} below.
\end{remark}}

\begin{figure}
\begin{minipage}{0.28\textwidth}
\includegraphics[clip, trim=9.43cm 17.1cm 6.07cm 6.35cm, width=\textwidth]{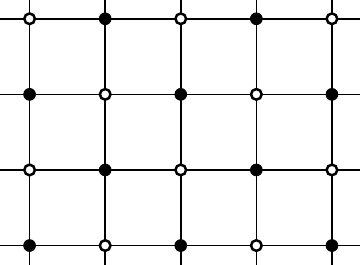}

\vskip 12pt

\includegraphics[clip, trim=7.75cm 8cm 7.75cm 15cm, width=\textwidth]{Semb_Z2.pdf}
\end{minipage}
\hskip 0.12\textwidth
\begin{minipage}{0.28\textwidth}
\includegraphics[clip, trim=9.2cm 6.1cm 6.3cm 11.4cm, width=\textwidth]{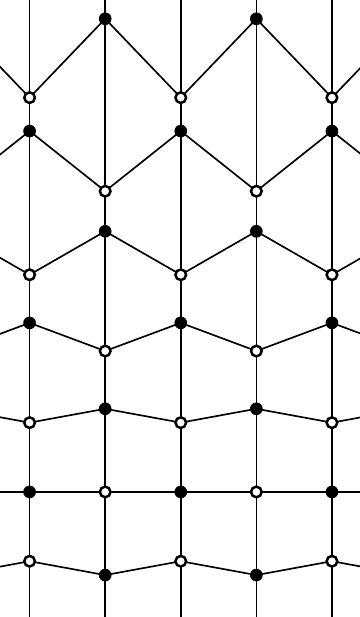}
\end{minipage}

\caption{\textsc{Left:} Two periodic s-embeddings of the \emph{same} critical homogeneous Ising model on~$\mathbb{Z}^2$. The corresponding discrete surfaces~$(\cS;\cQ)$ in the space~$\R^{2,1}$ are isometric; see Remark~\ref{rem:remintro-R21etc}. \textsc{Right:} An s-embedding of the non-critical homogeneous Ising model on~$\mathbb{Z}^2$; {see also Remark~\ref{rem:massive-iso} below and~\cite[Section~5.2]{CHM}.} \label{fig:Z2-emb}}
\end{figure}

By a direct computation (see Proposition~\ref{prop:shol=3term}), it is not hard to see that Smirnov's interpretation~\eqref{eq:Pr=Pr-iso} of the propagation equation~\eqref{eq:3-terms} \emph{can} be generalized to the context of s-embeddings. More precisely, the following equivalence holds:
\begin{itemize}
\item Let~$\cS=\cS_\cX$ be an s-embedding of the weighted graph~$(G,x)$ associated with a Dirac spinor~$\cX$. Then, a real-valued spinor~$X_\varpi$ satisfies the propagation equation~\eqref{eq:3-terms} around a quad~$z\in\Dm(G)$ if and only if there exists a number~$F_\varpi(z)\in\C$ such that
    \begin{equation}
    \label{eq:Pr=Pr-semb}
    \eta_c X_\varpi(c)\ =\ |\cX(c)|\cdot \Pr{F_\varpi(z)}{\eta_c\R}\,.
    \end{equation}
\end{itemize}
This equivalence gives rise to a notion of \emph{s-holomorphic functions on s-embeddings}, see Section~\ref{sub:dimers} for more details and for a link with more general t-holomorphic functions on \mbox{\emph{t-embeddings}} of weighted bipartite planar graphs, which were recently introduced and studied in~\cite{CLR1}, in particular motivated by the Ising model context.

\smallskip

\begin{figure}

\centering \includegraphics[clip, trim=5.24cm 13.24cm 6.36cm 9.16cm, width=\textwidth]{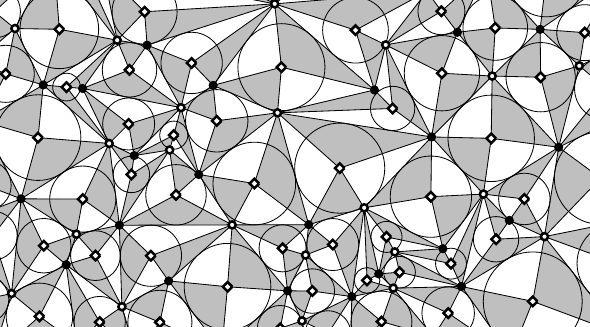}

\caption{A {`generic'} s-embedding. The images~$\cS(v^\bullet)$, $\cS(v^\circ)$ of vertices~$v^\bullet\in G^\bullet$, $v^\circ\in G^\circ$ are shown as black and white discs; the points~$\cS(z)$,~$z\in\Dm(G)$, are shown as white rhombi. Note that we do \emph{not} require the convexity of tangential quads~$\cS^\dm(z)$ and that we allow vertices of degree~$2$ and multiple edges in both $G^\bullet$ and~$G^\circ$.\label{fig:Semb}}
\end{figure}

Let us informally summarize the preceding discussion as follows:
\begin{itemize}
\item Each choice of a Dirac spinor~$\cX$ -- i.e., of a complex-valued solution of the propagation equation~\eqref{eq:3-terms} considered on an \emph{abstract} graph~$\Upsilon^\times(G)$ -- provides an {interpretation}~\eqref{eq:Pr=Pr-semb} of all other solutions~$X_\varpi$ of~\eqref{eq:3-terms} via a \emph{discrete complex structure} defined by the s-embedding~$\cS_\cX$ {(or, more precisely, by the discrete surface~$(\cS_\cX;\cQ_\cX)$ in~$\R^{2,1}$; see Remark~\ref{rem:remintro-R21etc} above).}
\end{itemize}

We conclude this section by recalling the breakthrough idea of \mbox{Smirnov~\cite{Smi-ICM06,Smi-I}}, which served as a cornerstone for a series of works on the convergence of correlation functions in the critical Ising model on~$\mathbb{Z}^2$ to the scaling limits predicted by the CFT (e.g., see a brief survey~\cite{Ch-ECM16} of these developments and references therein):
\begin{itemize}
\item Working with discrete domains~$\Od\to\Omega$, $\delta\to 0$, one can view s-holomorphic functions~$F_\varpi$ as solutions to certain discrete Riemann-type boundary value problems, and use appropriate discrete complex analysis techniques to prove the convergence of these solutions to their continuous counterparts.
\end{itemize}

In our paper we continue to develop this philosophy. However, let us emphasize once more an important addition to the guideline described above: in many interesting setups, one should first find an appropriate embedding of an abstract planar graph (or to re-embed a given graph properly) so that discrete complex analysis techniques become available. This is what s-embeddings were suggested for in~\cite{Ch-ICM18}.

\subsection{Assumptions and main convergence results}
The `algebraic' part (see Sections~\ref{sec:definitions} and~\ref{sec:operators}) of this paper, which includes definitions of relevant discrete differential operators and algebraic identities between them, is developed in the full generality. However, the proof of our main convergence result -- Theorem~\ref{thm:FK-conv} -- requires many `analytic' estimates and at the moment we are able to give it only under very restrictive assumptions~\Unif\ and~\Qflat\ on a family of s-embeddings~$\cS^\delta$ with~$\delta\to 0$, which are discussed below. On the other hand, even this setup provides a much greater generality as compared to the isoradial context; we refer the reader to Section~\ref{sub:further} for a further discussion.

\begin{assumpintro}[\Unif] There exists constants~$r_0,R_0,\theta_0>0$ such that all edge lengths of~$\cS^\delta$ are uniformly comparable to~$\delta$:
\[
r_0\delta\ \le\ |\cX^\delta(c)|^2=|\cS^\delta(v^\bullet(c))-\cS^\delta(v^\circ(c))|\ \le\ R_0\delta
\]
and all angles of quads in~$\cS^\delta$ are uniformly bounded from below by $\theta_0$. (This assumption also implies that all the interaction parameters~\eqref{eq:x=tan-theta} are uniformly bounded away from~$0$ and from~$1$; e.g., see the formula~\eqref{eq:theta-from-S}.)
\end{assumpintro}

\begin{assumpintro}[\Qflat] With a proper choice of global additive constants, the functions~$\cQ^\delta$ satisfy the uniform (both in space and in~$\delta$) estimate~$|\cQ^\delta(v)|=O(\delta)$.
\end{assumpintro}
Let us discuss the last assumption in several particular cases of s-embeddings:
\begin{itemize}
\item In the setup of isoradial {grids} one has~$\cQ^\delta=0$ on~$G^\circ$ and~$\cQ^\delta=\delta$ on~$G^\bullet$.
\smallskip
\item For the critical Ising model on circle patterns introduced by Lis in~\cite{Lis}, which generalizes the isoradial setup, the functions~$\cQ^\delta$ vanish on~$G^\bullet$ and are equal to (minus) the radii of the corresponding circles on~$G^\circ$.
\smallskip
\item For {doubly periodic} graphs carrying a critical Ising model, the existence of periodic s-embeddings~$\cS^\delta$ with \emph{periodic} (and thus~$O(\delta)$) functions~$\cQ^\delta$ is guaranteed by~\cite[Lemma~13]{KLRR}, see also Lemma~\ref{lem:double-periodic} below.
\end{itemize}

Our main convergence result is an analogue of~\cite[Theorem~2.2]{Smi-I} and~\cite[Theorem~A]{ChSmi2} in the setup of s-embeddings. Let~$\Od$ be simply connected discrete domains drawn on s-embeddings~$\cS^\delta$, with wired boundary conditions on the arc~$(a^\delta b^\delta)$ and free boundary conditions on the complementary arc~$(b^\delta a^\delta)$; see Fig.~\ref{fig:Dob}. Let
\begin{equation}
\label{eq:KC-Dob-def}
X^\delta(c)\ :=\ \E_{\Omega^\delta}[\chi_c\mu_{(b^\delta a^\delta)^\bullet}\sigma_{(a^\delta b^\delta)^\circ}]
\end{equation}
be the basic fermionic observable in the domain~$\Omega^\delta$ and~$F^\delta$ be the corresponding s-holomorphic function defined by~\eqref{eq:Pr=Pr-semb}.

\begin{figure}
\includegraphics[clip, trim=4.4cm 14.4cm 4.8cm 4cm, width=0.91\textwidth]{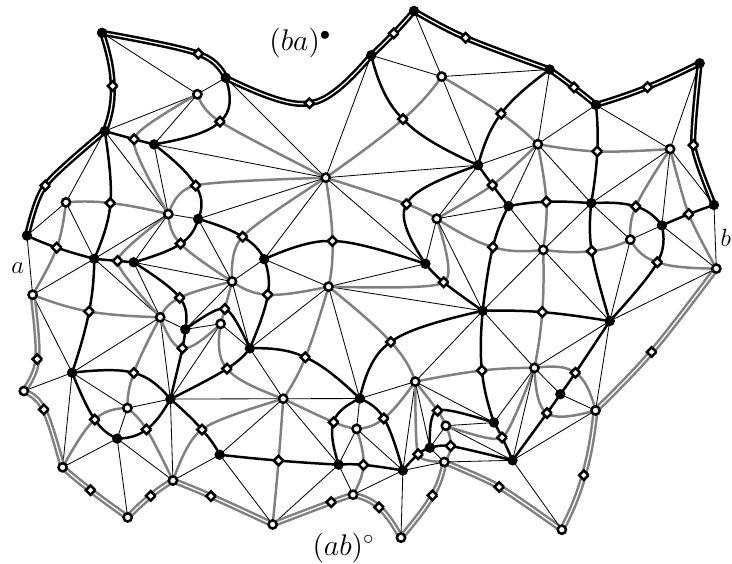}
\caption{An example of a domain with wired/free boundary conditions drawn on an s-embedding as in Theorem~\ref{thm:FK-conv}. We view the boundary arc~$(ab)^\circ$ as a single vertex of the graph~$G^\circ$ and the boundary arc~$(ba)^\bullet$ as a single vertex of the dual graph~$G^\bullet$.\label{fig:Dob}}
\end{figure}

\begin{theorem}\label{thm:FK-conv}
Let s-embeddings~$\cS^\delta$, $\delta\to 0$, satisfy the assumptions \Unif\ and \Qflat. Assume that discrete domains~$(\Od;a^\delta,b^\delta)$ drawn on~$\cS^\delta$ converge to a bounded simply connected domain~$(\Omega;a,b)$ in the Carath\'eodory sense (e.g., see~\cite[Section~3.2]{ChSmi1} for the definition). Then, for each compact subset~$K$ of $\Omega$, the following uniform convergence holds as~$\delta\to 0$:
\[
F^\delta(z)\ =\ \sqrt{\Phi'(\cS^\delta(z))}+o(1),\qquad \cS^\delta(z)\in K,
\]
where~$\Phi:\Omega\to\R\times(0,1)$ is a conformal mapping that sends~$a,b$ to~$\pm\infty$.
\end{theorem}

To prove Theorem~\ref{thm:FK-conv} we use a strategy that substantially differs from the one used in~\cite{Smi-I,ChSmi2}. The reason is that both these papers heavily rely upon a very peculiar comparison of the {functions~$H^\delta:=\int \Im[(F^\delta(z)^2)]dz$ constructed out of s-holomorphic observables $F^\delta$} with sub- and super-harmonic functions on~$G^\bullet$ and~$G^\circ$. Though this \emph{s-positivity} property has an analogue in the setup of s-embeddings (see Corollary~\ref{cor:s-positivity}), the link with harmonic functions is missing, which forces us to develop a totally different approach.

{In~\cite{Smi-I,ChSmi2}, the aforementioned sub-/super-harmonicity of~$H^\delta$ was used for two purposes: (i) to get a priory regularity estimates of~$F^\delta$ in the bulk of~$\Omega^\delta$ and, more importantly, (ii) to prove that the $0/1$ Dirichlet boundary values of $H^\delta$ survive in the limit~$\delta\to 0$. As for (i), one can relatively easily develop an alternative proof following the approach from~\cite[Section~6]{CLR1}; see Section~\ref{sub:regularity} below. However, the analysis near the boundary is more challenging. Loosely speaking, in the standard square grid setup our approach reads as follows. In a very thin~$\delta^{1-\eta}$ neighborhood of~$\partial\Omega^\delta$, the function~$H^\delta$ can be controlled via a `crude' estimate~$F^\delta=O(\delta^{-\frac{1}{2}+\eta})$ coming from the $O(\delta^\eta)$ decay of the magnetization as~$\delta\to 0$. Inside the $\delta^{1-\eta}$ interior of~$\Omega^\delta$ we use the fact that the Laplacian of (a mollified version of) $H^\delta$ is bounded by~$O(\delta/d^3)$ and hence by~$O(d^{-2+\eta})$, where $d$ denotes the distance to the boundary. This guarantees (see Lemma~\ref{lem:cont-estimate}) that~$H^\delta=O(d^{\eta/5})$, uniformly in both~$d$ and~$\delta$.} It is worth noting that this new strategy of the proof also provides a quantitative estimate (of the form~$O(\delta^\alpha)$, $\alpha>0$) on the speed of convergence. {We believe that it can also be used to give an alternative proof of an analogue of Theorem~\ref{thm:FK-conv} for the massive model on~$\mathbb{Z}^2$ and, more generally, on isoradial grids; the result recently obtained by Park~\cite{park-iso} via a different approach to the control of boundary values.}

\smallskip

One of essential ingredients of our proof of Theorem~\ref{thm:FK-conv} {(that, in particular, gives us the polynomial decay of the magnetization on~$\cS^\delta$)} is the so-called \emph{uniform crossing estimates} for the random cluster (or Fortuin--Kasteleyn; e.g., see~\cite{duminil-smirnov}) representation of the Ising model on~$\cS^\delta$.
This is why we also prove a supplementary theorem, which can be viewed as a \emph{weak} analogue of~\cite[Theorem~C]{ChSmi2} for very special \emph{`discrete rectangles'} on s-embeddings satisfying~$\Unif$\ and~$\Qflat$.

\begin{theorem}\label{thm:RSW-selfdual} Let~$x_1<x_2$,~$y_1<y_2$, $\cR:=(x_1,x_2)\times(y_1,y_2)\subset\C$, and~$\cS^\delta$ be s-embeddings satisfying the assumptions~\Unif\ and~\Qflat\ with $\delta\to 0$. There exists discretizations~$\cR^\delta=[\cR(x_1,x_2;y_1,y_2)]^{\circ\bullet\circ\bullet}_{\cS^\delta}$\,of the rectangle~$\cR$ with boundaries staying within~$O(\delta)$ from the boundaries of~$\cR$, and with wired boundary conditions at the bottom~$(a^\delta_1b^\delta_1)^\circ$ and the top~$(a^\delta_2b^\delta_2)^\circ$ sides of~$\cR^\delta$, and free at the left~$(b^\delta_2a^\delta_1)^\bullet$ and the right~$(b^\delta_1a^\delta_2)^\bullet$ sides (see also Fig.~\ref{fig:SDdomain}) such that
\[
\liminf_{\delta\to 0}\E_{\cR^\delta}[\sigma_{(a^\delta_1b^\delta_1)^\circ}\sigma_{(a^\delta_2b^\delta_2)^\circ}]\ \ge\ \cst\,>\,0,
\]
where the constant~$\cst>0$ depends only on constants in~\Unif\ and~\Qflat.
\end{theorem}
\begin{remark} (i) Under assumptions~\Unif\ and \Qflat, we consider Theorem~\ref{thm:RSW-selfdual} to be much less conceptual than Theorem~\ref{thm:FK-conv} and put the former after the latter because of this reason; {see also Section~\ref{subsub:RSW-Mahfouf} for a discussion of a very recent work by Mahfouf that vastly generalizes Theorem~\ref{thm:RSW-selfdual}.} Still, let us emphasize that Theorem~\ref{thm:RSW-selfdual} and Corollary~\ref{cor:circuits} actually \emph{precede} Theorem~\ref{thm:FK-conv} in our approach. The proof of Theorem~\ref{thm:FK-conv} is given in Section~\ref{sec:convergence} and relies upon the results from Section~\ref{sec:operators}; the proof of Theorem~\ref{thm:RSW-selfdual} is given in Section~\ref{sec:RSW} and can be read directly after Section~\ref{sec:definitions}.

\smallskip

\noindent (ii) Though the guideline idea (due to Smirnov, cf.~\cite{kemppainen-smirnov-I}) of using fermionic observables to prove Theorem~\ref{thm:RSW-selfdual} is exactly the same as in~\cite{ChSmi2}, its implementation is also totally different from~\cite{ChSmi2} due to the lack of comparison with harmonic functions. In particular, a special construction of `straight boundaries' of rectangles~$\cR^\delta$ on plays a crucial role in the proof. Of course, \emph{a posteriori}, Theorem~\ref{thm:RSW-selfdual} holds for all discretizations due to the monotonicity with respect to boundary conditions.
\end{remark}

We conclude this section by a standard corollary of Theorem~\ref{thm:RSW-selfdual}, which is exactly the prerequisite mentioned above that we need for the proof of Theorem~\ref{thm:FK-conv}. Given ${u}\in\C$ and~$d>0$, denote
\begin{align*}
\boxbox(u,d)\ :=\ ([\Re u-3d,\Re u+3d]&\times [\Im u-3d,\Im u+3d])\\
&\smallsetminus((\Re u-d,\Re u+d)\times(\Im u-d,\Im u+d))
\end{align*}
and let~$\mathbb{P}^{\operatorname{free}}_{\boxbox(u,d)}$ be the probability in the random cluster (or Fortuin-Kasteleyn) representation of the Ising model with free boundary conditions on both the outer and the inner boundaries of the annulus~$\boxbox_d(u)$.

\begin{corollary}\label{cor:circuits} There exists constants~$L_0>0$ and~$p_0>0$ depending only on constants in the assumptions~\Unif\ and \Qflat\ such that for all~$u\in\C$, $d\ge L_0\delta$, and all s-embeddings~$\cS^\delta$ satisfying \Unif, \Qflat\ and covering the disc~$B(u,5d)$, the following is fulfilled:
\[
\mathbb{P}^{\operatorname{free}}_{\boxbox(u,d)}\bigl[\,\mathrm{there~exists~an~{open}~circuit~in}~{\boxbox(u,d)}\,\bigr]\ \ge\ p_0.
\]
A similar uniform estimate holds for the dual model.
\end{corollary}
\begin{remark} It is well known that Theorem~\ref{thm:FK-conv} and Theorem~\ref{thm:RSW-selfdual} together imply the convergence of interfaces in the FK-Ising model with Dobrushin boundary conditions to Schramm's SLE(16/3) curves. We refer the reader to~\cite{kemppainen-smirnov-I} for the analysis of interfaces in the bulk of~$\Omega$ and to~\cite{karrila-boundary} for a treatment near the prime ends~$a,b$.

 However, to extend this result to the convergence of full \emph{loop ensembles} as in~\cite{kemppainen-smirnov-III} (see also~\cite{benoist-hongler}, \cite{garban-wu} and a brief discussion in~\cite[Section~5]{Ch-ICM18}) one needs stronger crossing estimates similar to those obtained in~\cite{CDCH} by methods specific to the isoradial setup. {Such `strong RSW' estimates on s-embeddings are missing as for now though we hope that the strategy used in a recent paper~\cite{DCMT-fractal} together with techniques developed in Section~\ref{sec:RSW} below can be combined to get the desired result at least for all critical doubly periodic planar Ising models; see also Section~\ref{subsub:RSW-Mahfouf} below.}
 \end{remark}

\subsection{Further perspectives} \label{sub:further} Let us start this \emph{informal} discussion by a comment on the convergence results for correlation functions (e.g., energy densities~\cite{hongler-smirnov} or spins~\cite{CHI}) to the CFT scaling limits. Though the existing proofs of such results are heavily based upon the sub-/super-harmonicity of functions~$H_{F}$ (e.g., see \cite[Lemma~3.8]{Smi-I} and \cite[Proposition~3.6]{ChSmi2}), this comparison with \emph{harmonic} functions can be totally avoided and replaced by the results of this paper via an abstract comparison principle for functions~$H_F$ themselves (see Proposition~\ref{prop:HF-comparison}). An appropriate version of the convergence arguments from~\cite{hongler-smirnov,CHI} will appear in~\cite{CHI-mixed}. Note however that these proofs also rely upon an explicit construction of certain full-plane kernels which is not available yet even in the {doubly periodic} setup.

We now move to discussing possible \emph{generalizations} of the results presented in this paper. As already mentioned above, assumptions~\Unif\ and, especially, \Qflat\ do not look conceptual for Theorem~\ref{thm:RSW-selfdual}. A much more reasonable starting point for the uniform crossing estimates on s-embeddings would be
\begin{assumpintro}[\LipKd] There exists a positive constant~$\kappa<1$ such that
\begin{equation}
\label{eq:LipKd}
|\cQ^\delta(v')-\cQ^\delta(v)|\le\kappa\cdot |\cS^\delta(v')-\cS^\delta(v)|\quad \text{if}\quad |\cS^\delta(v')-\cS^\delta(v)|\ge\delta.
\end{equation}
\end{assumpintro}
{It is clear that one \emph{cannot} hope to prove uniform crossing estimates without such a Lipschitzness assumption. E.g., if we start with the standard s-embedding of the critical Ising model on the square grid and replaces the underlying Dirac spinor~$\cX$ by $\cX+t\overline{\cX}$ with $|t|\uparrow 1$ as discussed in Remark~\ref{rem:remintro-R21etc} (see also Fig.~\ref{fig:Z2-emb}), then the new s-embeddings of the \emph{same} model are heavily stretched in one direction as compared to the standard embedding, which is not compatible with RSW-type estimates. The same happens for \mbox{s-embeddings} of a non-critical model: for each~$\kappa<1$, the graph~$(\cS^\delta;\cQ^\delta)$ is not $\kappa$-Lipschitz (even on large scales) in regions of size greater than the characteristic length of the model; see Fig.~\ref{fig:Z2-emb} and Remark~\ref{rem:massive-iso} below.

Also,} note that the assumption~\LipKd\ makes sense without~$\Unif$ and {thus can be used as} a \emph{definition} of the `mesh sizes'~$\delta\to 0$ of {s-embeddings~$\cS^\delta$ in} convergence results; cf.~\cite{CLR1}. (It is also easy to see that the assumption~\Unif\ implies~\LipKd\ with the parameter~$\delta$ in \LipKd\ being a constant multiple of that in \Unif.) This suggests the following research direction:
\begin{enumerate}\renewcommand{\theenumi}{{\bf \Roman{enumi}}}
\item Prove Theorem~\ref{thm:RSW-selfdual} under the assumption~\LipKd\ and, if needed, {a `technical'} assumption on absence of \emph{drastic} degeneracies in s-embeddings similar to the assumption~\ExpFat\ used in~\cite{CLR1}; see Assumption~\ref{assump:ExpFat} below.
\end{enumerate}
{It is worth noting that a substantial progress towards an affirmative answer to this question has been obtained in a recent work~\cite[Section~5]{mahfouf-thesis} of Mahfouf; see Section~\ref{subsub:RSW-Mahfouf} below for more comments.}

\smallskip

The preceding discussion naturally leads to another question:
\begin{enumerate}\renewcommand{\theenumi}{{\bf \Roman{enumi}}}
\setcounter{enumi}{1}
\item Which \emph{spectral} property of the propagation operator~\eqref{eq:3-terms} (or, equivalently, of the related Kac--Ward or the Fisher--Kasteleyn matrices, see~\cite[Section~3]{CCK}) implies the existence of a Dirac spinor~$\cX$ such that the corresponding s-embedding $\cS_\cX$ satisfies the estimate~\eqref{eq:LipKd} for big enough~$\delta$?
\end{enumerate}
Similarly to the periodic setup, this property should be related to the existence of `zero mode' solutions of the equation~\eqref{eq:3-terms}, i.e., to spectral characteristics of the propagator near~$\lambda=0$; we also refer the interested reader to~\cite[Theorem~1.1 and Sections~5.2--5.3]{CHM} for a related study of the so-called {zig-zag layered} Ising model.

\smallskip

From our perspective, the most intriguing and presumably the most important research direction immediately related to the results presented in this article is
\begin{enumerate}\renewcommand{\theenumi}{{\bf \Roman{enumi}}}
\setcounter{enumi}{2}
\item to prove an analogue of Theorem~\ref{thm:FK-conv} and analogues of the results~\cite{hongler-smirnov,CHI,park,CIM-massive} on the convergence of energy density and spin correlations for \mbox{s-embeddings} satisfying the assumption~\LipKd\ (supplemented, if needed, by a mild assumption similar to~\ExpFat).
\end{enumerate}
Though we are still rather far from achieving this goal at the moment, let us mention two relevant observations that immediately follow from our results:

\smallskip

(a) If~$\cQ^\delta\not\to 0$ as~$\delta\to 0$, one should \emph{not} expect that fermionic observables~$F^\delta$ on s-embeddings~$\cS^\delta$ have \emph{holomorphic} scaling limits. In fact, each (subsequential) limit~$\cQ^\delta\circ(\cS^\delta)^{-1}\to\vartheta$ gives rise to a non-trivial conformal structure on the domain~$\Omega$ coming from the surface~$(z,\vartheta(z))_{z\in\Omega}$ viewed in the \emph{Minkowski} space~$\R^{2,1}$. Moreover, the mean curvature of this surface appears as the mass in the free fermion description of the model as~$\delta\to 0$. We refer the reader to a brief discussion in Section~\ref{sub:shol-limits} for more details on this geometric interpretation and leave this research direction for the future study; see also a related project~\cite{CLR1,CLR2} on the bipartite dimer model and its connection to { maximal} surfaces in the Minkowski space~$\R^{2,2}$.

(b) Contrary to fluctuations of height functions in the bipartite dimer model considered in~\cite{CLR1,CLR2}, the critical Ising model correlations require a \emph{non-trivial scaling} to be made before passing to the limit~$\delta\to 0$. On the square grid (and on isoradial {grids}), these scaling factors are $\delta^{-\frac{1}{2}}$ per Kadanoff--Ceva fermion, $\delta^{-1}$ per energy density, and $\delta^{-\frac{1}{8}}$ per spin (e.g., see~\cite{Ch-ECM16,CHI-mixed}). Aiming to generalize such convergence results to s-embeddings, one should replace these factors appropriately. For Kadanoff--Ceva fermions~$X_\varpi(c)$ the relevant local factor~$|\cX(c)|^{-1}$ is clearly visible from~\eqref{eq:Pr=Pr-semb} and one should expect a similar factor~$|\cX(c)|^{-2}$ for energy densities. However, finding appropriate scaling factors for \emph{spins} seems to be more involved.

\subsection{Organization of the paper} This paper does not formally rely upon the previous work on the convergence of correlation functions in the critical Z-invariant Ising model on isoradial {grids.}
Nevertheless, we keep the notation as close to~\cite{ChSmi2} as possible so that readers familiar with the isoradial context can benefit from comparison of our methods with those from~\cite{Smi-I,ChSmi2}.

We collect preliminaries on fermionic observables in Section~\ref{sec:definitions}. Note that Sections~\ref{sub:HF-def}--\ref{sub:shol-limits} contain {a} new (as compared, e.g., to~\cite{ChSmi2}) approach to the a priori regularity theory for s-holomorphic functions, which relies upon the ideas developed in~\cite{CLR1} in a more general context of the bipartite dimer model. In Section~\ref{sec:operators} we discuss discrete differential operators on s-embeddings. Though they generalize the usual Cauchy--Riemann and Laplacian operators on isoradial {grids,} such a generalization is not at all straightforward. We give a proof of Theorem~\ref{thm:FK-conv} in Section~\ref{sec:convergence}, the whole strategy of the proof is new even in the square grid case. Section~\ref{sec:RSW} is devoted to the proof of Theorem~\ref{thm:RSW-selfdual} and does not rely upon Sections~\ref{sec:operators} and~\ref{sec:convergence}, so the readers primarily interested in RSW-type estimates on s-embeddings can go to Section~\ref{sec:RSW} right after Section~\ref{sec:definitions}. {The appendix} contains the proof of a useful estimate `in {the} continuum' (Lemma~\ref{lem:cont-estimate}), upon which our proof of Theorem~\ref{thm:FK-conv} is based.

\medskip

\noindent {\bf Acknowledgements.} The author is grateful to Niklas Affolter, {Arseniy Akopyan,} Scott Armstrong, {Mikhail Basok,} Dmitry Belyaev, Ilia Binder, Olivier Biquard, Hugo Duminil-Copin, Christophe Garban, Cl\'ement Hongler, Konstantin Izyurov, Beno\^it Laslier, Alexander Logunov, R\'emy Mahfouf, Eugenia Malinnikova, Jean-Christophe Mourrat, S.~C. Park, Olga Paris-Romaskevich, {Sanjay Ramassamy,} Marianna Russkikh, Stanislav Smirnov and Yijun Wan for {helpful discussions and remarks. We also thank R\'emy Mahfouf and the referees for carefully reading the manuscript and providing a useful feedback.}

This research was partly supported by the ANR-18-CE40-0033 project DIMERS. Last but not least, the author would like to gratefully acknowledge the support provided by the ENS--MHI chair funded by the MHI, which made this work possible.

\section{Definitions and preliminaries}\label{sec:definitions}
\setcounter{equation}{0}

\newcommand\vcirc[1]{v^\circ_1,\ldots,v^\circ_{#1}}
\newcommand\svcirc[1]{\sigma_{v^\circ_1}\ldots \sigma_{v^\circ_{#1}}}
\newcommand\vbullet[1]{v^\bullet_1,\ldots,v^\bullet_{#1}}
\newcommand\muvbullet[1]{\mu_{v^\bullet_1}\ldots\mu_{v^\bullet_{#1}}}

\subsection{Notation and Kadanoff--Ceva fermions}\label{sub:notation}
Let~$G$ be a planar graph (we do allow multiple edges and vertices of degree two in~$G$ but do \emph{not} allow loops and vertices of degree one) with the combinatorics of the sphere or that of the plane, considered \emph{up to homeomorphisms} preserving the cyclic order of edges at each vertex.  In the former case we always assume that one of the faces of~$G$ is fixed under such homeomorphisms and call this face the \emph{outer} face of~$G$; because of this we also call this setup the \emph{disc} case.

Throughout the paper we use the following notation:
\begin{itemize}
\item $G=G^\bullet$ is the original graph, its vertices are typically denoted by~$v^\bullet\in G^\bullet$;


\item $G^\circ$ is the graph dual to~$G$, its vertices are typically denoted by~$v^\circ\in G^\circ$ (in the disc case, we include to~$G^\circ$ the vertex corresponding to the outer face);


\item $\Lambda(G)$ {is the graph whose vertex set is the union of vertices of $G^\circ$ and $G^\bullet$} with the natural incidence relation: the faces of~$\Lambda(G)$ have degree four and are in a bijective correspondence with edges of~$G$;


\item $\Dm(G)$ is the graph dual to~$\Lambda(G)$, we typically denote its vertices \mbox{$z\in\Dm(G)$} and often call~$z$ a \emph{quad}, this is what generalizes rhombi in the isoradial~setup; {sometimes we also identify $z\in\Dm(G)$ with the corresponding edge of $G$;}


\item $\Upsilon(G)$ is the medial graph of~$\Lambda(G)$, its vertices are in a bijective correspondence with edges~$(v^\bullet v^\circ)$ of~$\Lambda(G)$, we typically denote these vertices by~$c\in\Upsilon(G)$ and sometimes call them \emph{corners} (of faces) of~$G$.
\end{itemize}
We often consider \emph{double covers} of~$\Upsilon(G)$, see~\cite[Fig.~27]{Mercat-CMP}, \cite[Fig.~6]{ChSmi2} or \cite[Fig.~4]{CIM-massive}:
\begin{itemize}
\item $\Upsilon^\times(G)$ stands for the double cover of~$\Upsilon(G)$ that branches over \emph{each} face of~$\Upsilon(G)$, i.e., over each~$v^\bullet\in G^\bullet$, each $v^\circ\in G^\bullet$, and each~$z\in\Dm(G)$ (if $G$ is finite, this notion makes sense since~$\#(G^\bullet)+\#(G^\circ)+\#(\Dm(G))$ is even);


\item given a set~$\varpi=\{\vbullet{m},\vcirc{n}\}\subset \Lambda(G)$ with $n,m$ even, we denote by~$\Upsilon^\times_\varpi(G)$ the double cover of~$\Upsilon(G)$ that branches over all its faces \emph{except} those from~$\varpi$,
    and by~$\Upsilon_\varpi(G)$ the double cover that branches \emph{only} over~$\varpi$.


\item We say that a function defined on one of these double covers is a \emph{spinor} if its values at the two lifts of the same~$c\in\Upsilon(G)$ differ by the sign.
\end{itemize}

Recall that we consider the Ising model~\eqref{eq:intro-Zcirc} defined on faces of~$G$ (including the outer face in the disc case), i.e., on vertices of~$G^\circ$. The \emph{domain walls} representation (aka the low-temperature expansion) assigns to a spin configuration~$\sigma:G^\circ\to\{\pm 1\}$ a subset~$C$ of edges of~$G$ separating different spins; clearly this is a $2$-to-$1$ mapping of spin configurations onto the set~$\cE(G)$ of even subgraphs of~$G$.

Given~$\vcirc{n}\in G^\circ$ with~$n$ even, let a subgraph $\gamma^\circ=\gamma_{[\vcirc{n}]}\subset G^\circ$ {have} even degree at all vertices of~$G^\circ$ except~$\vcirc{n}$ and odd degree at these vertices (the reader can think about a collection of paths on~$G^\circ$ linking~$\vcirc{n}$ pairwise). Denote
\[
x_{[\vcirc{n}]}(e)\ :=\ (-1)^{e\cdot\gamma_{[\vcirc{n}]}}\,x(e),\quad e\in E(G),
\]
where~$e\cdot\gamma=0$ if~$e$ does not cross~$\gamma$ and~$e\cdot\gamma=1$ otherwise. It is easy to see that
\begin{equation}
\label{eq:Esigma}
\textstyle \mathbb E\big[\svcirc{n}\big]\ =\ {x_{[\vcirc{n}]}(\cE(G))}\big/{x(\cE(G))},
\end{equation}
where~$x(\cE(G)):=\sum_{{C}\in\cE(G)}x(C)$, $x(C):=\prod_{e\in C}x(e)$, and similarly for~$x_{[\vcirc{n}]}$.

Now let~$m$ be even,~$\vbullet{m}\in G^\bullet$ and a subgraph~$\gamma^\bullet=\gamma^{[\vbullet{m}]}\subset G^\bullet$ has even degree at all vertices of~$G^\bullet$ except~$\vbullet{m}$ and odd degree at these vertices. Following Kadanoff and Ceva~\cite{kadanoff-ceva} let us abbreviate by~$\muvbullet{m}$ the effect of the change of the signs of the interaction constants~$J_e\mapsto -J_e$ on edges~$e\in\gamma^\bullet$ (which is equivalent to the change of parameters~$x(e)\mapsto 1/x(e)$ for~$e\in\gamma^\bullet$). More accurately, if we consider the product
\[
\textstyle \muvbullet{m}\ :=\ \exp\big[-2\beta\sum_{e\in\gamma^{[\vbullet{m}]}}J_e\sigma_{v^\circ_-(e)}\sigma_{v^\circ_+(e)}\,\big]\,,
\]
then, still using domain walls representations of $\sigma$, it is not hard to see that
\begin{equation}
\label{eq:Emu}
\textstyle \mathbb E\big[\muvbullet{m}\big]\ =\ x(\cE^{[\vbullet{m}]}(G))\big/{x(\cE(G))},
\end{equation}
where~$\cE^{[\vbullet{m}]}$ denotes the set of all subgraphs of~$G$ that have even degrees at all vertices except~$\vbullet{m}$ and odd degrees at these vertices. In particular, this expectation does not depend on the choice of $\gamma^\bullet$, similarly to the fact that the expectation~\eqref{eq:Esigma} does not depend on the choice of~$\gamma^\circ$. Further, a straightforward generalization of~\eqref{eq:Esigma} and~\eqref{eq:Emu} reads as
\begin{equation}
\label{eq:Emusigma}
\textstyle \mathbb E\big[\muvbullet{m}\svcirc{n}\big]\ =\ x_{[\vcirc{n}]}(\cE^{[\vbullet{m}]}(G))\big/{x(\cE(G))},
\end{equation}
where~$\muvbullet{m}$ should be understood as above. However, this expression does depend on the choice of the paths~$\gamma^\circ\subset G^\circ$ and~$\gamma^\bullet\subset G^\bullet$, more precisely its \emph{sign} depends on the parity of the number of intersections of~$\gamma^\circ$ and~$\gamma^\bullet$.

\smallskip

In general, there is no way to make the choice of the sign in~\eqref{eq:Emusigma} canonical just staying on the Cartesian product $(G^\bullet)^{\times m}\!\times (G^\circ)^{\times n}$. However, one can pass to the natural double cover of this Cartesian product (which has the same branching structure as the spinor~$[\,\prod_{p=1}^m\prod_{q=1}^n(\cS(v^\bullet_p)-\cS(v^\circ_q))\,]^{1/2}$, where~$\cS:\Lambda(G)\to\C$ is an \emph{arbitrarily} chosen embedding). By doing this, one can view the expectations~\eqref{eq:Emusigma} as spinors on the double cover of~$(G^\bullet)^{\times m}\!\times (G^\circ)^{\times n}$ described above. We refer an interested reader to~\cite{CHI-mixed} for a more extensive discussion of this definition.

\smallskip

Though we do not use this fact below, let us mention that \emph{disorders} $\muvbullet{m}$ are dual objects to spin variables $\svcirc{n}$ under the Kramers--Wannier duality: indeed, the expression~\eqref{eq:Emu} is nothing but the \emph{high-temperature expansion} (e.g., see~\cite[Section~7.5]{duminil-parafermions}) of the spin correlation in the dual Ising model on~$G^\bullet$ with interaction parameters~$x^*(e^*):=(1-x(e))/(1+x(e))$. In particular, the roles of the graphs~$G^\bullet$ and~$G^\circ$ are equivalent and one can also see that
\[
x_{[\vcirc{n}]}(\cE^{[\vbullet{m}]}(G))\big/{x(\cE(G))}\ =\ x^*_{[\vbullet{m}]}(\cE^{[\vcirc{n}]}(G^\circ))\big/{x^*(\cE(G^\circ))}
\]
by extending the Kramers--Wannier duality to the mixed correlations~\eqref{eq:Emusigma}.

Let~$v_0^\bullet,v_0^\circ,v_1^\bullet,v_1^\circ\in\Lambda(G)$ be vertices of a quad~$z\in\Dm(G)$ listed counterclockwise as in { Fig.~\ref{fig:notation}}. Then, the propagation equation for the Kadanoff--Ceva fermionic variables~$\chi_c$ defined by~\eqref{eq:KC-chi-def} can be seen as a simple corollary of the identity
\[
\mu_{v_0^\bullet}\mu_{v_1^\bullet}\ =\ \exp[-2\beta J_e\sigma_{v_0^\circ}\sigma_{v_1^\circ}]\,=\,\tfrac{1}{2}\cdot(x(e)+(x(e))^{-1})-\sigma_{v_0^\circ}\sigma_{v_1^\circ}\cdot \tfrac{1}{2}(x(e)-(x(e))^{-1}).
\]
Under the parametrization~\eqref{eq:x=tan-theta} this identity can be written as
\[
\mu_{v_0^\bullet}\mu_{v_1^\bullet}\cdot\sin\theta_z\ =\ 1-\sigma_{v_0^\circ}\sigma_{v_1^\circ}\cdot \cos\theta_z,
\]
which gives~\eqref{eq:3-terms} after the formal multiplication by~$\mu_{v^\bullet_0}\sigma_{v^\circ_0}$. (To get a rigorous proof for the expectations~$\E[\muvbullet{m}\svcirc{n}]$ with~$v^\bullet_m=v^\bullet(c)=v^\bullet_0(z)$ and~$v^\circ_n=v^\circ(c)=v^\circ_0(z)$, one should consider the effect of adding/removing the edge~$(v_0^\bullet(z) v_1^\bullet(z))$ to the corresponding disorder paths~$\gamma^\bullet$.)

\subsection{Definition of s-embeddings}\label{sub:semb-definition}

The following definition was proposed in~\cite{Ch-ICM18}.

\begin{definition}\label{def:cS-def}
Given a weighted planar graph $(G,x)$ with the combinatorics of the plane and a complex-valued solution $\cX:\Upsilon^\times(G)\to\C$ of the equation~\eqref{eq:3-terms}, we call $\cS=\cS_\cX:\! \Lambda(G)\to\C$ an s-embedding of $(G,x)$ defined by $\cX$ if
\begin{equation}
\label{eq:cS-def}
\cS(v^\bullet(c))-\cS(v^\circ(c))=(\cX(c))^2
\end{equation}
for each $c\in\Upsilon^\times(G)$. Given~$z\in\Dm(G)$, we denote by~$\cS^\dm(z)\subset\C$ the quadrilateral with vertices~$\cS(v_0^\bullet(z))$, $\cS(v_0^\circ(z))$, $\cS(v_1^\bullet(z))$, $\cS(v_1^\circ(z))$.

We call an s-embedding $\cS$ proper if the quads~$\cS^\dm(z)$ do not overlap with each other, and
non-degenerate if none of the quads~$\cS^\dm(z)$ degenerates to a segment.
\end{definition}

\begin{remark}
The quadrilaterals~$\cS^\dm(z)$ cannot be self-intersecting; the simplest way to prove this fact is to use the consistency of the definition~\eqref{eq:cQ-def} of the associated function {$\cQ_\cX$: if opposite sides of~$\cS^\dm(z)$ intersected each other, the sum of their lengths would be strictly bigger than the sum of the lengths of two other sides. Note, however,} that we do \emph{not} require the convexity of~$\cS^\dm(z)$; {see Fig.~\ref{fig:Semb}.}
\end{remark}

It is also useful to extend the definition of~$\cS$ from~$\Lambda(G)$ to~$\Dm(G)$ by setting
\begin{equation}\label{eq:cS(z)-def}
\begin{array}{l}
\cS(v_p^\bullet(z))-\cS(z):=\cX(c_{p0})\cX(c_{p1})\cos\theta_z,\\[2pt]
\cS(v_q^\circ(z))-\cS(z):=-\cX(c_{0q})\cX(c_{1q})\sin\theta_z,
\end{array}
\end{equation}
where~$c_{p0}$ and~$c_{p1}$ (respectively,~$c_{0q}$ and~$c_{1q}$) are assumed to be neighbors on the double cover~$\Upsilon^\times(G)$.
The consistency of the definitions~\eqref{eq:cS-def},~\eqref{eq:cS(z)-def} around~$z\in\Dm(G)$ can be easily deduced from the propagation equation~\eqref{eq:3-terms}.

\begin{definition}\label{def:cQ-def}
Given~$\cS=\cS_\cX$, we construct a function \mbox{$\cQ=\cQ_\cX:\Lambda(G)\to\R$}, defined up to a global additive constant by requiring that
\begin{equation}
\label{eq:cQ-def}
\cQ(v^\bullet(c))-\cQ(v^\circ(c))\ :=\ |\cX(c)|^2\,=\,|\cS(v^\bullet(c))-\cS(v^\circ(c))|\,.
\end{equation}
\end{definition}

\begin{remark} In~\cite[Section~6]{Ch-ICM18}, the notation~$L$ was used instead of~$\cQ$. We change the notation due to the fact (e.g., see~\cite[Section~7]{KLRR}) that~$\cS$ can be viewed as a \mbox{t-embedding}~$\cT$ of the corresponding bipartite dimer model and~$\cQ$ coincides with the origami map~$\cO$ of~$\cT$ under this correspondence; see Section~\ref{sub:dimers} for more comments.
\end{remark}

{Again, the consistency of the definition~\eqref{eq:cQ-def} (i.e., the fact that the sum of increments of~$\cQ$ around~$z$ vanishes) easily follows from the equation~\eqref{eq:3-terms}.} Geometrically, this means that~$\cS^\dm(z)$ is a tangential quad: the (lines containing) sides of~$\cS^\dm(z)$ are tangent to a circle. It is easy to see from~\eqref{eq:cS-def},~\eqref{eq:cS(z)-def} that~$\cS(z)$ is the center of this circle. Also, a simple computation shows that
\begin{align}
r_z\ :=&\ \Im[\,\cX(c_{01}(z))\overline{\cX(c_{10}(z))}\,]\cos\theta_z\sin\theta_z =\Im[\,\cX(c_{11}(z))\overline{\cX(c_{00}(z))}\,]\cos\theta_z\sin\theta_z \notag \\
=&\ \Im[\,\cX(c_{01}(z))\overline{\cX(c_{00}(z))}\,]\cos\theta_z =\Im[\,\cX(c_{00}(z))\overline{\cX(c_{10}(z))}\,]\sin\theta_z
\label{eq:rz=}
\end{align}
is the radius of the circle inscribed into~$\cS^\dm(z)$, where we assume that~$c_{11}$\,---\,$c_{01}$\,---\,$c_{00}$\,---\,$c_{10}$ are chosen to be neighbors on the double cover~$\Upsilon^\times(G)$; see {Fig.~\ref{fig:notation}}.

\smallskip

A natural question is how to recover the Ising model weights~$\theta_z$ from the geometric characteristics of tangential quads~$\cS^\dm$. { The answer is given by the formula}
\begin{equation}
\label{eq:theta-from-S}
\tan\theta_z\ =\ \biggl(\frac{\sin\phi_{v_0^\bullet z}\sin\phi_{v_1^\bullet z}}{\sin\phi_{v_0^\circ z}\sin\phi_{v_1^\circ z}}\biggr)^{\!1/2},
\end{equation}
where~$\phi_{vz}$ denotes the half-angle of the quad~$\cS^\dm(z)$ at the vertex~$\cS(v)$; we do not rely upon {\eqref{eq:theta-from-S} below. Particular cases of the setup covered by this paper include:}
\begin{itemize}
\item {\bf Critical Z-invariant Ising model on isoradial {grids.}}  Let \mbox{$\Lambda\!=\!\Gamma^\bullet\!\cup\Gamma^\circ$} be a \emph{rhombic lattice} and denote~$\cX(c):=(v^\bullet(c)-v^\circ(c))^{1/2}$. It is easy to check that~$\cX$ satisfies the propagation equation~\eqref{eq:3-terms} for the critical \mbox{Z-invariant} weights on the isoradial {grid.} We refer an interested reader to~\cite{Mercat-CMP} for
    a discussion of special solutions of the propagation equation~\eqref{eq:3-terms} on discrete Riemann surfaces, and to a series of recent papers~\cite{BdTR-Ising-dimers,deTiliere-18,park-iso,CIM-massive} for a discussion of the near-critical model on isoradial grids.
\smallskip
\item {\bf Critical Ising model on circle patterns} introduced by Lis in~\cite{Lis} as a generalization of the isoradial context. In this setup the quads~$\cS^\dm(z)$ are \emph{kites}, which allows one to view~$\cS$ as a circle pattern. Using the Kac--Ward matrices techniques, Lis managed to prove that the Ising model on centers of circles with interaction constants given by~\eqref{eq:theta-from-S} is critical. We also refer the interested reader to~\cite[Section~3]{CCK} for a discussion of the link between Kac--Ward matrices and the propagation equation~\eqref{eq:3-terms}. In this setup, the function~$\cQ_\cX$ defined by~\eqref{eq:cQ-def} vanishes on~$G^\bullet$ and equals the (minus) radius of the corresponding circle on~$G^\circ$.
\smallskip
\item {\bf Critical Ising model on {doubly periodic} graphs.} For a {doubly periodic} planar weighted graph~$(G,x)$ carrying the Ising model, the first natural question is to find an algebraic condition on the weights~$x_e$ describing the critical point. This was done in~\cite{Li} and~\cite{cimasoni-duminil}, we refer the reader to the latter paper for more details. Clearly, this condition does \emph{not} fix an embedding of the graph into~$\C$. The question of finding special embeddings that would allow to prove the conformal invariance of a given critical {doubly periodic} Ising model was posed by Duminil-Copin and Smirnov~\cite[Question~8.5]{duminil-smirnov} shortly after the papers~\cite{ChSmi1,ChSmi2} appeared but remained open until recently. The answer is provided by the following lemma (and Theorem~\ref{thm:FK-conv}).
\end{itemize}
\begin{lemma}\label{lem:double-periodic} Let~$(G,x)$ be a {doubly periodic} graph carrying a critical Ising model. There exists a unique (up to a multiplicative constant and the conjugation) {doubly periodic} solution~$\cX$ of the equation~\eqref{eq:3-terms} such that the corresponding function~$\cQ_\cX$ is also {doubly periodic}. Moreover, the s-embedding~$\cS_\cX$ constructed from~$\cX$ is proper.
\end{lemma}
\begin{proof}
This lemma appeared in~\cite[Lemma~13]{KLRR} in the bipartite dimer model context. For the completeness of the exposition we sketch its proof below staying in the Ising model terminology.

The fact that the propagation equation~\eqref{eq:3-terms} has exactly two independent real-valued periodic solutions~$\cX_1,\cX_2$ at criticality follows from~\cite{Li,cimasoni-duminil} and~\cite[Section~3]{CCK}. Let~$\cX^0:=\cX_1+i\cX_2$ and note that each linear combination~$\cX=\cX^0+k\overline{\cX^0}$, $|k|\ne 1$, of these solutions gives rise to a \emph{non-degenerate} proper \mbox{s-embedding}~$\cS_\cX$; this can be deduced along the lines of~\cite{kenyon-sheffield}, we refer the reader to~\cite[Lemma~14]{KLRR} for details. Loosely speaking, the key idea is to consider a homotopy~$(\cS+t\cQ)_{t\in [0,1]}$ and to view~$\cS+\cQ$ as a T-graph (see Section~\ref{sub:dimers} below) to which~\cite[Theorem~4.6]{kenyon-sheffield} applies.
Still, one needs to prove that there exists a {unique $k_0\in\mathbb{D}$} in the unit disc $\mathbb{D}=\{k\in\C:|k|<1\}$ leading to a \emph{{ doubly periodic}} function~$\cQ_\cX$.

Let~$A_\mathcal{S},B_\mathcal{S}\in\C$ and~$A_\mathcal{Q},B_\mathcal{Q}\in\R$ be the increments of the functions~$\cS_{\cX^0}$ and~$\cQ_{\cX^0}$, respectively, along the horizontal and the vertical periods of the torus. The fact that the \mbox{s-embedding} $\cS_\cX$, $\cX=\cX^0+k\overline{\cX^0}$, is non-degenerate implies that
\[
(nA_\mathcal{S}+mB_\mathcal{S})+2(nA_\mathcal{Q}+mB_\mathcal{Q})\cdot k+(n\overline{A}_\mathcal{S}+m\overline{B}_\mathcal{S})\cdot k^2\ \ne\ 0\ \ \text{for all}\ \ k\in\mathbb D
\]
and for all~$(n,m)\in\mathbb Z^2\smallsetminus\{(0,0)\}$; otherwise $\cS_\cX$ would have a zero increment along the corresponding period~$(n,m)$ of the torus. Therefore, $|nA_\mathcal{Q}+mB_\mathcal{Q}|\le|nA_\mathcal{S}+mB_\mathcal{S}|$ for all $n,m\in\mathbb Z$. By continuity, this implies that
\begin{equation}
\label{eq:x1-exist-per}
|A_\mathcal{Q}+tB_\mathcal{Q}|\ \le\ |A_\mathcal{S}+tB_\mathcal{S}|\ \ \text{for all}\ \ t\in\R.
\end{equation}

Further, the {periodicity} of~$\cQ_\cX$ is equivalent to the conditions
\[
(1\!+\!|k|^2)\cdot A_\mathcal{Q}\,+\,2\Re[kA_\mathcal{S}]\ =\ 0\ =\ (1\!+\!|k|^2)\cdot B_\mathcal{Q}\,+\,2\Re[kB_\mathcal{S}].
\]
Each of these two equations describes a circle (or a line if, say, $A_\mathcal{Q}=0$) in the complex plane, symmetric with respect to the unit circle~$\mathbb T=\partial\mathbb D$. Thus, we need to prove that these two circles intersect each other. A straightforward computation shows that this intersection condition holds if
\begin{equation}
\label{eq:x2-exist-per}
|A_\mathcal{Q}B_\mathcal{S}-B_\mathcal{Q}A_\mathcal{S}|^2\ \le\ (\Im A_\mathcal{S}\overline{B}_\mathcal{S})^2.
\end{equation}
{In fact, it is not hard to check that} \eqref{eq:x2-exist-per} is \emph{equivalent} to~\eqref{eq:x1-exist-per} {(which is not very surprising in view of Remark~\ref{rem:remintro-R21etc}).} It remains to rule out the case when~\eqref{eq:x2-exist-per} becomes an equality and the two circles touch each other at \mbox{$\alpha\in\mathbb T$}. In this case the function~$\cQ_{\cX^0}+\Re[\alpha\cS_{\cX^0}]=\Re[\alpha(\cS_{\cX^0}+\overline{\alpha}\cQ_{\cX^0})]$ should be {doubly periodic}, a contradiction with the non-degeneracy of the T-graph~$\cS_{\cX^0}+\overline{\alpha}\cQ_{\cX^0}$.
\end{proof}

\subsection{Link with the bipartite dimer model, s-holomorphicity and S-graphs}\label{sub:dimers}
It is well known (e.g., see~\cite{Dubedat-bos},~\cite[Section~2.2.3]{deTiliere-18} and references therein) that to each planar graph~$(G,x)$ carrying the Ising model one can {naturally} associate a dimer model on a bipartite planar graph~$\Upsilon^\circ(G)\cup\Upsilon^\bullet(G)$, with dimer weights~$1,\cos\theta_z,\sin\theta_z$ as shown on {Fig.~\ref{fig:notation}}. (Note that only a very {particular} family of dimer models {appears} in this way.) In what follows we assume that we work in the bulk of~$G$; in other words, we do not discuss nuances of this correspondence near the boundary of~$G$ if it has the combinatorics of the disc.

In the bipartite model context, a notion of \emph{Coulomb gauges} was recently proposed in~\cite{KLRR}; exactly the same notion appears in~\cite{CLR1} under the name~\emph{t-embeddings}. Under the combinatorial correspondence mentioned above, s-embeddings can be viewed as a particular case of t-embeddings: see~\cite[Section~7]{KLRR} and~\cite[{Section~8.2}]{CLR1} for more details. We now recall basic facts about this correspondence in what concerns s/t-embeddings of the relevant graphs, the notation used below matches that of~\cite{CLR1}.

Let~$\cS$ be a proper non-degenerate embedding of the graph~$(G,x)$. Note that \mbox{$\Lambda(G)\cup\Dm(G)=(\Upsilon^\circ(G)\cup\Upsilon^\bullet(G))^*$}. The following holds:
\begin{itemize}
\item $\cS:\Lambda(G)\cup\Dm(G)\to\C$ is a t-embedding of the graph~$(\Upsilon^\circ(G)\cup\Upsilon^\bullet(G))^*$;

\smallskip

\item the function~$\eta_{c^\bullet}=\eta_{c^\circ}:= \overline{\varsigma}\eta_c$ (where~$\eta_c$ is given by~\eqref{eq:def-eta}) is an \emph{origami square root function} (see~\cite[Definition~2.4]{CLR1}) of the t-embedding~$\cT=\cS$;

\smallskip

\item $\cQ:\Lambda(G)\to\R$ is the restriction of the \emph{origami map}~\mbox{$\cO:\Lambda(G)\cup\Dm(G)\to\C$} associated to the t-embedding~$\cT=\cS$. Let us also set $\cQ(z):=\cO(z)$ for~\mbox{$z\in\Dm(G)$}; note that~$\Im\cQ(z)>0$.

\item {If one extends the mapping~$\cQ\circ\cS^{-1}$ from~$\Lambda(G)\cup\Dm(G)$ to the complex plane in a piecewise affine way, then $\cQ\circ\cS^{-1}$ is an isometry on each of the four triangles forming a quad $\cS^\dm(z)$. (This isometry preserves the orientation of white triangles and reverses the orientation of black ones.) In particular,
    \[
    \qquad |\cQ(w_1)-\cQ(w_2)|\le|\cS(w_1)-\cS(w_2)|\ \ \text{for~all}\ w_1,w_2
    \]
    since thus obtained (complex-valued) extension~$\cQ\circ\cS^{-1}$ is $1$-Lipschitz.}

\end{itemize}

In what follows, one should set~$\varsigma:=e^{i\frac{\pi}{4}}$ in~\eqref{eq:def-eta} to keep the notation consistent, e.g., with~\cite[Definition~3.1]{ChSmi2}, where the term \emph{s-holomorphicity} was coined. Choosing another value~$\varsigma\in\mathbb{T}$ leads to trivial modifications.
\begin{definition}\label{def:s-hol}
A complex-valued function~$F$ defined on (a subset of)~$\Dm(G)$ is called s-holomorphic if
\begin{equation}
\label{eq:s-hol}
\Pr{F(z)}{\eta_c\R}\ =\ \Pr{F(z')}{\eta_c\R}
\end{equation}
for each pair of quads $z,z'\in\Dm(G)$ adjacent to the same edge~$(v^\circ(c)v^\bullet(c))$ of~$\cS$.
\end{definition}

Under the correspondence described above, s-holomorphicity is a particular case of \emph{t-holomorphicity}; the latter notion was introduced and studied in~\cite{CLR1}. Below we use the notation introduced in {Fig.~\ref{fig:notation}; see also Fig.~\ref{fig:Sdef} below.}
\mbox{T-holomorphic} functions are defined on either white \mbox{$F_\frw^\circ:\Upsilon^\circ(G)\to\C$} or black $F_\frb^\bullet:\Upsilon^\bullet(G)\to\C$ faces of a t-embedding~$\cT$ and satisfy conditions similar to~\eqref{eq:s-hol}, with~$\eta_c\R$ replaced by~$\eta_{c^\bullet}\R$ or~$\eta_{c^\circ}\R$, respectively; see~\cite[Section~3]{CLR1} for precise definitions. Functions~$F_\frw$ are called t-white-holomorphic whilst~$F_\frb$ are called \mbox{t-black-holomorphic}. {Note that~\cite{CLR1} also uses the notation~$F_\frw^\bullet$ and~$F_\frb^\circ$ for the corresponding projections of the `true' complex values~$F_\frw^\circ$ and~$F_\frb^\bullet$ of these functions onto the lines~$\eta_{c^\bullet}\R$ and~$\eta_{c^\circ}\R$, respectively.} It is easy to see that
\begin{itemize}
\item given an s-holomorphic functions~$F$, the function~$F_\frw$ (resp., $F_\frb$) defined as
\[ F_\frw^\circ(c^\circ_{p,1-p}(z))\,:=\,\overline{\varsigma}F(z),\quad F_\frb^\bullet(c^\bullet_{p,p}(z))\,:=\,\overline{\varsigma}F(z),\quad p=0,1,
\]
is t-white- (resp., t-black-)holomorphic; 

\smallskip

\item vice versa, each t-white-(resp., t-black-)holomorphic function $F_\frw$ (resp.,~$F_\frb$) can be obtained from an s-holomorphic function~$F=\varsigma F_\frw^\circ=\varsigma F_\frb^\bullet$.
\end{itemize}

In the isoradial context, it is well known (e.g., see~\cite[Lemma~3.4]{ChSmi2}) that \emph{real-valued} solutions to the propagation equation~\eqref{eq:3-terms} can be viewed as s-holomorphic functions and vice versa, the correspondence is given by~$X(c):=\Re[\overline{\eta}_cF(z)]$ for corners~$c$ adjacent to a quad~$z$. (Note that~$X(c)$ is defined on the double cover~$\Upsilon^\times(G)$ and not just on~$\Upsilon(G)$ since~$\eta_c$ is a spinor on~$\Upsilon^\times(G)$.)
The following proposition generalizes this link to the setup of s-embeddings; see also~\cite[{Section~8.2}]{CLR1}.

\begin{proposition}\label{prop:shol=3term} Let~$\cS=\cS_\cX$ be a proper s-embedding and $F$ be an s-holomorphic function defined on (a subset of)~$\Dm(G)$. Then, the spinor~$X$ defined on corners \mbox{$c\in\Upsilon^\times(G)$} adjacent to~$z\in\Dm(G)$ by the formula
\begin{align}
X(c)\ &=\ |\cS(v^\bullet(c))-\cS(v^\circ(c))|^{\frac{1}{2}}\cdot\Re[\overline{\eta}_c F(z)] \notag\\
 &=\ \Re[\overline{\varsigma}\cX(c)\cdot F(z)]\ =\ \overline{\varsigma}\cX(c)\cdot\Pr{F(z)}{\eta_c\R}
\label{eq:X-from-F}
\end{align}
satisfies the propagation equation~\eqref{eq:3-terms}.

Vice versa, if~$X:\Upsilon^\times(G)\to\R$ satisfies the propagation equation~\eqref{eq:3-terms}, then there exists {an} s-holomorphic function~$F$ such that~\eqref{eq:X-from-F} is fulfilled.
\end{proposition}
\begin{proof}
Following the preceding discussion, let us interpret~$F$ as, say, a t-black holomorphic function~${F_\frb^\bullet}=\overline{\varsigma}F$. By the definition of t-holomorphicity, we have
\[
\overline{\varsigma}\Pr{F}{\eta_c\R}\ =\ \Pr{F_\frb^\bullet}{\eta_{c^\circ}\R}\ {=:}\ F_\frb^\circ(c^\circ)
\]
and, using the notation of {Fig.~\ref{fig:notation}},
\[
F_\frb^\circ(c_{00}^\circ)(\cS(v_0^\circ)-\cS(v_0^\bullet))
+F_\frb^\circ(c_{10}^\circ)(\cS(z)-\cS(v_0^\circ))
+F_\frb^\circ(c_{01}^\circ)(\cS(v_0^\bullet)-\cS(z))=0.
\]
Due to the formulas~\eqref{eq:cS-def},~\eqref{eq:cS(z)-def}, this identity is equivalent to
\[
-F_\frb^\circ(c_{00}^\circ)\cX(c_{00})
+F_\frb^\circ(c_{10}^\circ)\cX(c_{10})\sin\theta_z
+F_\frb^\circ(c_{01}^\circ)\cX(c_{01})\cos\theta_z=0,
\]
i.e., to the propagation equation~\eqref{eq:3-terms} for the quantities~$X(c)$ defined by~\eqref{eq:X-from-F}.

Vice versa, let~$X$ satisfies~\eqref{eq:3-terms}. Applying the same arguments in the reverse order one sees that ratios~$X(c)/\cX(c)$ can be viewed as the values~$\cX_\frb^\circ=\Pr{\cX^\bullet_{\frb}}{\eta_{c^\circ}\R}$ of a t-black-holomorphic function and we can set~$F:=\varsigma F^\bullet_{\frb}$ using the correspondence between the notions of t-holomorphicity and s-holomorphicity discussed above.
\end{proof}
\begin{corollary}
\label{cor:F-from-X} Let~$F$ and~$X$ be as in Proposition~\ref{prop:shol=3term}. Then,
\[
F(z)\ =\ -i\varsigma\cdot\frac{\overline{\cX(c_{01}(z))}\,X(c_{10}(z))-\overline{\cX(c_{10}(z))}\,X(c_{01}(z))} {\Im[\,\overline{\cX(c_{01}(z))}\,\cX(c_{10}(z))\,]}
\]
and similarly for other pairs of~$c_{pq}(z)\in\Upsilon^\times(G)$.
\end{corollary}
\begin{proof} This is a simple computation: if~$u_1$, $u_2$ are projections of~$w\in\C$ on two lines passing through the origin, then~$w=i(|u_2|^2u_1-|u_1|^2u_2)/\Im(u_2\overline{u}_1)$.
\end{proof}

\begin{remark} Let us emphasize that the propagation equation~\eqref{eq:3-terms} does \emph{not} depend on a particular choice of the embedding of~$(G,x)$ into the complex plane while the \mbox{s-holomorphicity} condition~\eqref{eq:s-hol} heavily relies upon this choice. In other words, each Dirac spinor~$\cX$ gives rise to an equivalent \emph{interpretation} of~\eqref{eq:3-terms} in the geometric terms related to the corresponding s-embedding~$\cS_\cX$.
\end{remark}

\begin{figure}
\hskip 0.01\textwidth \begin{minipage}{0.48\textwidth} $\phantom{x}$

\includegraphics[clip, trim=4.7cm 6.2cm 9.6cm 14.8cm, width=\textwidth]{Sdef_tikz.pdf}

\vskip 12pt

\noindent {\textsc{(A)}} The image of a quad $z\!=\!(v_0^\bullet v_0^\circ v_1^\bullet v_1^\circ)$ in an s-embedding \mbox{$\cS:\Lambda(G)\!\cup\!\Dm(G)\!\to\!\C$}; note that~$\cS$ can be also viewed as a t-embedding of the (dual to the) graph \mbox{$\Upsilon^\bullet(G)\cup\Upsilon^\circ(G)$} carrying the corresponding bipartite dimer model.
\end{minipage}\hskip 0.02\textwidth\begin{minipage}{0.48\textwidth}
\includegraphics[clip, trim=12.3cm 6.45cm 2cm 13.18cm, width=\textwidth]{Sdef_tikz.pdf}

\vskip 12pt

\noindent \hskip0.1\textwidth \begin{minipage}{0.88\textwidth}{\textsc{(B)}} The image of the same quad under \mbox{$\cS\!-\!i\cQ$} and \mbox{$\cS\!-\!i\overline{\cQ}$}; the arrows indicate possible jumps of the corresponding random walks.\end{minipage}\end{minipage}

\caption{{A tangential quad~$\cS^\dm(z)$ and its image in the S-graph:} {the positions~$(\cS-i\cQ)(z)$ and~$(\cS-i\overline{\cQ})(z)$ are horizontally aligned.\label{fig:Sdef}}}
\end{figure}

We now introduce a notion of \emph{S-graphs} associated to s-embeddings, the term is chosen by analogy with T-graphs, which play a crucial role in the analysis of t-holomorphic functions on t-embeddings.
\begin{definition} Given a proper non-degenerate s-embedding~$\cS=\cS_\cX:\Lambda(G)\to\C$, the associated function~$\cQ=\cQ_\cX:\Lambda(G)\to\R$, and a complex number~$\alpha\in\mathbb{T}$, we call~$\cS+\alpha^2\cQ:\Lambda(G)\to \C$ an S-graph associated to~$\cS$.

An S-graph is called non-degenerate if~$(\cS+\alpha^2\cQ)(v)\ne(\cS+\alpha^2\cQ)(v')$ for~$v\ne v'$.
\end{definition}

\begin{remark} \label{rem:Sgraph-deg} Provided that the s-embedding~$\cS$ is proper and non-degenerate, it is easy to see that degeneracies in~$\cS+\alpha^2\cQ$ occur if and only if~$\cS(v^\bullet)-\cS(v^\circ)\in -\alpha^2\R_+$ for an edge~$(v^\circ v^\bullet)$ of~$\Lambda(G)$; see also a discussion in Section~\ref{sub:cuts}.
\end{remark}

Clearly, S-graphs associated to an s-embedding~$\cS$ are T-graphs~$\cT+\alpha^2\cQ$ associated to the same embedding~$\cT=\cS$ viewed as a t-embedding of the corresponding bipartite dimer model, where we do not keep track of the positions of the centers of quads~$z\in\Dm(G)$; see~\cite[Section~4]{CLR1}. Note however that (since the origami map \mbox{$\cQ=\cO$} is real-valued on~$\Lambda(G)$) the restriction of~$\cT+\alpha^2\cQ$ onto~$\Lambda(G)$ can be viewed as either~$\cT+\alpha^2\cO$ or \mbox{$\cT+\alpha^2\overline{\cO}$} from the dimer model perspective; see {Fig.~\ref{fig:Sdef}B.} (The choice between~$\cO$ and~$\overline{\cO}$ corresponds to the identification of s-holomorphic functions with either t-white- or t-black-holomorphic functions discussed above.)

\smallskip

The following property of S-graphs associated with non-degenerate s-embeddings can be easily obtained from the definition~\eqref{eq:cQ-def} of the function~$\cQ$.
\begin{itemize}
\item For each quad~$z\in\Dm(G)$ and each~$p=0,1$, $q=0,1$, we have
\begin{equation}
\label{eq:S:Im>Im}
\Re[\,\overline{\alpha}^2(\cS+\alpha^2\cQ)(v_p^\bullet(z))\,]\ \ge\ \Re[\,\overline{\alpha}^2(\cS+\alpha^2\cQ)(v_q^\circ(z){)}\,],
\end{equation}
the equality is possible only if~$\cS+\alpha^2\cQ$ is degenerate. In particular, the image~$(\cS+\alpha^2\cQ)^\dm(z)$ of~$\cS^\dm(z)$ in a non-degenerate S-graph is always \emph{non-convex.} Moreover, at most one edge~$(v_p^\bullet(z)v_q^\circ(z))$ of a quad~$z\in\Dm(G)$ can  collapse to a point if the S-graph is degenerate.
\end{itemize}

The following two simple geometric properties of S-graphs can be easily obtained from the correspondence with T-graphs explained above; see { Fig.~\ref{fig:Sdef}B.}
\begin{itemize}
\item For each~$z\in\Dm(G)$, the points~$(\cS+\alpha^2\cQ)(z)$ and~$(\cS+\alpha^2\overline{\cQ})(z)$ are the intersection points of the opposite sides of the non-convex quad~$(\cS+\alpha^2\cQ)^\dm(z)$.

\smallskip

\item The line passing through these two intersection points is parallel to~$i\alpha^2\R$,
\end{itemize}

As discussed in~\cite[Section~4.2]{CLR1}, t-holomorphic functions on a t-embedding~$\cT$ can be viewed as gradients of harmonic functions on the associated T-graphs. More precisely (see~\cite[Proposition~4.15]{CLR1}), t-white-holomorphic functions can be identified with gradients of $\alpha\R$-valued harmonic functions on the T-graph~$\cT+\alpha^2\cO$, while \mbox{t-black-holomorphic} functions are gradients of~$\overline{\alpha}\R$-valued harmonic functions on the T-graph~$\cT+\overline{\alpha^2\cO}$. {When} translated to the Ising model context, this identification reads as follows. {(It is worth noting that} we do \emph{not} rely upon this material in our paper {except that in a discussion given in Section~\ref{sub:shol-limits} below} and list it here {mostly} for the completeness of the presentation.)

\begin{itemize}
\item For each~$\alpha\in\mathbb{T}$, s-holomorphic functions~$F(z)$ on~$\cS$ can be viewed as gradients of $\varsigma\alpha\R$-valued functions~$\varsigma\rI_{\alpha\R}[\overline{\varsigma}F]$ that are linear on edges of the \mbox{S-graph} $\cS+\alpha^2\cQ$ and admit an \emph{affine} continuation to quads~$(\cS+\alpha^2\cQ)^\diamond(z)$ (in other words, the four points~$(\cS(v),{ \varsigma\rI_{\alpha\R}[\overline{\varsigma}F]}(v))\in\R^{2,1}$,~$v=v_{0,1}^\bullet(z),v_{0,1}^\circ(z)$, must be coplanar for each~$z\in\Dm(G)$). {Moreover, there exists a complex-valued function~$\rI_\C[{\overline\varsigma}F]$ defined on $\Lambda(G)$ such that $\rI_{\alpha\R}[\overline{\varsigma}F]=\Pr{\rI_\C[\overline{\varsigma}F]}{\alpha\R}$ for all $\alpha\in\mathbb{T}$.
    More precisely, $\rI_\C[\overline{\varsigma}F]$ is the primitive of the closed differential form
    \[
    \qquad \overline{\varsigma}Fd\cS+\varsigma \overline{F}d\cQ\ =\ F_\frb d\cS+\overline{F_\frb}d\cQ\,.
    \]
    Slightly abusing the notation, one can similarly define $\rI_\C[\overline{\varsigma}F]$ not only on vertices of~$\Lambda(G)$ or, equivalently, on those of the s-embedding~$\cS$ but also inside quads~$\cS^\dm(z)$. Also, note that this closed form can be equivalently written as $F_\frw d\cS+\overline{F_\frw}d\overline{\cQ}$ when working on edges of~$\Lambda(G)$ but this expression leads to a different continuation of~$\rI_\C[\overline{\varsigma}F]$ inside~$\cS^\dm(z)$.}
\end{itemize}

Below (notably, in Section~\ref{sec:RSW}) we will often choose a particular~$\alpha\in\mathbb{T}$ for concreteness, our preferred choice will be~$\alpha=\overline{\varsigma}=e^{-i\frac{\pi}{4}}$. In this case, the condition~\eqref{eq:S:Im>Im} reads as $\Im(\cS-i\cQ)(v^\circ_q)\ge\Im(\cS-i\cQ)(v^\bullet_p)$ for both $q=0,1$ and $p=0,1$.

\subsection{Functions~${H_F}$ associated to s-holomorphic functions~${F}$} \label{sub:HF-def}
In this section we discuss the notion of a ``primitive of the square of {an} s-holomorphic function'', which was introduced by Smirnov in his seminal paper~\cite{Smi-I} in the square grid context (though only on the set~$\Lambda(G)$ and not on~$\Dm(G)$). We first give an \emph{abstract} definition of these primitives for spinors on~$\Upsilon^\times(G)$ satisfying the propagation equation~\eqref{eq:3-terms} and then discuss its \emph{geometric} interpretation in the context of s- and t-embeddings.
\begin{definition}
\label{def:HX-def} Let a real-valued spinor~$X$ on $\Upsilon^\times(G)$ (locally) {satisfy} the propagation equation~\eqref{eq:3-terms}. We define a function~$H_X$ on~$\Lambda(G)\cup\Dm(G)$ by setting
\begin{equation}
\label{eq:HX-def}
\begin{array}{rcll}
H_X(v^\bullet_p(z))-H_X(z)&:=&X(c_{p0}(z))X(c_{p1}(z))\cos\theta_z, & p=0,1,\\[2pt]
H_X(v^\circ_q(z))-H_X(z)&:=&-X(c_{0q}(z))X(c_{1q}(z))\sin\theta_z,& q=0,1,\\[2pt]
H_X(v^\bullet_p(z))-H_X(v^\circ_q(z))&:=&(X(c_{pq}(z)))^2,
\end{array}
\end{equation}
similarly to~\eqref{eq:cS-def} and~\eqref{eq:cS(z)-def}; note that~$H_X$ is defined up to a global additive constant,
\end{definition}
\begin{remark}\label{rem:HX(z)=} The consistency of the definition~\eqref{eq:HX-def} easily follows from the propagation equation~\eqref{eq:3-terms}. Moreover, a straightforward computation shows that~\eqref{eq:HX-def} and~\eqref{eq:3-terms} also imply the identity
\begin{equation}
\label{eq:HX(z)=}
H_X(z)=\tfrac{1}{2}\big[(H_X(v^\bullet_0)+H_X(v^\bullet_1))\cdot\sin^2\theta_z+ (H_X(v^\circ_0)+H_X(v^\circ_1))\cdot\cos^2\theta_z\big].
\end{equation}
In other words, the value $H_X(z)$ defined by~\eqref{eq:HX-def} equals the weighted average of the four values~$H_X(v^{\bullet,\circ}_{0,1})$, with coefficients~$\tfrac{1}{2}\sin^2\theta_z$ and~$\tfrac{1}{2}\cos^2\theta_z$.
\end{remark}

Now let~$\cS$ be an s-embedding of~$(G,x)$. Recall that Proposition~\ref{prop:shol=3term} provides a correspondence {between} real-valued spinors~$X$ satisfying~\eqref{eq:3-terms} and s-holomorphic functions on~$\cS$. This correspondence can be used to translate Definition~\ref{def:HX-def} to the language of s-holomorphic functions. On the other hand, in the previous section we saw that each s-holomorphic function~$F$ on~$\cS$ can be viewed as both t-white- or t-black-holomorphic functions~$F_\frw$ or~$F_\frb$ provided that~$\cS=\cT$ is viewed as a t-embedding of the corresponding bipartite dimer model on~$\Upsilon^\circ(G)\cup\Upsilon^\bullet(G)$. Further, in the dimer model context, given a pair of t-holomorphic functions~$F_\frw$ and~$F_\frb$ one can consider a differential form
\begin{equation}
\label{eq:FwFbdT-closed}
\tfrac{1}{2}\Re(F_\frw^\circ F^\bullet_{\frb} d\cT+F_\frw^\circ\overline{F_\frb^\bullet}d\cO),
\end{equation}
which turns out to be closed; see~\cite[Proposition~3.10]{CLR1}. If both~$F_\frw^\circ=\overline{\varsigma}F$ and \mbox{$F_\frb^\bullet=\overline{\varsigma}F$} correspond to the same s-holomorphic function~$F$, this allows us to define the primitive
\begin{equation}
\label{eq:HF-def}
H_F:=\int\tfrac{1}{2}\Re(\overline{\varsigma}^2F^2d\cS+|F|^2d\cQ)=\int\tfrac{1}{2}(\Im(F^2d\cS)+\Re(|F|^2d\cQ)),
\end{equation}
on~$\Lambda(G)\cup\Dm(G)$. According to~\cite[{Remark~3.11}]{CLR1}, $H_F$ can be also viewed as a piecewise affine function on faces of the t-embedding~$\cT=\cS$. Note however that~$H_F$ is \emph{not} affine on quads~$\cS^\dm(z)$, which are composed of four triangular faces of~$\cT$, each of the four carrying its own constant differential form~$d\cQ$; see {Fig.~\ref{fig:Sdef}A.}

\begin{lemma} Let an s-holomorphic function $F$, defined on a (subset) of~$\Dm(G)$, and a real-valued spinor~$X$ on~$\Upsilon^\times(G)$ be related by the identity~\eqref{eq:X-from-F}. Then, the functions~$H_F$ and~$H_X$ coincide (up to a global additive constant).
\end{lemma}
\begin{proof} Indeed, for a quad~$\cS^\dm(z)=(\cS(v_0^\bullet)\cS(v_0^\circ)\cS(v_1^\bullet)\cS(v_1^\circ))$ as in {Fig.~\ref{fig:Sdef}A} we have
\begin{align*}
H_X(v^\bullet_0)-H_X(z)\ &=\ X(c_{00})X(c_{01})\cos\theta_z\\
&=\ \overline{\varsigma}^2(\cS(v^\bullet_0)-\cS(z))\cdot \Pr{F(z)}{\eta_{00}\R}\cdot\Pr{F(z)}{\eta_{01}\R}\\
&=\ \tfrac{1}{4i}(\cS(v^\bullet_0)-\cS(z))\cdot(F(z)\!+\!\eta_{00}^2\overline{(F(z))}{}^2)\cdot(F(z)\!+\!\eta_{01}^2\overline{(F(z))}{}^2),
\end{align*}
where we write~$\eta_{00}$ and~$\eta_{01}$ instead of~$\eta_{c_{00}}$ and~$\eta_{c_{01}}$ for shortness.
It remains to note that
\[
\eta_{00}^2\eta^2_{01}(\cS(v^\bullet_0)-\cS(z))\ =\ \varsigma^4\cdot (\overline{\cS(v^\bullet_0)}-\overline{\cS(z)})\ =\ -(\overline{\cS(v^\bullet_0)}-\overline{\cS(z)})
\]
since $\cS(z)$ lies on the bisector of the angle~$\cS(v_1^\circ)\cS(v_0^\bullet)\cS(v_0^\circ)$ and that
\begin{align*}
(\eta_{00}^2+\eta^2_{01})(\cS(v^\bullet_0)-\cS(z))\ &=\ 2\varsigma^2\cos\phi_{v^\bullet_0 z}\cdot |\cS(v^\bullet_0)-\cS(z)|\\
&=\ 2i\cdot \Re(\cQ(v^\bullet_0)-\cQ(z)),
\end{align*}
where~$\phi_{v^\bullet_0 z}$ is the half-angle of the quad~$\cS^\dm(z)$ at the vertex~$\cS(v^\bullet_0)$. The computation of the other increments inside the quad~$\cS^\dm(z)$ is similar.
\end{proof}
\begin{remark} In the isoradial context, the function~$\cQ$ is constant on both~$G^\bullet$ and~$G^\circ$.
Therefore, on each of these graphs the function~$H_F$ can be viewed as the primitive of the form~$\tfrac{1}{2}\Im[F^2d\cS]$, {without the second term~$\frac12|F|^2d\cQ$;} cf.~\cite[Section~3.3]{ChSmi2}.
\end{remark}

We now state a \emph{maximum principle} for functions~$H_F=H_X$ associated with s-holomorphic functions. This statement
holds in the full generality and does \emph{not} rely upon any particular property of the Ising weights. In what follows, when speaking about the adjacency relation on~$\Lambda(G)\cup\Dm(G)$, we view this graph as the dual to~$\Upsilon^\bullet(G)\cup\Upsilon^\circ(G)$; see { Fig.~\ref{fig:notation}} for an illustration. In other words, given $u,u'\in \Lambda(G)\cup\Dm(G)$ we say that~$u\sim u'$ if
\begin{itemize}
\item either $u$ and $u'$ are adjacent vertices of the graph~$\Lambda(G)$
\item or exactly one of these vertices belongs to~$\Dm(G)$ and the other is one of its four adjacent vertices of~$\Lambda(G)$.
\end{itemize}

\begin{proposition} \label{prop:HF-max}
Let a spinor~$X:\Upsilon^\times(G)\to\R$ locally {satisfy} the propagation equation~\eqref{eq:3-terms} and the function~$H_X:\Lambda(G)\cup\Dm(G)\to\R$ be defined by~\eqref{eq:HX-def}. Then, $H_X$ cannot attain an extremum at an interior (in the sense of the adjacency relation on~$\Lambda(G)\cup\diamondsuit(G)$ described above) vertex in the domain of its definition.
\end{proposition}

\begin{proof} Assume first that the function~$H_X$ attains an extremal value at an~\emph{isolated} vertex~$u\in\Lambda(G)\cup\Dm(G)$.
It immediately follows from the identity~\eqref{eq:HX(z)=} that the case~$u\in\Dm(G)$ is impossible. To rule out the case~$u\in\Lambda(G)=G^\bullet\cup G^\circ$, assume that~\mbox{$u=v^\circ\in G^\circ$}; the other case is symmetric. Due to Definition~\ref{def:HX-def}, the value~$H_X(v^\circ)$ cannot be a local maximum as~$H_X(v^\bullet)\ge H_X(v^\circ)$ provided that~$v^\circ\sim v^\bullet\in G^\bullet$. Let~$z_k\in\Dm(G)$, $k=1,\ldots,n$, be the adjacent to~$v^\circ$ vertices of~$\Dm(G)$ and denote by $c_k\in\Upsilon(G)$, $k=1,\ldots,n$, the `corners' adjacent to~$v^\circ$. It easily follows from~\eqref{eq:HX-def} that
\[
\textstyle \prod_{k=1}^n(H_X(z_k)-H_X(v^\circ))\ =\ - \prod_{k=1}^n\sin\theta_{z_k}\cdot\prod_{k=1}^n(X(c_k))^2,
\]
the `$-$' sign comes from the fact that the double cover~$\Upsilon^\times(G)$ branches over~$v^\circ$. Therefore, the value~$H_x(v^\circ)$ cannot be a local minimum.

It remains to consider the degenerate case when the extremum of~$H_X$ is attained at \emph{several} interior vertices simultaneously.
Let~$V\subset \Lambda(G)\cup\Dm(G)$ be a connected component of this set. It follows from~\eqref{eq:HX(z)=} that, if~$z\in V\cap\Dm(G)$, then all the four adjacent to~$z$ vertices~$v^{\bullet,\circ}_{0,1}\in\Lambda(G)$ should also belong to~$V$; as usual, here and below we use the notation introduced in {Fig.~\ref{fig:notation}}. Finally, it easily follows from Definition~\ref{def:HX-def} that, if~$v^{\bullet,\circ}\in V\cap G^{\bullet,\circ}$ are adjacent vertices, then both~$z\in\Dm(G)$ such that~$v^\bullet\sim z\sim v^\circ$ also belong to the set~$V$ (indeed, if
~$H_X(v_p^\bullet)=H_X(v_q^\circ)$ for some~$p,q\in\{0,1\}$, then~$X(c_{pq})=0$ and hence $H_X(z)=H_X(v_p^\bullet)=H_X(v_q^\circ)$). Therefore, the set~$V$ cannot contain more than one interior vertex of the domain of definition of~$H_X$ unless it contains all of them.
\end{proof}

\begin{remark} Let us recall that a general Kadanoff--Ceva fermionic observable
\[
X_\varpi(c)=\E[\chi_c\sigma_{v_1^\bullet}\ldots \mu_{v_{m-1}^\bullet}\sigma_{v_1^\circ}\ldots \sigma_{v_{n-1}^\circ}],\quad c\in\Upsilon^\times_\varpi(G)
\]
is a spinor on the double cover~$\Upsilon^\times_\varpi(G)$, which does \emph{not} branch over vertices from the set~$\varpi=\{v_1^\bullet,\ldots,v_{m-1}^\bullet,v_1^\circ,\ldots,v_{n-1}^\circ\}$. Therefore, the corresponding function~$H_X$, well-defined on~$\Lambda(G)$, can only have maxima at points~$v_p^\bullet$, $p=1,\ldots,m-1$, and minima at~$v_q^\circ$, $q=1,\ldots,n-1$, but not other extrema in the interior of~$\Lambda(G)\cup\Dm(G)$.
\end{remark}

Remarkably enough, Proposition~\ref{prop:HF-max} admits an important generalization: the following abstract \emph{comparison principle} was communicated to us by S.~C.~Park; {see~\cite[Lemma~3.6 and Proposition~4.4]{park-iso} for some intuition behind these statements coming from a treatment of the near-critical model on~$\mathbb{Z}^2$ and the comparison principle for solutions of quasilinear elliptic PDEs.} Though we do not use Proposition~\ref{prop:HF-comparison} in our paper, we decided to include it for the completeness of the exposition. Let us repeat that, similarly to Proposition~\ref{prop:HF-max}, this statement holds in the full generality, i.e., without any assumption on the Ising weights and/or the combinatorics of the planar graphs under consideration.
\begin{proposition} \label{prop:HF-comparison} Let spinors~$X,Y:\Upsilon^\times(G)\to\R$ locally satisfy the propagation equation~\eqref{eq:3-terms} and the functions~$H_X,H_Y:\Lambda(G)\cup\Dm(G)\to\R$ be defined by~\eqref{eq:HX-def}. Then, the difference $H_X-H_Y$ cannot have an extremum at an interior vertex of its domain of definition.
\end{proposition}

{Before giving a proof, let us mention an equivalent formulation of this fact: if spinors $X_1,X_2:\Upsilon^\times(G)\to\R$ (locally) satisfy the equation~\eqref{eq:3-terms}, then the function $H[X_1,X_2]:=H_{\frac{1}{2}(X_1+X_2)}-H_{\frac{1}{2}(X_1-X_2)}$ satisfies the maximum principle. Note that
\[
H[X_1,X_2](v_p^\bullet(z))-H[X_1,X_2](u_q^\circ(z))\ =\ X_1(c_{pq}(z))X_2(c_{pq}(z))\,,
\]
which is nothing but the polarization identity applied to Definition~\ref{def:HX-def}.}

\begin{proof}[Proof {of Proposition~\ref{prop:HF-comparison}}] This observation is due to S.~C.~Park. As in the proof of Proposition~{\ref{prop:HF-max},} assume first that the function~$H_X-H_Y$ has an \emph{isolated} extremum, which has to be attained at a vertex~$u\in\Lambda(G)$ (the case~$u\in\Dm(G)$ is ruled out by the identity~\eqref{eq:HX(z)=}). Without loss of generality, assume that~$u=v^\circ\in G^\circ$ and (exchanging the roles of~$X$ and~$Y$ if needed) that~$H_X-H_Y$ attains a local minimum at~$v^\circ$. Let~$z_k=(v^\circ v_k^\bullet v^\circ_k v_{k+1}^\bullet)\in \Dm(G)$, $k=1,\ldots,n$ be the neighboring to~$v^\circ$ quads listed counterclockwise, and denote~$c_k=(v^\circ v_k^\bullet)$. Since~$v^\circ$ is an isolated minimum of the function~$H_X-H_Y$, we must have the following strict inequalities:
\[
\begin{array}{l}
H_X(v_k^\bullet)-H_X(v^\circ)>H_Y(v_k^\bullet)-H_Y(v^\circ),\\
H_X(z_k)-H_X(v^\circ)>H_Y(z_k)-H_Y(v^\circ),
\end{array}\quad k=1,\ldots,n,
\]
Due to Definiton~\ref{def:HX-def}, this implies that, for all~$k=1,\ldots,n$, we have
\[
|X(c_k)|>|Y(c_k)|\quad \text{and}\quad X(c_k)X(c_{k+1})>Y(c_k)Y(c_{k+1})
\]
provided that the lifts of~$c_k$ to the double cover~$\Upsilon^\times(G)$ are chosen so that~$c_k\sim c_{k+1}$ on this double cover. In particular, we should have~$X(c_k)X(c_{k+1})>0$ for all~$k$, which leads to a contradiction with the fact that~$\Upsilon^\times(G)$ branches over~$v^\circ$.

It remains to rule out the degenerate case when an extremum of~$H_X-H_Y$ is attained at \emph{several} neighboring vertices~$u\in V$ of the graph~$\Lambda(G)\cup\Dm(G)$ simultaneously. Replacing the functions~$X$ and~$Y$ by $(1+\varepsilon)\cdot X$ and~$(1+\varepsilon')\cdot Y$ with~$\varepsilon,\varepsilon'\to 0$, one can break such a degeneracy (and thus apply the arguments given above before passing to the limit $\varepsilon,\varepsilon'\to 0$) unless \emph{both} functions~$H_X$ and~$H_Y$ are constant on the set~$V$. The latter scenario can be ruled out similarly to the proof of Proposition~\ref{prop:HF-max}: if~$z\in V\cap \Dm(G)$, then all its four neighbors must belong to~$V$ due to the identity~\eqref{eq:HX(z)=}, and if two adjacent vertices~$v^\bullet\sim v^\circ$ belong to~$V$, then~$z\in V$ for both neighboring (i.e., satisfying~$v^\bullet\sim z\sim v^\circ$) vertices~$z\in\Dm(G)$.
\end{proof}

We conclude this section {with} a useful application of Proposition~\ref{prop:HF-max} to observables~$X(c)=\E[\chi_c\mu_{(ba)^\bullet}\sigma_{(ab)^\circ}]$ from Theorem~\ref{thm:FK-conv}. As~$X(a)=X(b)=1$, one can choose an additive constant in the definition of the associated function~$H_X$ so that
\begin{equation}
\label{eq:HF2-bc}
H_X((ab)^\circ)=0\quad \text{and}\quad H_X((ba)^\bullet)=1.
\end{equation}

\begin{corollary} Let the observable~$X(c)$ be given by~\eqref{eq:KC-Dob-def} and the function~$H_X:\Lambda(G)\cup\Dm(G)\to\R$ be defined as above. Then, all the values of~$H_X$ belong to~$[0,1]$.
\end{corollary}

\subsection{Random walks on S-graphs}\label{sub:rwalks}
We now move to \emph{random walks} on S-graphs, which are nothing but random walks on T-graphs discussed in~\cite[Section~4]{CLR1}, rephrased in the Ising model context. Though one can always avoid such a translation by working directly on T-graphs instead of S-graphs, we feel that the forthcoming discussion might be of interest for those readers who prefer to keep the analysis of s-holomorphic functions on s-embeddings self-contained.

The most important output of this discussion is Proposition~\ref{prop:PrF-mart}, which says that, for each~$\alpha\in\mathbb{T}$, the functions~$\Pr{F}{\alpha\R}$ are martingales with respect to a certain directed random walk~$\widetilde{X}_t$ on vertices of the S-graph \mbox{$\cS-\overline{\varsigma}^2\overline{\alpha}^2\cQ=\cS+i\overline{\alpha}^2\cQ$.} (For~instance, functions~$\Re F$ are martingales with respect to a directed random walk on~$\cS+i\cQ$ while~$\Im F$ are martingales with respect to \emph{another} walk on~$\cS-i\cQ$.) This observation remained unnoticed until recently (for generic values~$\alpha\in\mathbb{T}$) even in the square lattice context.

For simplicity, below we assume that all S-graphs under consideration are non-degenerate; the same results in the degenerate case can be obtained by continuity in~$\alpha$ of the laws of continuous time random walks on~$\cS+\alpha^2\cQ$ defined below; see also~\cite[Remark~4.7 and Remark~4.18]{CLR1}.

A careful reader might have noticed a mismatch in the notation along the discussion given above: the functions~$\Pr{F(z)}{\alpha\R}$ are defined on quads~$z\in\Dm(G)$ but we pretend that they are martingales with respect to random walks on \mbox{$v\in\Lambda(G)$}. This inconsistency is eliminated by the following definition; see also Fig.~\ref{fig:SgraphRW}.
\begin{definition}
\label{def:valpha} Let~$\cS$ be a proper non-degenerate s-embedding,~$\alpha\!\in\!\mathbb{T}$, and assume that the S-graph~$\cS+\alpha^2\cQ$ is non-degenerate. Define a mapping~$z\mapsto v^{(\alpha)}(z)$ by requiring that~${(\cS+\alpha^2\cQ)}(v^{(\alpha)}(z))$ is the non-convex vertex of the quad~$(\cS+\alpha^2\cQ)^\dm(z)$.
\end{definition}
It is easy to see that~$v^{(\alpha)}$ defines a \emph{bijection} between~$z\in\Dm(G)$ and~$v\in\Lambda(G)$ away from the boundary of~$G$. Note that this bijection depends on~$\alpha$, the change in the correspondence happens at those~$\alpha$ for which the S-graph~$\cS+\alpha^2\cQ$ is degenerate.
\begin{definition}
For a non-degenerate S-graph~$\cS+\alpha^2\cQ$, let~$X_t=X_t^{(\alpha)}$ be a continuous time random walk on~$\Lambda(G)$ defined as follows (see also Fig.~\ref{fig:SgraphRW}):

\smallskip

\noindent (i) the only non-zero outgoing jump rates from~$v(z)=v^{(\alpha)}(z)$ are those leading to the three other vertices of the quad~$(\cS+\alpha^2\cQ)^\dm(z)$;

\smallskip

\noindent (ii) these three rates are chosen so that both coordinates of the process~$(\cS+\alpha^2\cQ)(X_t)$ and the process~$ |(\cS+\alpha^2\cQ)(X_t)|^2-t$ are martingales.
\end{definition}

\begin{figure}
\centering\includegraphics[clip, trim=5.48cm 16.04cm 6cm 6cm, width=0.81\textwidth]{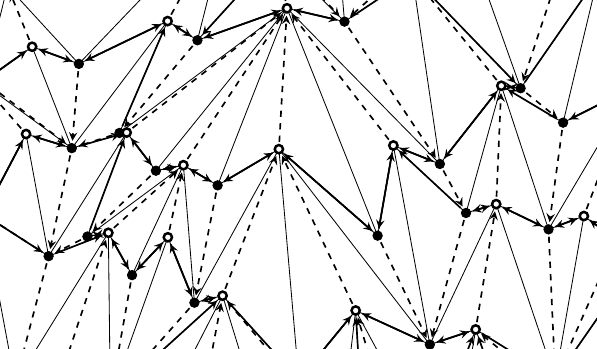}
\caption{The S-graph~$\cS-i\cQ$ {(solid edges),} arrows {(thick edges and dashed diagonals)} indicate possible jumps of the \emph{forward} random walk; {cf. Fig.~\ref{fig:Sdef}B.} For~$\varsigma=e^{i\frac{\pi}{4}}$, s-holomorphic functions~$F$ are gradients of $\R$-valued harmonic functions on \mbox{$\cS-i\cQ$}. More importantly, the functions~$\Re F$ are harmonic with respect to the \emph{time-reversal} of this random walk; see Proposition~\ref{prop:PrF-mart}.\label{fig:SgraphRW}}
\end{figure}

Recall that the S-graph~$\cS+\alpha^2\cQ$ can be viewed either as a T-graph~$\cS+\alpha^2\cQ$ or as a T-graph~$\cS+\alpha^2\overline{\cQ}$ if we add positions of points~$z\in\Dm(G)$ into consideration, see {Fig.~\ref{fig:Sdef}B.} To each of these T-graphs one can associate a natural continuous time random walk on~$\Lambda(G)\cup\Dm(G)$ satisfying the same properties (i), (ii) as~$X_t$; e.g., see~\cite[Definition~4.4]{CLR1}. Denote these walks by~$X_t^\bullet$ and~$X_t^\circ$, respectively.

It is clear that restricting~$X_t^\bullet$ and~$X_t^\circ$ to~$\Lambda(G)$ (i.e., declaring the position of the process unchanged when it jumps from a vertex~$v\in\Lambda(G)$ to~$z\in\Dm(G)$ until it makes the next jump to another vertex~$v'\in\Lambda(G)$) one obtains processes that differ from~$X_t$ only by a time change. (Indeed, each of the processes~$X_t,X_t^\bullet,X_t^\circ$ has martingale coordinates on~$\cS+\alpha^2\cQ$, which defines the transition probabilities uniquely.) Let~$\rX_n$ be the corresponding \emph{discrete time} random walk on~$\Lambda(G)$.

Recall now that the invariant measure~$\mu^\bullet$ of the continuous-time random walk~$X^\bullet_t$ (resp.,~$\mu^\circ$ for~$X^\circ_t$) on~$\Lambda(G)\cup\Dm(G)$ is given by the area of triangles of the \mbox{t-embedding} $\cT=\cS$; see \mbox{\cite[Corollary~4.9 and Proposition~4.11(vi)]{CLR1}} for precise statements. The normalization (ii) of the variance of~$X_t$ implies that the average time required for a move from~$v\in\Lambda(G)$ to the next vertex~$v'\in\Lambda(G)$ is the same for each of the three processes~$X_t$, $X_t^\bullet$, $X_t^\circ$. Comparing these processes with the discrete time random walk~$\rX_n$, it is easy to see that
\begin{equation}
\label{eq:inv-mu=}
\mu(v(z)):=\mu^\bullet(v(z))+\mu^\bullet(z)=\mu^\circ(v(z))+\mu^\circ(z)=\tfrac{1}{2}\Area(\cS^\dm(z))
\end{equation}
is the invariant measure of~$X_t$ (independently of the choice of~$\alpha\in\mathbb{T}$).

\begin{definition}
Let a (continuous time) random walk~$\widetilde{X}_t=\widetilde{X}_t^{(\alpha)}$ be the time reversal (with respect to the invariant measure~\eqref{eq:inv-mu=}) of the walk~$X_t=X_t^{(\alpha)}$.
\end{definition}
The surprising relevance of backward random walks on T-graphs was pointed out in~\cite[Proposition~4.17]{CLR1}, the following proposition is nothing but the translation of this statement to the Ising model context; see also Fig.~\ref{fig:SembBdry}.
\begin{proposition}
\label{prop:PrF-mart} Let~$F$ be an s-holomorphic function defined on (a subset of)~$\cS$. Then, for each~$|\alpha|=1$, the function~$\Pr{F}{\alpha\R}$ is a martingale with respect to the backward random walk~$\widetilde{X}_t=\widetilde{X}_t^{(i\varsigma\overline{\alpha})}$ on the S-graph~$\cS-\varsigma^2\overline{\alpha}^2\cQ=\cS-i\overline{\alpha}^2\cQ$.
\end{proposition}
\begin{proof} By definition (see {Fig.~\ref{fig:Sdef}B}), the time reversal~$\widetilde{X}_t^\bullet$ (resp., $\widetilde{X}_t^\circ$) of the random walk~$X_t^\bullet$ (resp., $X_t^\circ$) on~$\cS-\varsigma^2\overline{\alpha}^2\cQ$  (resp., $\cS-\varsigma^2\overline{\alpha}^2\overline{\cQ}$) has the following property: its only allowed jump from~$z\in\Dm(G)$ is to the vertex~$v(z)\in\Lambda(G)$. (Note that this fits the definition of a t-holomorphic function~$F_\frb^\bullet$ (resp., $F_\frw^\circ$) which has the same values $\overline{\varsigma}F$ on both black (resp., white) faces of~$\cT=\cO$ corresponding to~$z$.) It is not hard to see (e.g., passing to discrete time random walks) that, up to time parameterizations, the trajectories (restricted to~$\Lambda(G)$) of the time reversals of~$X_t^\bullet$, $X_t^\circ$ have the same law as trajectories of the time reversal of~$\rX_n$. Therefore, the claim follows directly from~\cite[Proposition~4.17]{CLR1}.
\end{proof}
\begin{remark}\label{rem:associatedRW} In some situations it is convenient to view random walks~$X_t$ and~$\widetilde{X}_t$ as defined on the graph~$\Dm(G)$ rather than on~$\Lambda(G)$ (or on the corresponding \mbox{S-graphs}). Recall the relevant bijection~$v^{(\alpha)}:\Dm(G)\to\Lambda(G)$ is provided by Definition~\ref{def:valpha}. In Section~\ref{sec:RSW}, we denote these walks on~$\Dm(G)$ by~$Z_t$ and~$\widetilde{Z}_t$, respectively, and will call them \emph{forward and backward walks associated with the S-graph~$\cS-\varsigma^2\overline{\alpha}^2\cQ$.}
\end{remark}
It is worth noting that the fact that functions~$\Pr{F}{\alpha\R}$ (where~$F$ is assumed to be \mbox{s-holomorphic} on an s-embedding~$\cS$) are martingales with respect to \emph{some} directed random walk $\widetilde{Z}_t$ on~$\Dm(G)$ can be deduced from the definition~\eqref{eq:s-hol} via a simple computation similar to the proof of~\cite[Lemma~4.19]{CLR1}. For instance, if~$z_0,z_1,\ldots,z_n=z_0\in\Dm(G)$ are neighbors of~$v\in\Lambda(G)$ listed counterclockwise and~$\eta_{k+1}\in\mathbb{T}$ corresponds to the edge~$(vv_{k+1})$ separating~$z_{k}$ and~$z_{k+1}$, then the condition~\eqref{eq:s-hol} can be written as
\[
(\Re F(z_{k+1})-\Re F(z_k))\cdot\Re\eta_{k+1}+(\Im F(z_{k+1})-\Im F(z_k))\cdot\Im\eta_{k+1}=0.
\]
This implies the identity
\begin{equation}
\label{eq:ReF-RW}
\sum_{k=0}^{n-1}\biggl(\frac{\Re\eta_{k+1}}{\Im\eta_{k+1}}-\frac{\Re\eta_k}{\Im\eta_k}\biggr)\cdot\Re F(z_{k})\ =\ 0
\end{equation}
and it is easy to see that all except one of the coefficients in this sum are positive. Moreover (see also Fig.~\ref{fig:SembBdry}),
\begin{align*}
\frac{\Re\eta_{k+1}}{\Im\eta_{k+1}}-\frac{\Re\eta_k}{\Im\eta_k}<0\quad&\Leftrightarrow \quad \text{the~horizontal~axis~$\R$ lies~in~between~$\eta_k\R$~and~$\eta_{k+1}\R$}\\[-4pt]
&\Leftrightarrow\ \ {\begin{array}{l}\text{the ray $v-i\R_+$ points inside~$\cQ^\diamond(z_k)$ if~$v\in G^\bullet$}\\ \text{the ray~$v+i\R_+$ points inside~$\cQ^\diamond(z_k)$ if $v\in G^\circ$}\end{array}}\\[2pt]
&\Leftrightarrow\quad (\cS-i\cQ)(v)~\text{is~the~concave~vertex~of}~(\cS-i\cQ)^\dm(z_k).
\end{align*}
A striking feature of the identification of~$\widetilde{Z}_t$ with the time reversal of a `nice' balanced random walk~$X_t$ on an S-graph is that this allows to derive the so-called \emph{uniform crossing estimates} for~$\widetilde{Z}_t$ from those for~$X_t$, see~\cite[Section~6.3]{CLR1}.

\begin{remark}\label{rem:|F|-max-principle}
Let~$F$ be an s-holomorphic function defined on (a subset of)~$\Dm(G)$. It immediately follows from Proposition~\ref{prop:PrF-mart} that all its projections~$\Re[\,\overline{\alpha}{F}\,]$, $\alpha\in\mathbb{T}$, satisfy the maximum principle. Varying~$\alpha$, one concludes that the function~$|F-f_0|$ also satisfies the maximum principle for each constant~$f_0\in\C$.
\end{remark}

\subsection{A priori regularity theory for s-holomorphic functions}\label{sub:regularity}
We are now ready to discuss crucial a~priori regularity properties of s-holomorphic functions following the results of~\cite[Section~6]{CLR1}. Let
\[
\osc_B F\ :=\ \max_{z,z'\in B}|F(z')-F(z)|.
\]
\begin{theorem} \label{thm:F-Hol} There {exist} constants~$\beta=\beta(\kappa)>0$ and~$C=C(\kappa)>0$ such that the following estimate holds for all s-holomorphic functions~$F$ defined in a ball of radius~$R>r$ drawn over an s-embedding~$\cS$ satisfying the assumption~$\LipKd$:
\[
\osc_{\{z:\cS(z)\in B(u,r)\}} F\ \le\ C(r/R)^\beta\osc_{\{z:\cS(z)\in B(u,R)\}}F
\]
provided that~$r\ge\cst\cdot\delta$ for a constant depending only on~$\kappa$.
\end{theorem}
\begin{proof} Since s-holomorphic functions are a particular case of t-holomorphic ones, the claim directly follows from~\cite[Proposition~6.13]{CLR1}. Roughly speaking, the idea of the proof is to consider the functions~$\Re F$ and~$\Im F$ separately, and to use Proposition~\ref{prop:PrF-mart} together with uniform crossing estimates for the backward random walks on relevant T- (or S-)graphs to control the oscillations.

The role of the assumption~\LipKd\ is two-fold. First, it guarantees that the distances on~$\cS$ and on~$\cS+\alpha^2\cQ$ are uniformly comparable {(above the scale~$\delta$).} Second, {much} more importantly, it implies the uniform ellipticity (above the scale~$\delta$) of forward random walks on T- or (S-)graphs, see~\cite[Proposition~6.4]{CLR1}.
\end{proof}
The next theorem is a more-or-less straightforward analogue of~\cite[Theorem~6.17]{CLR1} for s-holomorphic functions. However, instead of the primitive of~$F_\frw^\circ$ or~$F_\frb^\bullet$ as in~\cite{CLR1}, here we use the function~$H_F$ constructed from~$F$ via~\eqref{eq:HF-def}. Recall that such functions~$H_F$ satisfy the maximum principle (see Proposition~\ref{prop:HF-max}).

\begin{theorem} \label{thm:F-via-HF} For each~$\kappa<1$ there exist constants~$\gamma_0=\gamma_0(\kappa)>0$ and \mbox{$C_0=C_0(\kappa)>0$} such that the following alternative holds. Let~$F$ be {an} s-holomorphic function defined in a ball of radius~$r$ drawn over an s-embedding~$\cS$ satisfying the assumption~$\LipKd$. Then,
\[
\begin{array}{rcl}
\text{either}\ \max_{\{z:\cS(z)\in B(u,\frac{1}{2}r)\}}|F|^2&\le& C_0r^{-1}\cdot\osc_{\{v:\cS(v)\in B(u,r)\}}H_F\\[4pt]
\text{or}\ \max_{\{z:\cS(z)\in B(u,{\frac{3}{4}r})\}}|F|^2&{>}& \exp(\gamma_0r\delta^{-1})\cdot C_0r^{-1}\osc_{\{v:\cS(v)\in B(u,r)\}}H_F
\end{array}
\]
provided that~$r\ge\cst\cdot\delta$ for a constant depending only on~$\kappa$.
\end{theorem}

\begin{remark} The first alternative is a standard Harnack-type estimate of the `gradient' of the function~$H_F$ via its maximum, similar to the estimate of the gradient of a continuous harmonic function. The second one describes a pathological scenario when the function~$F$ has exponentially big oscillations. Unfortunately, we do not know how to rule out this scenario using only the assumption~\LipKd; this is why we introduce an additional mild assumption~\ExpFat\ below.
\end{remark}

\begin{proof} Theorem~\ref{thm:F-Hol} implies the existence of two constants~$A=A(\kappa)>1$ and \mbox{$r_0=r_0(\kappa)>0$} depending only on~$\kappa$ and such that
\begin{equation}
\label{eq:oscF-A}
\osc_{\{z:\cS(z)\in B(w,r)\}}F\ \le\ \tfrac{1}{24}(1-\kappa)\cdot\osc_{\{z:\cS(z)\in B(w,Ar)\}}F\ \ \text{if}\ \ r\ge r_0\delta.
\end{equation}

Let~$C_0={32}(1-\kappa)^{-1}A$. Assume that~$\osc_{\{v:\cS(v)\in B(u,r)\}}H_F=1$ and that \mbox{$M_0^2:=|F(z_0)|^2> C_0r^{-1}$} at a point~$z_0\in\cS^{-1}(B(u,\frac{1}{2}r))$. We claim that one can iteratively construct a sequence of points~$z_0,z_1,\ldots$ that satisfy the following conditions:
\begin{equation}
\label{eq:zn->zn+1}
\begin{array}{rcl}
M_{n+1}:=|F(z_{n+1})|&\ge& 2M_n,\\[4pt]
|\cS(z_{n+1})-\cS(z_n)|&\le&{A}\cdot \max\{{2}(1-\kappa)^{-1}M_n^{-2},r_0\delta\}.
\end{array}
\end{equation}

Indeed, assume that~$z_n$ is already constructed and that~$r\ge r_0\delta$ is such that
\[
\osc_{\{z:|\cS(z)-\cS(z_n)|\le r\}}F\ \le\ \varepsilon M_n,\quad \text{where}\quad \varepsilon:=\tfrac{1}{8}(1-\kappa)>0.
\]
{Since~$\varepsilon\le \frac12$, this inequality yields~$\Re[\,(F(z)/F(z_n))^2\,]\ge (1-\varepsilon)^2$} for all~$z$ such that~$|\cS(z)-\cS(z_n)|\le r$. {We now use definition~\eqref{eq:HF-def} of the function~$H_F$ and integrate the differential form} $\tfrac{1}{2}(\Im(F^2d{\mathcal S})+\Re(|F|^2d\cQ))$ along the segment of length~$2r$ centered at~$\cS(z_n)$ and going in the direction~$i(\overline{F(z_n)})^2\R$. (Recall that this form can be viewed as defined in the complex plane in a piecewise constant manner and that~$|d\cQ|\le|d\cS|$.) {This gives the} estimate
\begin{align*}
1\ \ge\ \osc_{\{z:|\cS(z)-\cS(z_n)|\le r\}}H_F\ &\ge\ {2r\cdot\tfrac12}\big((1-\varepsilon)^2M_n^2-(\kappa\cdot M_n^2+(2\varepsilon\!+\!\varepsilon^2)M_n^2\big)\\
&{=\ r\cdot(1-\kappa-4\varepsilon)M_n^2}\ =\ r\cdot \tfrac{1}{2}(1-\kappa)M_n^2\,.
\end{align*}
where the first term comes from the minimal possible contribution of~$\Im(F^2d\cS)$ while the second is the maximal possible contribution of~$\Re(|F|^2d\cQ)$. To avoid a contradiction, we must have~$r\le 2(1-\kappa)^{-1}M_n^{-2}$, which means that
\[
\osc_{\{z:|\cS(z)-\cS(z_n)|\le r_n\}}F\ \ge\ \tfrac{1}{8}(1-\kappa)M_n,\ \ \text{where}\ \ r_n:=\max\{{ 2}(1-\kappa)^{-1}M_n^{-2},r_0\delta\}.
\]
Using~\eqref{eq:oscF-A} we obtain an estimate $\osc_{\{z:|\cS(z)-\cS(z_n)|\le A r_n\}}F\ge 3M_n$, which guarantees the existence of {a} point~$z_{n+1}$ satisfying~\eqref{eq:zn->zn+1}.

It is easy to see that
\[
|{\cS}(z_{n+1})-{\cS}(z_n)|\ \le\ \max\{4^{-n}|{\cS}(z_1)-{\cS}(z_0)|\,,\,{A}r_0\delta\}\ \le\ \max\{4^{ -n-2}r\,,\,{ A}r_0\delta\}.
\]
Therefore, the sequence of points~$\cS(z_n)$ has to make at least~${\frac{1}{8}r}\cdot ({A}r_0\delta)^{-1}$ steps to leave the disc~$B(u,{\frac34}r)$ staring from~${\cS(z_0)\in}B(u,\frac{1}{2}r)$. Since the value~$|F(z_n)|$ at least doubles on each step, the proof is complete if we set~$\gamma_0:={\frac{1}{8}A^{-1}}r_0^{-1}\log 2$.
\end{proof}


It is easy to see that the second alternative from Theorem~\ref{thm:F-via-HF} is not compatible with the assumption~\Unif\ provided that~$\delta$ is small enough. (Indeed, in this case one has a trivial estimate~$|F(z)|^2=O(\delta^{-1}\max_{p=0,1,q=0,1}|H(v^\bullet_p(z))-H(v^\circ_q(z))|)$.) Clearly, one can assume much less to rule out this `exponential blow-up' scenario. The following assumption is a variation of a similar condition from~\cite{CLR1}.
\begin{assumption}\label{assump:ExpFat} We say that a family of proper s-embeddings~$\cS^\delta$ {with~$\delta\to 0$} satisfies the assumption~\ExpFat\ {(or, more accurately, \ExpFatPrime)} on a set~$U\subset\C$ if {there exist auxiliary scales~$\delta'=\delta'(\delta)$ such that $\delta'\to 0$ as~$\delta\to 0$ and the following holds:}
\begin{center} if one removes all quads~$(\cS^\delta)^\dm(z)$ with {radii}~$r_z\ge \delta \exp(-\delta'\delta^{-1})$ from~$U$, then\\ { each of the remaining (vertex-)connected components has diameter at most~$\delta'$.}
\end{center}
\end{assumption}
In the general case, there is no uniform notion of the size of quads~$\cS^\dm(z)$, thus~$\delta$ simply denotes a scale starting from which the assumption~\LipKd\ is fulfilled. Under the assumption~$\Unif$ this is, up to a multiplicative constant, just the same parameter and~\ExpFat\ holds with a huge margin: {indeed, in this case one has~$r_z\ge\mathrm{cst}\cdot\delta$ for \emph{all} quads and hence~$\delta'$ can be chosen to be a multiple of~$\delta$.}

\begin{corollary}\label{cor:H-Lip} Let~$\kappa<1$ and a sequence of s-embeddings~$\cS^\delta$ {with~$\delta\to 0$} satisfies both assumptions~\LipKd\ and~\ExpFat\ in a disc~$U=B(u,r)$. Assume that~$F^\delta$ is an \mbox{s-holomorphic} function on~$\cS^\delta$ and that~${\osc_{\{v:\cS(v)\in U\}}} H_{F^\delta}\le M$ for all~$\delta$. Then, the following uniform estimate holds {for sufficiently small $\delta$}:
\begin{equation}
\label{eq:F-via-HF-Harnack}
|F^\delta(z)|^2\ {\le\ C_0r^{-1}M}\quad \text{if}\ \ \cS^\delta(z)\in B(u,\tfrac{1}{2}r),
\end{equation}
{where $C_0=C_0(\kappa)>0$ is the constant from Theorem~\ref{thm:F-via-HF}.} In particular, the functions~$H_{F^\delta}$ are uniformly Lipschitz on compact subsets of~$U$.
\end{corollary}

{
\begin{remark} \label{rem:F-via-HF-Harnack}
In fact, to prove the uniform estimate~\eqref{eq:F-via-HF-Harnack} there is no need to assume that~$r$ remains fixed as~$\delta\to 0$: the proof given below only requires that $r\ge \mathrm{cst}(\kappa)\cdot\max\{\delta,\delta'\}$ for a constant~$\mathrm{cst}(\kappa)$ depending on~$\kappa$ only.
\end{remark}}

\begin{proof} {Let~$r\ge 4\max\{\gamma_0^{-1},1\}\cdot \delta'$ and $r\ge C_1\delta$, where
$\gamma_0=\gamma_0(\kappa)$ is the constant from Theorem~\ref{thm:F-via-HF} and the constant~$C_1=C_1(\kappa)>0$ is chosen so that
\[
\delta r^{-1}\log(C_0\delta r^{-1})\ge -\tfrac{1}{2}\gamma_0\quad \text{if}\quad r\ge C_1\delta\,.
\]}
It follows from Corollary~\ref{cor:F-from-X} and from the formula~\eqref{eq:rz=} that
\[
|F^{\delta}(z)|^2\ \le\ \sin^2(2\theta_z)r_z^{-2}\,\cdot\!\!{\max_{p,q\in\{0,1\}}|\cX^\delta(c_{pq}(z))|^2(X^\delta(c_{pq}(z)))^2\ \le\ r_z^{-2}\delta M.}
\]
{If $r_z\ge \delta\exp(-\delta'\delta^{-1})$, this crude estimate implies
\[
|F^\delta(z)|^2\ \le\ \delta^{-1}\exp(2\delta'\delta^{-1})M\
\le\ \exp(\gamma_0r\delta^{-1})\cdot C_0r^{-1}M
\]
since the last inequality is equivalent to saying that~$2\delta'\le (\gamma_0+\delta r^{-1}\log(C_0\delta r^{-1}))\cdot r$.

Now let us consider a point~$\cS^\delta(z_0)\in B(u,\frac{3}{4}r)$. Assumption~\ExpFat\ guarantees that $\cS(z_0)$ is surrounded by a circuit of quads $\cS^\delta(z)\in B(u,r)$ for which the estimate given above holds: otherwise, the vertex-connected component of~$\cS^\delta(z_0)$ remained after removing all quads with~$r_z\ge \delta\exp(-\delta'\delta^{-1})$ would have diameter greater than~$\frac{1}{4}r\ge \delta'$. Applying the maximum principle for the absolute values of s-holomorphic functions (see Remark~\ref{rem:|F|-max-principle}) we conclude that
\[
|F^\delta(z)|^2\ \le\ \exp(\gamma_0r\delta^{-1})\cdot C_0r^{-1}M\ \ \text{everywhere in~$B(u,\tfrac34r)$}
\]
and thus rule out the second (pathological) alternative in Theorem~\ref{thm:F-via-HF}. }

As the functions~$H_{F^\delta}$ can be obtained from~$F^\delta$ by integrating a piecewise constant differential form~\eqref{eq:HF-def} directly in~$\C$ (and not just along edges of~$\cS^\delta$), the a~priori Harnack estimate~\eqref{eq:F-via-HF-Harnack} implies the uniform Lipschitzness of~$H_{F^\delta}$.
\end{proof}
\begin{remark}\label{rem:F-precomp} In the setup of Corollary~\ref{cor:H-Lip}, the functions~$F^\delta$, $\delta\to 0$, form a precompact family in the topology of uniform convergence on compact subsets of~$B(u,r)$. Indeed, these functions are uniformly bounded and $\beta$-H\"older on scales above~$\cst(\kappa)\cdot \delta$ due to Theorem~\ref{thm:F-Hol}.
Thus, the precompactness of~$\{F^\delta\}$ follow from a version of the Arzel\`a-Ascoli theorem: {if a subsequence of~$\{F^\delta\}$ converges pointwise, say, on the set of all rational points in $B(u,r)$, then the limit is $\beta$-H\"older on \emph{all} scales and the uniform convergence on $\overline{B(u,r_1)}$ with $r_1<r$ follows by exactly the same arguments as if all~$F^\delta$ were equicontinuous on all scales.}
\end{remark}

\subsection{Subsequential limits of s-holomorphic functions} \label{sub:shol-limits}
We now discuss subsequential limits of s-holomorphic functions, both under the assumption~$\Unif$\ and in the general context.

Assume that proper s-embeddings~$\cS^\delta$ satisfy the assumption~\LipKd\ and that their images cover a fixed ball~$U=B(u,r)\subset\C$. { Since the mappings $\cQ^\delta\circ(\cS^\delta)^{-1}$ (extended from vertices of~$\cS^\delta$ to~$U$ in a piecewise affine way) are $1$-Lipschitz,} one can always find a sequence~$\delta_k\to 0$ and a Lipschitz function~$\vartheta:U\to\R$ such that
\begin{equation}
\label{eq:Q-to-theta}
\cQ^\delta\circ (\cS^\delta)^{-1}\to\vartheta\quad \text{uniformly~on~compact~subsets~of}~U~\text{as}~\delta=\delta_k\to 0.
\end{equation}
Moreover, the function~$\vartheta$ is~$\kappa$-Lipschitz (with~$\kappa<1$) on all scales as the same is true for each of the functions~$\cQ^\delta\circ(\cS^\delta)^{-1}$ on scales above~$\delta$. Clearly, under the assumption~{\Qflat}\ there is no need to pass to a subsequence and~$\vartheta\equiv0$.
\begin{proposition} \label{prop:f-hol} In the setup of Corollary~\ref{cor:H-Lip}, Remark~\ref{rem:F-precomp} and~\eqref{eq:Q-to-theta}, let $f:U\to\C$ be a subsequential limit of s-holomorphic functions~$F^\delta$ on~$\cS^\delta$. Then,
\begin{equation}
\label{eq:closed-form}
\text{the~form}\ \ \overline{\varsigma}fdz+\varsigma\overline{f}d\vartheta\ \ \text{is~closed~in}\ \ U,
\end{equation}
where~$\varsigma=e^{i\frac{\pi}{4}}$ is chosen in~\eqref{eq:def-eta}. (Recall that~$\vartheta$ is a Lipschitz function, so contour integrals of~\eqref{eq:closed-form} over smooth contours are well defined, e.g., via Riemann sums.)

With a consistent choice of additive constants, the associated functions $H_{F^\delta}$ converge to $h:=\frac{1}{2}\int(\Re(\overline{\varsigma}^2f^2dz)+|f|^2d\vartheta)$ uniformly on compact subsets of~$U$.

In particular, if~$\vartheta\equiv 0$, then $f$ is holomorphic in~$U$ and~$h=\tfrac{1}{2}\int\Im(f^2dz)$.
\end{proposition}
\begin{proof} See~\cite[Proposition~6.15]{CLR1} for the proof of~\eqref{eq:closed-form}, {which is based upon the fact that the form~$\rI_\C[\overline{\varsigma}F^\delta]=\overline{\varsigma}F^\delta d\cS^\delta+\varsigma\overline{F}{}^\delta d\cQ^\delta$ mentioned in the end of Section~\ref{sub:dimers} is closed.}
The convergence of the associated functions~$H_F$ can be easily obtained from the formula~\eqref{eq:HF-def}, where the form~$\tfrac{1}{2}(\Re(\overline{\varsigma}^2(F^\delta)^2d\cS^\delta+|F^\delta|^2d\cQ^\delta))$ is viewed as defined in~$U\subset\C$ (and not just on edges of~$\cS^\delta$) in a piecewise constant manner.
\end{proof}

Though we do not handle the general case in this paper because of the current lack of technical tools, we nevertheless feel {it is} worth mentioning a rather unexpected appearance of the \emph{Lorentz geometry} (in the Minkowski space~$\R^{2,1}$) in the planar Ising model context, which {provides a useful interpretation of the condition~\eqref{eq:closed-form}} from Proposition~\ref{prop:f-hol}. To lighten the notation, in what follows we assume that~$\varsigma=1$ (this amounts to replacing~$f$ by~$\overline{\varsigma} f$); thus the condition~\eqref{eq:closed-form} states that the differential form~$fdz+\overline{f}d\vartheta$ is closed.

{At first, assume} for simplicity that the $\kappa$-Lipschitz, $\kappa<1$, function~$z\mapsto \vartheta=\vartheta(z)$ is smooth. Let
\begin{equation}
\label{eq:zeta-param}
{\mathbb{D}\ni\zeta}\ \mapsto\ (z,\vartheta)\in{U}\times\R\,{\subset\C\times\R}\,\cong\,\R^{2,1}
\end{equation}
be an orientation-preserving \emph{conformal parametrization} of the space-like surface $(z,\vartheta(z))_{z\in U}$ equipped with a (positive) metric induced from the ambient Minkowski space. Considering the scalar product (in~$\R^{2,1}$) of infinitesimal increments of the mapping~\eqref{eq:zeta-param} one sees that this mapping preserves angles if and only if
\begin{equation}
\label{eq:conf-param}
z_\zeta\overline{z}_\zeta\ =\ (\vartheta_\zeta)^2\quad \text{and}\quad |z_\zeta|>|\vartheta_\zeta|\ge |\overline{z}_\zeta|\,,
\end{equation}
where~$z_\zeta:=\partial z/\partial\zeta$ (similarly,~$\overline{z}_\zeta$ and~$\vartheta_\zeta$) stands for the Wirtinger derivative;
{the second condition corresponds to the fact that the surface~$(z,\vartheta(z))_{z\in U}$ is space-like and the parametrization~$\zeta$ is orientation-preserving. Assuming that the function~$f$ is differentiable (in~$\zeta$), one can write the condition~\eqref{eq:closed-form} as}
\[
f_{\bar\zeta}\cdot z_{\zeta}+\overline{f}_{\bar\zeta}\cdot \vartheta_\zeta\ =\ f_\zeta\cdot z_{\bar{\zeta}}+\overline{f}_\zeta\cdot \vartheta_{\bar \zeta}\,.
\]
In fact, it is not hard to see that this is only possible if both sides \emph{vanish} since the right-hand side is the complex conjugate of the left-hand one multiplied by~$\vartheta_\zeta/z_\zeta$. Thus,~\eqref{eq:closed-form} can be further rewritten as
\begin{equation}
\label{eq:f-Beltrami}
f_{\bar\zeta}\cdot (z_{\zeta})^{1/2}+\overline{f}_{\bar\zeta}\cdot ({\overline z}_\zeta)^{1/2}\ =\ 0\,;
\end{equation}
note that both square roots~$(z_\zeta)^{1/2}$ and~$(\overline{z}_\zeta)^{1/2}=\vartheta_\zeta\cdot (z_\zeta)^{-1/2}$ are well defined in~$U$ up to the common sign. Let us now denote
\begin{equation}
\label{eq:f-to-phi-change}
\phi\ :=\ f\cdot (z_\zeta)^{1/2}+\overline{f}\cdot (\overline{z}_\zeta)^{1/2}\,.
\end{equation}
A straightforward computation relying upon~\eqref{eq:conf-param} shows that
\[
\phi_{\bar\zeta}\ =\ m\cdot \overline{\phi}\,,\quad \text{where}\quad m=\frac{z_{\zeta\bar\zeta}}{2(z_\zeta z_{\bar \zeta})^{1/2}}= \frac{(|\vartheta_{\zeta\bar\zeta}|^2-|z_{\zeta\bar\zeta}|^2)^{1/2}}{2(|z_\zeta|-|z_{\bar\zeta}|)}.
\]
Finally, note that~$m=\frac{1}{2}H\ell$, where $\ell=|z_\zeta|-|\overline{z}_\zeta|=(|z_\zeta|^2+|\overline{z}_\zeta|^2-2|\vartheta_\zeta|^2)^{1/2}$ is the metric element of the surface~$(z,\vartheta(z))\in\R^{2,1}$ in the parametrization~\eqref{eq:zeta-param} and
\[
H\ =\ \pm\|(z_{\zeta\bar\zeta}{\,,\,}\vartheta_{\zeta\bar\zeta})\|_{\R^{2,1}}\cdot \ell^{-2}
\]
is the \emph{mean curvature} of this surface at a given point. In other words, $f$ satisfies the condition~\eqref{eq:closed-form} if and only if $\phi$ satisfies the massive Dirac (or Cauchy-Riemann) equation~$\partial_{\bar\zeta}\phi=m\overline{\phi}$, where the mass~$m=\frac 12 H\ell$ admits a fully geometric interpretation in terms of the space-like surface~$(z,\theta(z))_{z\in U}$ in the Minkowski space~$\R^{2,1}$.

\begin{remark} \label{rem:massive-iso}
The preceding discussion explains how massive {holomorphic functions} appear in the s-embeddings context. It is also known that these { functions} naturally appear in a near-critical model on regular lattices; e.g., see~\cite{park,park-iso} or~\cite{CIM-massive}. We refer the reader to~\cite[Section~3.3]{CIM-massive} where the link between the two pictures is explained. Namely, one starts with {an appropriate} pair of massive s-holomorphic functions $(\cF_1^\delta\,;\cF_\mathrm{i}^\delta)$ on an isoradial grid~$\Lambda^\delta$ and constructs an s-embedding~$\cS^\delta=\cS^\delta_\cX$ of~$\Lambda^\delta$, where $\cX^\delta=\cX^\delta_1-i\cX^\delta_\mathrm{i}$ is the corresponding solution of~\eqref{eq:3-terms}. {As $\delta\to 0$, the functions~$\cF_1^\delta$ converge (on compacts) to the function~$f_1(\zeta):=\exp(-2m\Im\zeta)$, the functions~$\cF_\mathrm{i}^\delta$ converge to~$f_\mathrm{i}(\zeta):=i\exp(2m\Im \zeta)$ and the graphs $(\cS^\delta;\cQ^\delta)$ converge to the surface
\[
(\Re\zeta\,;\tfrac{1}{4m}\sinh(4m\Im\zeta)\,;\tfrac{1}{4m}(1-\cosh(4m\Im \zeta)))_{\zeta\in\C}\,\subset\,\R^{2,1},
\]
which has a constant mean curvature~$m$ and of which~$\zeta$ is a conformal (and, moreover, an isometric) parametrization. In this example, even \emph{before} passing to the limit~$\delta\to 0$, the construction of the discrete surface~$(\cS^\delta;\cQ^\delta)$ out of a pair of massive s-holomorphic functions $(\cF^\delta_1,\cF^\delta_\mathrm{i})$ defined in the original plane~$\Lambda^\delta$ can be viewed as
a discrete Weierstrass-type parametrization of this surface; see~\cite[Eq.~(3.11)]{CIM-massive}.} Moreover, \cite[Proposition~3.21]{CIM-massive} provides a discretization of the formula~\eqref{eq:f-to-phi-change} that links \mbox{s-holomorphic} functions on~$\cS^\delta$ and massive s-holomorphic functions on~$\Lambda^\delta$.
\end{remark}

{The arguments given above do not directly apply to \emph{non-smooth} Lipschitz functions~$\vartheta$ as in this case even the notion of a `conformal uniformization' of the surface~$(z,\vartheta(z))_{z\in U}\subset\R^{2,1}$ requires a clarification. A possible way to bypass this issue is to construct an appropriate quasi-conformal homeomorphism~$z\mapsto\zeta$ such that~\eqref{eq:zeta-param} holds almost everywhere.

More precisely, if we assume that~$\zeta_{\bar{z}}=\mu\zeta_z$ or, equivalently, $\overline{z}_\zeta=-\overline{\mu}z_\zeta$ for a certain Beltrami coefficient~$\mu\in L^\infty(U)$ with~$\|\mu\|_\infty<1$, then the condition~\eqref{eq:conf-param} reads as
\[
-\overline{\mu}\cdot z_\zeta^2\ =\ (\vartheta_zz_\zeta+\vartheta_{\bar{z}}\overline{z}_\zeta)^2\ =\ (\vartheta_z-\overline{\mu}\vartheta_{\bar{z}})^2\cdot z_\zeta^2\,.
\]
(Note that the derivatives $\vartheta_z$ and~$\vartheta_{\bar{z}}=\overline{\vartheta_z}$ of the real-valued Lipschitz function~$\vartheta$ exist almost everywhere (in~$z$) and that quasi-conformal homeomorphisms send zero measure sets to zero measure sets; e.g., see~\cite[Theorem~3.3.7]{astala-iwaniec-martin-book}.) Thus, in order to find~$\mu$ we need to solve the quadratic equation
\[
\vartheta_z^2+\overline{\mu}(1-2|\vartheta_z|^2)+\overline{\mu}^2\vartheta_{\bar{z}}^2\ =\ 0.
\]
The fact that the function~$\vartheta(z)$ is $\kappa$-Lipschitz yields~$|\vartheta_z|\le \frac{1}{2}\kappa$ almost everywhere and one can find a solution of the quadratic equation given above by requiring that
\begin{equation}
\label{eq:mu-def}
-\frac{\overline{\mu}}{1+|\mu|^2}\ =\ \frac{\vartheta_z^2}{1-2|\vartheta_z|^2}\quad \text{with}\quad |\mu|<\frac{\frac14k^2}{1-\frac34k^2}<1.
\end{equation}
To summarize, for (non-smooth) $\kappa$-Lipschitz, $\kappa<1$, functions~$\vartheta:U\to\R$ one can construct a `conformal parametrization'~\eqref{eq:conf-param} simply by finding a homeomorphic solution~$U\ni z\mapsto \zeta\in \mathbb{D}$ of the Beltrami equation $\zeta_{\bar{z}}=\mu\zeta_z$, where $\mu$ is given by~\eqref{eq:mu-def}. The existence of such a $W^{1,2}_\mathrm{loc}$ quasi-conformal homeomorphism easily follows from the Ahlfors--Bers theorem (e.g., see~\cite[Theorem~5.3.4]{astala-iwaniec-martin-book}). Namely, one sets~$\mu:=0$ outside~$U$, constructs a normalized solution of the Beltrami equation $\zeta_{\bar{z}}=\mu\zeta_z$ in the full complex plane and then post-compose it with a conformal uniformization of the simply connected image of~$U$ onto the unit disc~$\mathbb{D}$.

Another source of difficulties in the non-smooth setup comes from the fact that the derivatives (either in~$z$ or in~$\zeta$) of subsequential limits of s-holomorphic functions are not necessarily well defined. 
As has been suggested in~\cite[Section~5]{mahfouf-thesis}, this issue can be -- at least partly -- overcome by working with the primitives of the differential forms~$fdz+\overline{f}d\vartheta$ (see Proposition~\ref{prop:f-hol}) rather than with functions~$f$ themselves. It is straightforward to check that these primitives satisfy the conjugate Beltrami equation $\overline{g}_\zeta=\nu g_\zeta$ with $\nu=\vartheta_\zeta/z_\zeta$; cf. the equation~\eqref{eq:f-Beltrami}, which \emph{formally} says that~$f_\zeta=-\overline{\nu}\overline{f}{}_\zeta$. Moreover, it is easy to see that~$g\in W^{1,2}_\mathrm{loc}$ as a function of~$\zeta\in\mathbb{D}$. (E.g., one has $g_\zeta=(f+\overline{f}\vartheta_z)z_\zeta+\overline{f}\vartheta_{\bar{z}}\overline{z}_\zeta$, the derivatives $z_\zeta$ and $\overline{z}_\zeta$ are in $L^2(\mathbb{D})$ due to the area principle, the derivatives~$\vartheta_z$ and~$\vartheta_{\bar{z}}$ are in~$L^\infty$, and $f$ is continuous.) Together, the conjugate Beltrami equation \emph{and} the a priori $W^{1,2}_\mathrm{loc}$ regularity of $g$ allow one to apply powerful analytic tools to the study of such subsequential limits; see~\cite[Section~5.5]{astala-iwaniec-martin-book} for a general discussion and~\cite[Section~5]{mahfouf-thesis} for an application of these techniques to a proof of Theorem~\ref{thm:RSW-selfdual} on general s-embeddings.}

\smallskip

We do \emph{not} elaborate the case \mbox{$\vartheta\not\equiv 0$} below and leave it for the future study.
Let us also mention that the preceding discussion suggests that there should also exist a natural interpretation of s-embeddings, s-holomorphicity and, even more importantly, an interpretation of discrete differential operators from the forthcoming Section~\ref{sec:operators} in the language of the \emph{discrete Lorentz geometry.} {We hope that our paper stimulates research progress in this direction.}

\section{Discrete complex analysis on s-embeddings}\label{sec:operators}
\setcounter{equation}{0}

\subsection{Basic differential operators associated to s-embeddings} Let $\cS=\cS_\cX$ be a proper s-embedding of a weighted planar graph $(G,x)$. In this section we introduce several discrete differential operators associated to $\cS$ and list their basic properties. For simplicity, below we always assume that $G$ has the topology of the plane (or that we consider only  \emph{bulk} vertices of $G$ in the disc setup). The following definition appeared in~\cite[Section~6]{Ch-ICM18}.
\begin{definition} \label{def:opaS}
For a function $H$ defined on (a subset of) $\Lambda$, we introduce the operator $\opaS$ as follows: for $z\in\Dm(G)$,
\[
[\opaS H](z)\,:=\,\frac{\mu_z}{4}\biggl[\frac{H(v^\bullet_0)}{\cS(v^\bullet_0)-\cS(z)}+\frac{H(v^\bullet_1)}{\cS(v^\bullet_1)-\cS(z)}- \frac{H(v^\circ_0)}{\cS(v^\circ_0)-\cS(z)}-\frac{H(v^\circ_1)}{\cS(v^\circ_1)-\cS(z)}\biggr],
\]
where the factor $\mu_z$ is chosen so that $[\opaS \overline{\cS}](z)=1$. We also set
\[
[\paS H](z)\ :=\ \overline{[\opaS \overline{H}](z)}.
\]
\end{definition}

\begin{remark} \label{rem:opaS-isorad}
In the isoradial context, one can easily see (or deduce from the next lemma) that~$\opaS$, $\paS$ coincide with the standard discrete Cauchy--Riemann operators~$\opa_{\iso}^{\Lambda}$, $\pa_{\iso}^{\Lambda}$; e.g., see~\cite[Section~2.4]{ChSmi1} for their definitions.
\end{remark}

The next lemma shows that the operators~$\opaS$ and~$\paS$ indeed can be viewed as discretizations of the differential operators~$\opa=\tfrac{1}{2}(\pa_x+i\pa_y)$ and~$\pa=\tfrac{1}{2}(\pa_x-i\pa_y)$.

\begin{lemma} \label{lem:opaS=} The following identities are fulfilled:
\[
\opaS 1=\opaS \cS=\opaS \cQ=0,\qquad \opaS\overline{{\cS}}=1.
\]

Moreover, under the assumption~\Unif\ one has
\[
[\opaS\phi](z)=[\opa\phi](z)+O(\delta\cdot\max\nolimits_{\cS^\dm(z)}|D^2\phi|)
\]
for each $C^2$-smooth function~$\phi$ defined on the quad~$\cS^\dm(z)\subset\C$.

Similar statements hold for the operator~$\paS$.
\end{lemma}
\begin{proof} Recall that, with a proper choice of signs (as in~\eqref{eq:cS(z)-def}), one has
\[
\cS(v^\bullet_p)-\cS(z)=\pm\cX(c_{p0})\cX(c_{p1})\cos\theta_z, \quad \cS(v^\circ_q)-\cS(z)= \pm\cX(c_{0q})\cX(c_{1q})\sin\theta_z.
\]
Assume that the corners~$c_{11}$\,---\,$c_{01}$\,---\,$c_{00}$\,---\,$c_{10}$ are chosen consecutively on~$\Upsilon^\times(G)$ as in~\eqref{eq:rz=}. Then, the propagation equation~\eqref{eq:3-terms} reads as
\[
\cX(c_{01})=\cX(c_{00})\cos\theta_z+\cX(c_{11})\sin\theta_z,\quad \cX(c_{10})=\cX(c_{00})\sin\theta_z-\cX(c_{11})\cos\theta_z.
\]

A simple computation shows that
\[
(\cX(c_{10})\cX(c_{11})-\cX(c_{01})\cX(c_{00}))\sin\theta_z+(\cX(c_{00})\cX(c_{10})+\cX(c_{01})\cX(c_{11}))\cos\theta_z=0
\]
and so~$[\opaS 1](z)=0$. As a corollary, we also have~$[\opaS S](z)=[\opa(\cS(\cdot)-\cS(z))](z)=0$.

\smallskip

Further, the identity~$[\opaS\cQ](z)=[\opaS(\cQ(\cdot)-\cQ(z))](z)=0$ holds since
\[
\frac{\cQ(v^\bullet_p)-\cQ(z)}{\cS(v^\bullet_p)-\cS(z)}= \frac{1}{2}\biggl[\frac{\overline{\cX(c_{p0})}}{\cX(c_{p0})}+\frac{\overline{\cX(c_{p1})}}{\cX(c_{p1})}\biggr],\quad
\frac{\cQ(v^\circ_q)-\cQ(z)}{\cS(v^\circ_q)-\cS(z)}= \frac{1}{2}\biggl[\frac{\overline{\cX(c_{0q})}}{\cX(c_{0q})}+\frac{\overline{\cX(c_{1q})}}{\cX(c_{1q})}\biggr].
\]

Finally, if $\phi$ is a $C^2$--smooth function on the quad~$\cS^\dm(z)\subset\C$, then
\[
\phi({v})=\phi(z)+(\cS(v)-\cS(z))[\pa\phi](z)+(\overline{\cS(v)}-\overline{\cS(z)})[\opa\phi(z)] +O(\delta^2\max\nolimits_{\cS^\dm(z)}|D^2\phi|)
\]
for all~$v\sim z$. This implies the last claim since~$\opaS 1=\opaS\cS=0$, $\opaS\overline{\cS}=1$ and the coefficients of the operator~$\opaS$ are~$O(\delta^{-1})$ under the assumption~\Unif.
\end{proof}

\begin{corollary} \label{cor:opaS-HF}
Let a function~$H_F$ be constructed from an s-holomorphic function~$F$ as in Section~\ref{sub:HF-def}; see~\eqref{eq:HF-def}. Then,
\[
\paS H_F=\tfrac{1}{4i}F^2\quad \text{and}\quad \opaS H_F=-\tfrac{1}{4i}\overline{F}{}^2.
\]
\end{corollary}
\begin{proof} It follows from the definition of the function~$H_F$ that its values at four vertices of a given quad~$z\in\Dm(G)$ coincide with the values of the function
\[
\tfrac{1}{2}\big[\tfrac{1}{2i}(F(z))^2\cS(\cdot)-\tfrac{1}{2i}(\overline{F(z)})^2\overline{\cS}(\cdot)+|F(z)|^2\cQ(\cdot)\big]
\]
up to an additive constant depending on~$z$. As~$\paS 1 =\paS\overline{\cS}=\paS\cQ=0$,~\mbox{$\paS\cS=1$,} this implies the identity~$\paS H_F=\tfrac{1}{4i}F^2$. The computation of~$\opaS H_F$ is similar.
\end{proof}

We now introduce another pair of operators $\pao$, $\opao$ (acting on functions defined on~$\Dm(G)$ rather than on~$\Lambda(G)$) associated to s-embeddings.

\begin{definition} \label{def:omega} For a vertex $c$ of the medial graph $\Upsilon(G)$, denote
\[
\alpha(c)=\alpha_\cS(c)\ :=\ \arg(\cS(v^\bullet(c))-\cS(v^\circ(c)))=2\arg\cX(c).
\]
Further, let an exact 1-{form} $d\omega=d\omega_\cS$ on oriented edges $(cd)$ of the medial graph $\Upsilon(G)$ be defined by
\[
d\omega((cd))=d\omega_\cS((cd))\ :=\ e^{i\alpha_\cS(d)}-e^{i\alpha_\cS(c)}.
\]
\end{definition}
\begin{remark} \label{rem:omega=S-S/Q-Q}
If a corner~$c\in\Upsilon(G)$ corresponds to an edge~$(vv')$ of~$\Lambda(G)$, then
\[
e^{i\alpha_\cS(c)}\ =\ ({\cS(v')-\cS(v)})\,/\,({\cQ(v')-\cQ(v)}),
\]
independently of whether~$v\in G^\bullet$ and~$v'\in G^\circ$ or~$v\in G^\circ$ and~$v'\in G^\bullet$.
\end{remark}
\begin{definition} \label{def:opao} Let a function $F$ be defined on (a subset of) $\Dm(G)$. For $v\in \Lambda(G)$, set
\[
[\opao F](v)\ :=\ \frac{1}{2i}\sum_{z_k\in\Dm(G):\ z_k\sim v}F(z_k)d\omega((z_kv)^*),
\]
where the edge $(z_kv)^*$ of $\Upsilon(G)$ is assumed to be oriented so that~$z_k$ stays {on} the right. Similarly, let
\[
[\pao F](v)\ :=\ -\frac{1}{2i}\sum_{z_k\in\Dm(G):\ z_k\sim v}F(z_k)
\overline{d\omega((vz_k)^*)}.
\]
\end{definition}
\begin{remark} \label{rem:opao-isorad}
In the isoradial context, the operator~$\opao$ (and similarly~$\pao$) coincides with the standard discrete Cauchy-Riemann operator~$\opa_{\iso}^{\Dm}$ {(e.g., see~\cite[Eq.~(3.2)]{ChSmi2} for the definition)} \emph{up to the sign} $-\Id^\bullet+\Id^\circ$ and up to a positive factor of order~$\delta^{-1}$ coming from the area of the cell at~$v$. The mismatch in the signs on~$G^\bullet$ and~$G^\circ$ is caused by the fact that we orient all the edges~$c$ of the s-embedding from~$v^\circ(c)$ to~$v^\bullet(c)$ in the definition of~$\alpha(c)$.
\end{remark}
Recall that~$r_z$ is the radius of the circle inscribed into a quad~$\cS^\dm(z)$ given by~\eqref{eq:rz=} and that~$\mu_z$ is the pre-factor in the Definition~\ref{def:opaS} of the operator~$\opaS$. The next lemma provides a link between the operators~$\opaS$, $\paS$ and~$\opao$, $\pao$.
\begin{lemma} \label{lem:pao-top=opaS} The following identities hold:
\[
\pao^\top=4U^{-1}R\opaS,\qquad \opao^{\,\top}=4\overline{U}{}^{-1}R\paS,
\]
where~$R=\diag(r_z)_{z\in\Dm(G)}$ and~$U=\diag(\mu_z)_{z\in\Dm(G)}$.
\end{lemma}
\begin{proof} For~$z\sim v$, where $z\in\Dm(G)$ and $v\in\Lambda(G)$, denote
\[
D_{v,z}:=\frac{r_z}{\cS(v)-\cS(z)}\ \ \text{if}\ \ v\in G^\bullet\quad \text{and}\quad D_{v,z}:=-\frac{r_z}{\cS(v)-\cS(z)}\ \ \text{if}\ \ v\in G^\circ,
\]
these quantities are nothing but the matrix coefficients of the operator~$4U^{-1}R\opaS$. A simple computation shows that, in both cases,~$D_{v,z}$ can be also written as
\[
D_{v,z}=-\frac{1}{2i}\overline{d\omega((vz)^*)}\,,
\]
where the edge~$(vz)^*$ of the medial graph~$\Upsilon(G)$ is oriented so that~$v$ is {on} the left.
Therefore, the matrix entries of the operators~$4U^{-1}R\opaS$ and~$\pao^\top$ coincide. The second identity can be obtained by taking the complex conjugation.
\end{proof}

The following operator~$\DS$ appeared in~\cite[Section~6]{Ch-ICM18} under the name \emph{s-Laplacian}.
\begin{proposition} \label{prop:DeltaS} The following identities hold:
\[
\DS\ :=\ -4\pao\opaS\ =\ -4\opao\paS\ =\ \DS^\top.
\]
Moreover, for~$v^\bullet\in G^\bullet$ and~$v^\circ\in G^\circ$, respectively, one has
\begin{align*}
[\DS H](v^\bullet)\ =\ &\textstyle \sum_{v_k^\bullet\sim v^\bullet}a_{v^\bullet v_k^\bullet}(H(v_k^\bullet)-H(v^\bullet))+
\sum_{v^\circ_k\sim v^\bullet}b_{v^\bullet v_k^\circ}(H(v_k^\circ)-H(v^\bullet)),\\
[\DS H](v^\circ)\ =\ &\textstyle \sum_{v_k^\bullet\sim v^\circ}b_{v^\circ v_k^\bullet}(H(v_k^\bullet)-H(v^\circ))-
\sum_{v^\circ_k\sim v^\circ}a_{v^\circ v_k^\circ}(H(v_k^\circ)-H(v^\circ)),
\end{align*}
with symmetric coefficients~$a_{vv'}=a_{v'v}>0$ and $b_{vv'}=b_{v'v}\in\R$. In particular, $a_{v_0^\bullet v_1^\bullet}=r_z^{-1}\sin^2\theta_z$ and $a_{v_0^\circ v_1^\circ}=r_z^{-1}\cos^2\theta_z$ for a quad~$z=(v_0^\bullet v_0^\circ v_1^\bullet v_1^\circ)$.
\end{proposition}

\begin{remark} (i) Explicit formulas expressing~$b_{vv'}$ via geometric characteristics of a tangential quad~$\cS^\dm(z)$ can be found in~\cite[Section~6]{Ch-ICM18}.

\smallskip

\noindent (ii) In the isoradial context, the s-Laplacian~$\Delta_\cS$ discussed above coincides with the operator $(\Id^\bullet-\Id^\circ)\Delta_{\iso}$ up to positive factors of order~$\delta$ coming from the areas of cells. In particular, in this context the coefficients~$b_{vv'}$ vanish. The mismatch of the signs on~$G^\bullet$ and~$G^\circ$ comes from Remark~\ref{rem:opao-isorad}; we believe that this mismatch is an intrinsic \emph{feature} of the formalism developed in this section.
\end{remark}

\begin{proof}[Proof of Proposition~\ref{prop:DeltaS}] Since~$\opaS 1=0$, it is clear that the operator $\DS:=-4\pao\opaS$ can be written in the form given above with \emph{some} (a priori, neither purely real nor symmetric) coefficients~$a_{vv'}$ and~$b_{vv'}$. Using Lemma~\ref{lem:pao-top=opaS} it is easy to see that
\[
a_{vv'}=-\frac{\mu_zr_z}{(\cS(v')-\cS(z))(\cS(v)-\cS(z))}=a_{v'v}
\]
if $v$ and $v'$ are two opposite vertices of a quad~$z$, and that
\[
b_{vv'}=\frac{\mu_zr_z}{(\cS(v')-\cS(z))(\cS(v)-\cS(z))}+\frac{\mu_{z'}r_{z'}}{(\cS(v')-\cS(z'))(\cS(v)-\cS(z'))}=b_{v'v}
\]
if~$(vv')=(zz')^*$ is an edge of the s-embedding. In particular,~$\Delta_\cS=\Delta_\cS^\top$.

Adopting the notation from Section~\ref{sub:semb-definition} and using the identity~$\opaS[\overline{\cS}(\cdot)-\overline{\cS}(z)]=1$ and~\eqref{eq:cS(z)-def},~\eqref{eq:rz=} one obtains the following expression:
\begin{align*}
\mu_z^{-1}\ &=\ \frac{1}{4}\biggl(\frac{\overline{\cX(c_{11})}}{\cX(c_{11})}-\frac{\overline{\cX(c_{00})}}{\cX(c_{00})}\biggr) \biggl(\frac{\overline{\cX(c_{10})}}{\cX(c_{10})}-\frac{\overline{\cX(c_{01})}}{\cX(c_{01})}\biggr)\\
&=\ \frac{\Im[\,\cX(c_{00})\overline{\cX(c_{11})}\,]\cdot \Im[\,\cX(c_{01})\overline{\cX(c_{10})}\,]}{\cX(c_{00})\cX(c_{01})\cX(c_{10})\cX(c_{11})}\ =\ \frac{r_z^2(\sin\theta_z\cos\theta_z)^{-2}}{\cX(c_{00})\cX(c_{01})\cX(c_{10})\cX(c_{11})}\,.
\end{align*}
Using~\eqref{eq:cS(z)-def} again, we get the required formulas for~$a_{v_0^\bullet v_1^\bullet}$ and $a_{v_0^\circ v_1^\circ}$. In particular, these coefficients are purely real.

It remains to check that the coefficients~$b_{vv'}$ are also purely real. For this purpose note that computations similar to those given above imply that
\[
\Im\biggl[\frac{\mu_zr_z}{(\cS(v_0^\bullet)-\cS(z))(\cS(v_0^\circ)-\cS(z))}\biggr]= -\frac{\sin\theta_z\cos\theta_z}{r_z}\cdot\Im\biggl[\frac{\cX(c_{11})}{\cX(c_{00})}\biggr]=
-\frac{1}{|\cX(c_{00})|^2}.
\]
A similar formula holds true for other corners of the quad~$z=(v_0^\bullet v_0^\circ v_1^\bullet v_1^\circ)$, with the same sign for~$c_{11}$ and with the \emph{opposite} one for~$c_{01}$ and $c_{10}$. As~$b_{vv'}$ is equal to the sum of two expressions coming from two quads adjacent to the edge~$(vv')$ of the s-embedding, this implies~$\Im[b_{vv'}]=0$.

To summarize, the operator~$\DS:=-4\pao\opaS$ is symmetric ($\DS=\DS^\top$) and has real coefficients. Therefore, one also has~$\DS=-4\opao\paS$.
\end{proof}

The following generalization of the miraculous Laplacian positivity property -- which was first discovered by Smirnov~\cite[Lemma~3.8]{Smi-I} in the square grid context -- appeared in~\cite[Section~6]{Ch-ICM18}.

\begin{corollary}[{\bf s-positivity property of functions~$\bm{H_F}$}]\label{cor:s-positivity} Assume that a function~$H_F$ is constructed from an s-holomorphic function~$F$ as in Section~\ref{sub:HF-def}. Then, the inequality~$[\DS H_F](v)\ge 0$ holds for all bulk vertices~$v\in \Lambda(G)$.

Moreover, under the assumption~\Unif, the quantity~$[\DS H_F](v)$ is a quadratic form in the values~$F(z_k)$ at quads~$z_k\in\Dm(G)$ adjacent to~$v\in\Lambda(G)$, which has~$O(1)$ coefficients and vanishes if the values~$F(z_k)$ are equal to each other for all~$z_k\sim v$.
\end{corollary}
\begin{proof} It directly follows from the identity~$\DS=-4\opao\paS$ and Corollary~\ref{cor:opaS-HF} that
\[
[\DS H_F](v)=-i\opao[F^2](v)\,.
\]
Recall that the operator~$\opao$ depends only on the \emph{directions} of the edges of the s-embedding adjacent to the vertex~$\cS(v)$ and that the s-holomorphicity property of the function~$F$ is also formulated in terms of these directions only. Therefore, the required claim follows from~\cite[Propostion~3.6]{ChSmi2}, which provides the s-positivity property in the isoradial context. (To check the signs, recall that, up to positive multiples, the operators $\pao$ and $\DS$ coincide with $(-\Id^\bullet+\Id^\circ)\pa^\Dm_{\iso}$ and $(\Id^\bullet-\Id^\circ)\Delta_{\iso}$, respectively, in the isoradial context.) In particular, the quadratic form mentioned above is nothing but the form~$Q^{(n)}$ from the proof of~\cite[Proposition~3.6]{ChSmi2}.
\end{proof}

\subsection{Coefficients~${A_z,B_z,C_z}$ and an approximation of the Laplacian by~$\cQ\Delta_\cS$ under the assumption~$\Qflat$} \label{sub:ABC}
We begin this section by introducing additional notation related to the differential operators~$\cQ\pao$, $\cQ\opao$, which we heavily use in the forthcoming Section~\ref{sec:convergence}. We then prove an important approximation property (Proposition~\ref{prop:DS-approximate}) of the s-Laplacian which shows that, \emph{under the assumption~\Qflat,} the coarse-grained operator~$\cQ\DS$ can be viewed as a discrete approximation to the standard {(continuous)} Laplacian~$\Delta=\pa_{xx}+\pa_{yy}$.

\smallskip

For a (bulk) quad~$z\in\Dm(G)$, denote
\begin{equation}\label{eq:ABC-def}
A_z:=[\pao^\top(\cQ\cS)](z),\qquad B_z:=[\pao^\top(\cQ^2)](z),\qquad C_z:=[\pao^\top(\cQ\overline{\cS})](z).
\end{equation}

\begin{remark} \label{rem:Cz-shift}
Since~$\opaS\cS=\opaS\cQ=0$ and~$\pao^\top=4U^{-1}R\opaS$ (see Lemma~\ref{lem:pao-top=opaS}), the coefficients~$A_z$ and~$B_z$ do not depend on the choice of the function~$\cQ$. However, this is \emph{not} true for the coefficient~$C_z$: it shifts by~$4\mu_z^{-1}r_z c$ if~$\cQ$ is replaced by~$\cQ+c$.
\end{remark}

\begin{lemma} (i) The coefficient~$A_z$ equals the area of the tangential quad~$\cS^\dm(z)$.

\smallskip

\noindent (ii) Under assumptions~\Unif\ and~\Qflat, one has~$A_z,B_z,C_z=O(\delta^2)$.
\end{lemma}
\begin{proof} (i) It follows from Lemma~\ref{lem:pao-top=opaS} and from the identity~$\opaS\cS=0$ that
\[
A_z\ =\ 4\mu_z^{-1}r_z\opaS[\cQ(\cdot)(\cS(\cdot)-\cS(z))](z)\ =\ r_z\cdot(\cQ(v_0^\bullet)+\cQ(v_1^\bullet)-\cQ(v_0^\circ)-\cQ(v_1^\circ)).
\]
By definition of the function~$\cQ$, this expression equals the product of~$r_z$ and the half-perimeter of the tangential quad~$\cS^\dm(z)$, i.e., to the area of this quad.
\smallskip

(ii) The estimate~$A_z=\Area(\cS^\dm(z))=O(\delta^2)$ is a triviality. The estimate \mbox{$B_z=O(\delta^2)$} follows from~\Unif\ since~$\pao^\top[\cQ^2]=\pao^\top[(\cQ(\cdot)-\cQ(z))^2]$ and the coefficients of the operator~$\pao^\top$ are of order~$1$. The last estimate~$C_z=O(\delta^2)$ follows in the same way using~\Qflat\ (cf. Remark~\ref{rem:Cz-shift}).
\end{proof}

In Section~\ref{sec:convergence} we also need the following analogue of Corollary~\ref{cor:opaS-HF} and~Lemma~\ref{lem:opaS=} for the operators~$(\cQ\pao)^\top$ and~$(\cQ\opao)^\top$\!, mentioned here for reference purposes.

\begin{lemma}\label{lem:paoQ=}
(i) Let a function~$H_F$ be constructed from an s-holomorphic function~$F$ as in Section~\ref{sub:HF-def}; see~\eqref{eq:HF-def}. Then,
\[
[\pao^\top(\cQ H_F)](z)=\tfrac{1}{4i}A_z\cdot(F(z))^2+\tfrac{1}{2}B_z\cdot |F(z)|^2-\tfrac{1}{4i}C_z\cdot (\overline{F(z)})^2.
\]

\smallskip

(ii) Under the assumptions~\Unif\ and~\Qflat, for each~$C^2$--smooth function~$\phi$ defined on the quad~$\cS^\dm(z)$ one has
\[
[\opao^\top(\cQ \phi)](z)=A_z\cdot [\opa\phi](z)+\overline{C}_z\cdot [\pa\phi](z)+O(\delta^3\cdot\max\nolimits_{\cS^\dm(z)}|D^2\phi|).
\]

\smallskip

(Similar statements are fulfilled for~$[\opao^\top(\cQ H_F)](z)$ and~$[\pao^\top(\cQ\phi)](z)$.)
\end{lemma}

\begin{proof} The proof of (i) is similar to the proof of Corollary~\ref{cor:opaS-HF}. To prove~(ii), note that~$\phi(\cdot)=\phi(z)+\overline{\cS}(\cdot)[\opa\phi](z)+\cS(\cdot)[\pa\phi](z)+O(\delta^2\max_{\cS^\dm(z)})|D^2\phi|$ and that the coefficients of the operator~$\opao^{\top}\cQ$ are of order~$O(|\cQ|)=O(\delta)$.
\end{proof}

Though the coefficients~$B_z,C_z$, even under assumptions~\Unif\ and~\Qflat, do not admit better \emph{pointwise} estimates than~$O(\delta^2)$, we now show that their oscillations lead to better a priori bounds for their \emph{averages} (with respect to the counting measure) over regions of the size bigger than~$\delta$; see also Remark~\ref{rem:BC=0-periodic}.

\begin{proposition} \label{prop:BC-average}
Let~$P$ be a square of size~$\ell\times\ell$ with~$\ell\ge\delta$, drawn over an s-embedding~$\cS$ satisfying the assumptions~\Unif\ and~\Qflat. Then,
\[
\sum_{z\in\Dm(G):\,\cS(z)\in P}B_z= O(\delta\ell)\qquad \text{and}\quad \sum_{z\in\Dm(G):\,\cS(z)\in P}C_z= O(\delta\ell).
\]
\end{proposition}
\begin{proof} For~$z\sim v$, denote by~$D_{v,z}:=\pm r_z/(\cS(v)-\cS(z))=O(1)$ the coefficient (matrix entry) of the operator~$\pao$ discussed in Lemma~\ref{lem:pao-top=opaS}; recall that one has the `$+$' sign in~$D_{v,z}$ if~$v\in G^\bullet$ and the `$-$' sign if~$v\in G^\circ$. Denote by \mbox{$P^\dm:=\cup_{z:\cS(z)\in P}\cS^\dm(z)$} the union of quads~$\cS^\dm(z)\subset\C$ with centers~$\cS(z)\in P$. Since~$\pao 1=0$, we have
\[
\sum_{z:\,\cS(z)\in P} B_z\ =\sum_{z:\,\cS(z)\in \Int P^\dm}[\pao^\top(\cQ^2)](z) \ =\sum_{z\sim v:\,\cS(z)\in \Int P^\dm,\,\cS(v)\in \pa P^\dm}D_{v,z}(\cQ(v))^2,
\]
where, as usual, we assume that~$z\in\Dm(G)$ and~$v\in\Lambda(G)$. The last sum contains~$O(\delta^{-1}\ell)$ terms of order~$O(\delta^2)$, which proves the required estimate for~$B_z$.

The analysis of the coefficients~$C_z$ is slightly more subtle. As above, the ``discrete integration by parts formula'' leads to the identity
\[
\sum_{z:\,\cS(z)\in P} C_z\ =\sum_{z:\,\cS(z)\in \Int P^\dm}[\pao^\top(\cQ\overline{\cS})](z) \ =\sum_{z\sim v:\,\cS(z)\in \Int P^\dm,\,\cS(v)\in \pa P^\dm}D_{v,z}\cQ(v)\overline{\cS(v)}.
\]
 For a vertex~$v\in\Lambda(G)$ such that $\cS(v)\in\pa P^\dm$, denote by~$v^\pm\in\Lambda(G)$ the preceding and the next vertex at the boundary of~$P^\dm$, when this boundary is tracked {in} the counterclockwise direction. Due to Definition~\ref{def:opao} and Remark~\ref{rem:omega=S-S/Q-Q}, for each such a boundary vertex of~$P^\dm$ one has the following identity:
 \[
\sum_{z:\,z\sim v,\, \cS(z)\in\Int P^\dm}D_{v,z}\ 
=\ -\frac{1}{2i}\cdot\biggl(\frac{\overline{\cS(v^-)}-\overline{\cS(v)}}{\cQ(v^-)-\cQ(v)} -\frac{\overline{\cS(v^+)}-\overline{\cS(v)}}{\cQ(v^+)-\cQ(v)}\biggr).
 \]
Therefore,
\begin{align*}
\sum_{z:\,\cS(z)\in P} C_z\ &=\ -\frac{1}{2i}\sum_{v:\,\cS(v)\in \pa P^\dm} \frac{\overline{\cS(v^+)}-\overline{\cS(v)}}{\cQ(v^+)-\cQ(v)}\cdot\big(\cQ(v^+)\overline{\cS(v^+)}-\cQ(v)\overline{\cS(v)}\big)\\
&=\ -\frac{1}{2i}\biggl[\ \sum_{v:\,\cS(v)\in \pa P^\dm}\big(\overline{\cS(v^+)}-\overline{\cS(v)}\big)\cdot\tfrac{1}{2}\big(\overline{\cS(v^+)}+\overline{\cS(v)}\big)\\
&\qquad\qquad  +\!\!\sum_{v:\,\cS(v)\in \pa P^\dm}\frac{(\overline{\cS(v^+)}-\overline{\cS(v)})^2}{\cQ(v^+)-\cQ(v)}\cdot\tfrac{1}{2}\big(\cQ(v^+)+\cQ(v)\big)\biggr].
\end{align*}
Note that the first sum vanishes since~$-\frac{1}{2i}\oint_{\pa P^\dm} \overline{z}d\overline{z}=0$. Finally, under the assumptions~\Unif\ and~\Qflat, the second sum contains~$O(\delta^{-1}\ell)$ terms of order~$O(\delta^2)$ and thus produces a total error~$O(\delta\ell)$ as claimed.
\end{proof}

\begin{remark} \label{rem:BC=0-periodic} If the graph~$(G,x)$, its s-embedding~$\cS$ and the associated function~$\cQ$ are {doubly periodic} (i.e., if we work with the canonical periodic embedding of~$(G,x)$ in the sense of Lemma~\ref{lem:double-periodic}), then the coefficients~$A_z,B_z,C_z$ are {doubly periodic} too. (This is a triviality for~$A_z,B_z$ and easily follows from the identity~$\opaS\cQ=0$ for~$C_z$.) In this situation Proposition~3.11 implies that the sum of the coefficients~$B_z$ over the fundamental domain vanishes and so does the sum of~$C_z$.
\end{remark}
We conclude this section by showing that a coarse-graining (with respect to the counting measure) of the operator~$\cQ\DS$ can be viewed as an approximation to the standard Laplacian in~$\C$. This observation is crucial for the analysis in Section~\ref{sec:convergence}. Let us emphasize that this fact requires a `flatness' assumption similar to~$\Qflat$ and does not hold true in general.

\begin{proposition}\label{prop:DS-approximate}
As above, let $P$ be a square of size~$\ell\times\ell$ with~$\ell\ge\delta$, drawn over an s-embedding~$\cS$ satisfying the assumptions~\Unif\ and~\Qflat. Then, for each~$C^3$--smooth function defined in a vicinity of $P$, the following holds:
\[
\ell^{-2}\hskip -16pt \sum_{v\in\Lambda(G):\,\cS(v)\in P}\cQ(v)[\DS\phi](v)\ =\ \Delta\phi+O\big(\delta\ell^{-1}\max_{P^\dm} |D^2\phi|\big)+O\big(\delta^{-2}\ell^3\max_{P^\dm} |D^3\phi|\big),
\]
where~$P^\dm\supset P$ denotes the union of all quads~$\cS^\dm(z)$ that have at least one vertex \mbox{$\cS(v)\in P$} and the Laplacian~$\Delta\phi=\phi_{xx}+\phi_{yy}$ can be evaluated at any point of~$P$ since the resulting error~$O(\ell\cdot \max_P|D^3\phi|)$ is absorbed by~$O(\delta^{-2}\ell^3\max_P|D^3\phi|)$.
\end{proposition}
\begin{proof} Due to the Taylor formula, each $C^3$--smooth function~$\phi$ on~$P^\dm$ can be written as a linear combination of the functions~$1$, $\cS$, $\overline{\cS}$, $\cS^2$, $|\cS|^2$, $\overline{\cS}{}^2$ up to an error $O(\ell^3\max\nolimits_{P^\dm}|D^3\phi|)$. Recall that~$\cQ=O(\delta)$, the coefficients of the s-Laplacian~$\DS$ are of order~$O(\delta^{-1})$, and there are~$O(\delta^{-2}\ell^2)$ terms in the sum. Therefore, the errors in the Taylor approximation produce the cumulative error~$O(\delta^{-2}\ell^3\max_{P^\dm}|D^3\phi|)$ in the statement.

\smallskip

It follows from Proposition~\ref{prop:DeltaS} that $\DS 1=\DS \cS=\DS \overline{\cS}=0$, thus linear terms~$1$, $\cS$, $\overline{\cS}$ of the Taylor expansion do not contribute to the sum. Hence, it remains to analyze the contributions of the quadratic terms~$\cS^2$, $|\cS|^2$ and~$\overline{\cS}{}^2$. To do this, first note that~$[\DS\cS^2](z)=[\DS(\cS(\cdot)-\cS(z))^2](z)=O(\delta)$ and similarly for~$\DS|\cS|^2$ and~$\DS\overline{\cS}{}^2$. Therefore, replacing the summation over vertices~$v\in \Lambda(G)$ such that~$\cS(v)\in P$ by the (a priori, larger) summation over~$v\in \Lambda(G)$ such that~$\cS(v)\in\Int P^\dm$, one adds no more than~$O(\delta^{-1}\ell)$ terms of order~$O(\delta^2)$, which can be absorbed in the error term~$O(\delta\ell^{-1}\max_{P^\dm}|D^2\phi|)$ in the final statement.

\smallskip To analyze the contribution of the term~$\cS^2$ (and, similarly, that of~$\overline{\cS}{}^2$), note that the identity~$\opaS\cQ=0$ implies~$\pao^\top\cQ=0$ and hence
\begin{align*}
\frac{1}{4}\sum_{v:\,\cS(v)\in \Int P^\dm}\cQ(v)[\DS\cS^2](v)\ &=\ \ -\!\!\!\!\sum_{v:\,\cS(v)\in \Int P^\dm}\cQ(v)[\pao\opaS\cS^2](v)\\
\\ & =\sum_{z\sim v:\,\cS(z)\in\Int P^\dm,\,\cS(v)\in\pa P^\dm}D_{v,z}\cQ(v)[\opaS\cS^2](z),
\end{align*}
where we use the same notation~$D_{v,z}=\pm r_z/(\cS(v)-\cS(z))=O(1)$ as in the proof of Proposition~\ref{prop:BC-average}.
Since $\cQ(v)=O(\delta)$, $[\opaS\cS^2](z)=[\opaS(\cS(\cdot)-\cS(z))^2](z)=O(\delta)$ and the sum along the boundary of~$P^\dm$ contains~$O(\delta^{-1}\ell)$ terms, its total contribution is~$O(\delta\ell^{-1}\max_{P^\dm}|D^2\phi|)$.

\smallskip

We now move on to the analysis of the contribution of the last remaining term~$|\cS|^2\cdot\pa\opa\phi=|\cS|^2\cdot \tfrac{1}{4}\Delta\phi$ in the Taylor expansion of the function~$\phi$ on~$P$. As above, we have the identity
\[
\frac{1}{4}\sum_{v:\,\cS(v)\in \Int P^\dm}\cQ(v)[\DS|\cS|^2](v)\ =\sum_{z\sim v:\,\cS(z)\in\Int P^\dm,\,\cS(v)\in\pa P^\dm}D_{v,z}\cQ(v)[\opaS|\cS|^2](z),
\]
Note that~$[\opaS|\cS|^2](z)=[\opaS|\cS(\cdot)-\cS(v)|^2](z)+\cS(v)=O(\delta)+\cS(v)$ for each of the vertices~$v\sim z$. Therefore, it remains to handle the sum
\[
\sum_{z\sim v:\,\cS(z)\in\Int P^\dm,\,\cS(v)\in\pa P^\dm}D_{v,z}\cQ(v)\cS(v),
\]
which can be done similarly to the proof of Proposition~\ref{prop:BC-average}. This gives the identity
\[
\frac{1}{4}\sum_{v:\,\cS(v)\in \Int P^\dm}\cQ(v)[\DS|\cS|^2](v)\ = -\frac{1}{2i}\oint_{\pa P^\dm}zd\overline{z}+O(\delta\ell)
\]
and we conclude by noting that~$-\frac{1}{2i}\oint_{\pa P^\dm}zd\overline{z}=\Area(P^\dm)=\ell^2+O(\delta\ell)$.
\end{proof}

\section{Convergence of the FK-Dobrushin observable}\label{sec:convergence}
\setcounter{equation}{0}

\subsection{{Strategy of the proof}}\label{sub:HF-approx-harm}
{As it has already been sketched in the introduction, the key ingredient of our proof of Theorem~\ref{thm:FK-conv} is a careful analysis} of functions~$H_F$ everywhere in~$\Od$ except a very thin (width~$\delta^{1-\eta}$) layer near the boundary~$\pa\Od$. Note that such an analysis is much more delicate than the {a priori regularity and the approximate harmonicity of~$H_F$} in the bulk of~$\Od$ (i.e., $O(1)$-away from~$\pa\Od$), which {are} relatively cheap {statements (see Corollary~\ref{cor:H-Lip} and Proposition~\ref{prop:f-hol}).}

Let us emphasize that the setup of Theorem~\ref{thm:HF-almost-harm} is a \emph{fixed} discrete domain~$\Od$, rather than a sequence of those with~$\delta\to 0$; the estimate~\eqref{eq:HF-almost-harm} is uniform with respect to~$\Od$ and~$\delta$ (provided that~$\diam\Od\le\cst$ and~$\delta$ is small enough, depending on~$\eta$ and constants in the assumptions~\Unif\ and~\Qflat). Theorem~\ref{thm:HF-almost-harm} is one of the central results of our paper; note that the uniform \emph{near-to-boundary} estimate~\eqref{eq:HF-almost-harm} is new even in the well-studied isoradial context.

\begin{theorem}\label{thm:HF-almost-harm}
Let $\Od\subset\C$ be a bounded simply connected discrete domain drawn on an s-embedding~$\cS^\delta$ satisfying assumptions~\Unif\ and~\Qflat. Assume that~$F$ is an s-holomorphic function in (the bulk of)~$\Od$ and that $|H_F|\le 1$ in~$\Od$, where the function~$H_F$ is constructed from~$F$ as in Section~\ref{sub:HF-def}.

Let~$\eta\in (0,1)$ and~$\Od_\intr\subset\C$ be (one of the connected components of) the~$\delta^{1-\eta}$-interior of~$\Od$. Denote by~$h_\intr$ the harmonic continuation of the function~$H_F$ from the boundary to the bulk of the domain~$\Od_\intr$ (i.e., $h_\intr$ is the solution of the \emph{continuous} Dirichlet problem in~$\Od_\intr$ with the boundary values given by~$H_F$).

There exists an exponent~$\alpha(\eta)>0$ such that, provided that~$\delta$ is small enough (depending only on~$\eta$ and constants in~\Unif, {\Qflat}), the following estimate holds:
\begin{equation}
\label{eq:HF-almost-harm}
|H_F-h_\intr|\ =\ O(\delta^{\alpha(\eta)})\ \ \text{uniformly~in}\ \ \Od_\intr\,,
\end{equation}
where the implicit constant in the O-estimate depends only on {$\eta$,} the diameter of~$\Od$ and constants in the assumptions~\Unif, \Qflat.
\end{theorem}

The strategy of the proof of Theorem~\ref{thm:HF-almost-harm} is described below. Let us emphasize that the central part of the argument is the estimate~\eqref{eq:Delta-key-estimate}, the proof of which is postponed until Section~\ref{sub:Delta-key-estimate}. Let~$\phi_0\in C_0^\infty(\C)$ be a (fixed once {and} forever) positive symmetric function such that~$\phi_0(u)=0$ for~$|u|\ge \tfrac{1}{2}$ and  $\int_\C\phi_0(u)dA(u)=1$.
\begin{proof}[{\bf Proof of Theorem~\ref{thm:HF-almost-harm} modulo the key estimate~\eqref{eq:Delta-key-estimate}}] Let~$0<\varepsilon\ll \eta$ be a small parameter to be fixed later and denote
\[
d_u:=\dist(u,\pa\Od),\quad \text{and}\quad \rho_u:=\delta^\varepsilon \crad(u,\Od)\asymp \delta^\varepsilon d_u\gg\delta\quad \text{for} \ \ u\in\Od_\intr\,.
\]
Above,~$\crad(u,\Od)$ is the \emph{conformal radius} of the point~$u$ in the domain~$\Od$ (which should be viewed as a polygon in~$\C$): if~$\varphi_u:\mathbb{D}\to\Od$ is a conformal mapping such that~$\varphi_u(0)=u$, then~$\crad(u,\Od)=|\varphi_u'(0)|$.
\begin{remark}
The advantage of using the conformal radius instead of the distance~$d_u$ is that the function~$u\mapsto\crad(u,\Od)$ is \emph{smooth}, its gradient is uniformly bounded and its second derivative is bounded by~$O(d_u^{-1})$, with absolute (i.e., independent of~$u$ and~$\Od$) constants; e.g., see Lemma~\ref{lem:crad-estimates}.
\end{remark}

We now introduce a \emph{running mollifier} $\phi(w,u)$ by setting
\[
\phi(w,u):=\rho_u^{-2}\phi_0(\rho_u^{-1}(w-u)),\quad u\in\Od_\intr,
\]
and consider the mollified function
\begin{equation}
\label{eq:tHF-def}
\widetilde{H}_F(u):=\int_{B(u,\rho_u)}\phi(w,u)H_F(w)dA(w),
\end{equation}
where the function~$H_F$ is {thought} to be continued from {(the images of) vertices \mbox{$\Lambda(G)\cup\Dm(G)$}} in the s-embedding to~$\Od_\intr\subset\C$ in a piecewise {affine way; see~\eqref{eq:HF-def}.} 

It easily follows from the a priori regularity of functions~$H_F$ {(see Corollary~\ref{cor:H-Lip}) that~$|H_F(w)-H_F(u)|=O(|w-u|\cdot d_u^{-1})$ for all~$w$ such that~$|w-u|\le \rho_u$. Therefore,}
\begin{equation}
\label{eq:x1}
\big|\widetilde{H}_F(u)-H_F(u)\big|=O(\rho_u\cdot d_u^{-1})=O(\delta^\varepsilon)\quad \text{uniformly in}~\Od_\intr.
\end{equation}
Now assume that we are able to prove a uniform estimate of the form
\begin{equation}
\label{eq:Delta-key-estimate}
|\Delta\widetilde{H}_F(u)|=O(\delta^pd_u^{-2-q})+O(\delta^{1-s}d_u^{-3})\quad \text{for~all}\ \ u\in\Od_\intr,
\end{equation}
with some exponents~$p,q,s\ge 0$ satisfying the inequalities
\[
p>q(1-\eta)\quad\text{and}\quad s<\eta.
\]
Since $d_u\ge\delta^{1-\eta}$ for all~$w\in\Od_\intr$, in this situation one can find $\alpha=\alpha(\eta,p,q,s)>0$ such that
\begin{equation}
\label{eq:x2}
|\Delta\widetilde{H}_F(u)|=O(\delta^\alpha d_u^{-2+\alpha})\ \ \text{uniformly~over}\ \ u\in\Od_\intr.
\end{equation}
Let~$\widetilde{h}$ denote the harmonic continuation of~$\widetilde{H}_F$ from~$\pa\Od_\intr$ to~$\Od_\intr$ (i.e., $\widetilde{h}$ is the solution of the \emph{continuous} Dirichlet problem in~$\Od_\intr$ with boundary values given by~$\widetilde{H}_F$).
Standard estimates (e.g., Lemma~\ref{lem:cont-estimate} applied to the function~$\widetilde{H}_F-\widetilde{h}$ in the $\delta^{1-\eta}$-interior~$\Od_\intr$ of the domain $\Omega^\delta$, viewed as a subset of~$\C$) allow to derive from~\eqref{eq:x2} that
\begin{equation}
\label{eq:x3}
|\widetilde{H}_F-\widetilde{h}|=O(\delta^\alpha)\ \ \text{uniformly~in}\ \ \Od_\intr,
\end{equation}
where an additional (compared to~\eqref{eq:x2}) factor in the~$O$-bound depends only on~$\alpha$ and~$\diam(\Od)$. By the maximum principle, the estimate~\eqref{eq:x1} also implies that \mbox{$|\widetilde{h}-h|=O(\delta^\varepsilon)$} in~$\Od_\intr$. Combining this with~\eqref{eq:x1} and~\eqref{eq:x3} we obtain the desired estimate~\eqref{eq:HF-almost-harm} with the exponent~$\min\{\alpha,\varepsilon\}>0$. The proof is complete.
\end{proof}

\subsection{Proof of the key estimate~\eqref{eq:Delta-key-estimate}} \label{sub:Delta-key-estimate}
In this section we prove the estimate~\eqref{eq:Delta-key-estimate} for the Laplacian of the mollified function~$\widetilde{H}_F$ given by~\eqref{eq:tHF-def}. Since~$\phi(\cdot,u)$ vanishes near the boundary of the disc~$B(u,\rho_u)$, we have
\begin{equation}
\label{eq:DHF-1}
\Delta\widetilde{H}_F(u)=\int_{B(u,\rho_u)}(\Delta_u\phi(w,u))H_F(w)dA(w),\quad u\in\Od_\intr.
\end{equation}
Recall that {we consider} the mollifier~$\phi(w,u)=\rho_u^{-2}\phi_0(\rho_u^{-1}(w-u))$ {with}~$\rho_u\asymp\delta^\varepsilon d_u$ and that~$\varepsilon>0$ is a small parameter which will be chosen at the end of the proof. Recall also that the function~${u}\mapsto\rho_u$ is smooth, its gradient is or order~$O(\delta^\varepsilon)$, and its second derivatives are of order~$O(\delta^\varepsilon d_u^{-1})=O(\delta^{2\varepsilon}\rho_u^{-1})$; see Lemma~\ref{lem:crad-estimates}.

Let us begin by stating a \emph{trivial} bound
\[
|\Delta\widetilde{H}(u)|\ =\ O(\rho_u^{-4}\cdot\rho_u^2) = O(\rho_u^{-2})=O(\delta^{-2\varepsilon}d_u^{-2}),
\]
where the first factor comes from the estimate~$|\Delta_u\phi(w,u)|=O(\rho_u^{-4})$ and the second is the area of the disc~$B(u,\rho_u)$. To get~\eqref{eq:Delta-key-estimate} instead of this trivial bound we need to \emph{improve} it by factors
\begin{equation}
\label{eq:error-wanted}
O(\delta^{p+2\varepsilon}d_u^{-q})\quad \text{or}\quad O(\delta^{1-s+2\varepsilon}d_u^{-1})\quad \text{with}\quad p>q(1-\eta)\quad \text{and}\quad s<\eta.
\end{equation}
The following analysis consists of seven steps: {in} each of these steps (except the final one) we modify the expression~\eqref{eq:DHF-1} and check that all the errors appearing along the way fit the condition~\eqref{eq:error-wanted}.

\medskip

\noindent {\bf Step~1.} \emph{Replace~$\Delta_u\phi(w,u)$ by~$\Delta_w\phi(w,u)$.}

\medskip

To control the error, note that the a priori regularity of s-holomorphic functions (see Theorems~\ref{thm:F-Hol} and \ref{thm:F-via-HF}(i)) imply that, in the disc~$B(u,\rho_u)$, the function~$H_F$ can be approximated by a \emph{linear} function~$L$ up to~$O(\rho_u\cdot(\rho_u/d_u)^\beta d_u^{-1})=O(\delta^{\varepsilon(1+\beta)})$, where $\beta>0$ is the a priori H\"older exponent of~$F$. It is easy to see that
\begin{align*}
\textstyle \int_{B(u,\rho_u)}(\Delta_w\phi(w,u))L(w)dA(w)\ &=\ \textstyle \int_{B(u,\rho_u)}\phi(w,u)(\Delta L)(w)dA(w)\ =\ 0,\\
\textstyle \int_{B(u,\rho_u)}(\Delta_u\phi(w,u))L(w)dA(w)\ &=\ \Delta L(u)\ =\ 0,
\end{align*}
(for a symmetric mollifier~$\phi_0$) and
\[
|\Delta_w\phi(w,u)-\Delta_u\phi(w,u)|=O(\delta^\varepsilon\cdot\rho_u^{-4}),
\]
the additional factor~$\delta^{\varepsilon}$ appears since we should differentiate~$\rho_u$ at least once. Therefore,
\begin{equation}
\label{eq:DHF-2}
\Delta\widetilde{H}_F(u)\ =\ \int_{B(u,\rho_u)}(\Delta_{w}\phi(w,u))H_F(w)dA(w)+\ O(\delta^\varepsilon\cdot \delta^{\varepsilon(1+\beta)}\cdot\rho_u^{-2}),
\end{equation}
which always fits~\eqref{eq:error-wanted} (with~$p=\varepsilon\beta$ and~$q=0$).

\medskip

\noindent {\bf Step~2.} \emph{Use Proposition~\ref{prop:DS-approximate} to replace $\Delta_w\phi(w,u)$ by~$\cQ(v)[\DS\phi(\cdot,u)](v)$.}

\medskip
More precisely, {in} this step we aim to prove that
\begin{align}
\int_{B(u,\rho_u)}(\Delta_w\phi(w,u))&H_F(w)dA(w)\ \ =\!\!\sum_{v:\,\cS(v)\in B(u,\rho_u)}\cQ(v)[\DS\phi(\cdot,u)](v)H_F(v) \notag\\
&+\ \big(O(\delta^{1-\gamma}d_u^{-1})+O(\delta^\gamma)+O(\delta^{1-3\gamma-\varepsilon}d_u^{-1})\big)\cdot \rho_u^{-2},
\label{eq:DHF-3}
\end{align}
which fits~\eqref{eq:error-wanted} if the exponent~$\gamma$ is chosen so that~$\gamma>2\varepsilon$ and~$3(\gamma+\varepsilon)<\eta$. Note that one can always find such~$\gamma$ provided that~$\varepsilon<\frac{1}{9}\eta$.

\smallskip

To prove~\eqref{eq:DHF-3}, cover~$B(u,\rho_u)$ by squares of size~$\ell=\delta^{1-\gamma}\ll\delta^{1-\eta}\le\rho_u$, where~$\gamma$ is chosen as explained above. On each of these squares the function~$H_F$ can be approximated by a \emph{constant} (of order~$O(1)$) up to an error~$O(\ell d_u^{-1})$, which leads to the first error term in~\eqref{eq:DHF-3}. (Note also that the main sum in~\eqref{eq:DHF-3} admits a similar trivial bound~$O(\rho_u^{-2})$ as the left-hand side since it contains~$O(\delta^{-2}\rho_u^2)$ terms of order~$O(\delta\cdot \delta|D^2\phi(\cdot,u)|)=O(\delta^2\rho_u^{-4})$.) The second and the third error terms in~\eqref{eq:DHF-3} come from those in Proposition~\ref{prop:DS-approximate}: the additional factors compared to the trivial bound are~$O(\delta\ell^{-1})=O(\delta^\gamma)$ and~$O(\delta^{-2}\ell^3\rho_u^{-1})=O(\delta^{1-3\gamma-\varepsilon}d_u^{-1})$.

\medskip

\noindent {\bf Step~3.} \emph{Use the factorization of~$\DS$ and the ``discrete integration by parts'' formula.}

\medskip

This step is an algebraic manipulation with the main term {on} the right-hand side of~\eqref{eq:DHF-3}. Denote~$\phi_u(\cdot):=\phi(\cdot,u)$ for shortness. Since $\DS=-4\pao\opaS$, we have
\begin{align}
&\sum_{v:\,\cS(v)\in B(u,\rho_u)}\cQ(v)[\DS\phi_u](v)H_F(v)\ =\ -4\sum_{z:\,\cS(z)\in B(u,\rho_u)}(\pao^\top[\cQ H_F])(z)[\opaS\phi_u](z) \notag\\
&=\ -4\sum_{z:\,\cS(z)\in B(u,\rho_u)}(\tfrac{1}{4i}A_z(F(z))^2+\tfrac{1}{2}B_z|F(z)|^2-\tfrac{1}{4i}C_z(\overline{F(z)})^2)[\opaS\phi_u](z),
\label{eq:DHF-4}
\end{align}
recall that the coefficients~$A_z,B_z,C_z$ were introduced in Section~\ref{sub:ABC}.

\medskip

\noindent {\bf Step~4.} \emph{Estimate the contribution of the~$B_z$ terms in~\eqref{eq:DHF-4} via Proposition~\ref{prop:BC-average}.}

\medskip

As on Step~2, let us cover the disc~$B(u,\rho_u)$ by squares of size~$\ell=\delta^{1-\gamma}$. We claim that
\begin{equation}
\label{eq:DHF-5}
\sum_{z:\,\cS(z)\in B(u,\rho_u)}B_z|F(z)|^2[\opaS\phi_u](z) = O(\delta^\gamma+\delta^{\beta(1-\gamma)}d_u^{-\beta}+\delta^{1-\gamma-\varepsilon}d_u^{-1})\cdot\rho_u^{-2},
\end{equation}
where~$\beta>0$ is the a priori H\"older exponent of s-holomorphic functions provided by Theorem~\ref{thm:F-Hol}. If~$\gamma>2\varepsilon$, $\beta(\eta-\gamma)>2\varepsilon$ and~$\gamma+3\varepsilon<\eta$, then all the three terms {on} the right-hand side fit the condition~\eqref{eq:error-wanted}. Note that one can find an exponent~$\gamma$ satisfying the required inequalities provided that~$\varepsilon<\min\{\frac{1}{5},\tfrac{\beta}{2(1+\beta)}\}\cdot \eta$.

To prove~\eqref{eq:DHF-5}, note that the function~$|F|^2$ can be approximated by a constant (of order~$d_u^{-1}$) with an error~$O((\ell/d_u)^\beta d_u^{-1})=O(\delta^{\beta(1-\gamma)}d_u^{-1-\beta})$ on each of these squares of size~$\ell$. This gives the second term in~\eqref{eq:DHF-5}. Similarly, the third term comes from the fact that the function~$\opaS\phi_u=\opa\phi_u+O(\delta|D^2\phi_u|)=O(\rho_u^{-1})$ is constant on each of these squares up to an error~$O(\ell\rho_u^{-2})=O(\delta^{1-\gamma-\varepsilon}d_u^{-1}\cdot\rho_u^{-1})$. The first term~$O(\delta^\gamma)$ comes from Proposition~\ref{prop:BC-average} as the improvement of the trivial bound~$O(\ell^2)$ to~$O(\delta\ell)$ for the sum of the coefficients~$B_z$ over a square of size~$\ell$.

\medskip

\noindent {\bf Step~5.} \emph{Replace the~$C_z$ terms in~\eqref{eq:DHF-4} by their conjugate.}

\medskip

This can be done in two ways. First, one can repeat the arguments given in the previous step to prove that these terms produce the same type of negligible errors as the terms~$B_z$, and so do the conjugated ones. An alternative (and more conceptual) way is to recall that the initial expression~\eqref{eq:DHF-1} is \emph{purely real}, thus the imaginary part of~\eqref{eq:DHF-4} is irrelevant anyway. After this conjugation we arrive at the expression
\[
\Delta\widetilde{H}_F(u)\ =\ -4\!\!\!\sum_{z:\,\cS(z)\in B(u,\rho_u)}\big(A_z[\opaS\phi_u](z)+\overline{C}_z[\paS\phi_u](z)\big)[\paS H_F](z)+O(\ldots),
\]
where the~$O(\ldots)$ terms fit the condition~\eqref{eq:Delta-key-estimate}. (Recall that~$[\paS H_F](z)=\frac{1}{4i}(F(z))^2$.)

\medskip

\noindent {\bf Step~6.} \emph{Use Lemma~\ref{lem:paoQ=}(ii) and ``integrate by parts'' once more.}

\medskip

It follows from Lemma~\ref{lem:opaS=} and Lemma~\ref{lem:paoQ=} that
\[
A_z[\opaS\phi_u](z)+\overline{C}_z[\paS\phi_u](z)\ =\ [\opao^{\,\top}(\cQ\phi_u)](z)+O(\delta^3|D^2\phi_u|),
\]
As~$A_z,C_z=O(\delta^2)$, the error term produces an improvement~$O(\delta\rho_u^{-1})=O(\delta^{1-\varepsilon}d_u^{-1})$ compared to the trivial bound, which fits~\eqref{eq:error-wanted} provided that~$\varepsilon<\tfrac{1}{3}\eta$. Therefore, we arrive at the expression
\begin{align}
\Delta\widetilde{H}_F(u)\ &=\ -4\!\!\!\sum_{z:\,\cS(z)\in B(u,\rho_u)}[\opao^{\,\top}(\cQ\phi_u)](z)[\paS H_F](z)+O(\ldots) \notag\\
&=\sum_{v:\,\cS(v)\in B(u,\rho_u)}\cQ(v)\phi_u(v)[\DS H_F](v)+O(\ldots),\label{eq:DHF-6}
\end{align}
where, as above, $O(\ldots)$ stands for error terms satisfying the condition~\eqref{eq:Delta-key-estimate}.

\medskip

\noindent {\bf Step~7.} \emph{Use the s-positivity~$\DS H_F\ge 0$ and shifts of~$\cQ$ to neglect the sum in~\eqref{eq:DHF-6}.}

\medskip

Note that the function~$H_F$ and, in particular, the value~$\Delta\widetilde{H}_F(u)$ do \emph{not} depend on a particular choice of the function~$\cQ$, defined up to a global additive constant. The analysis performed above also does not rely upon such a choice, provided that the assumption~\Qflat\ holds. In particular, if one shifts this additive constant so that the function~$\cQ$ is sign-definite (i.e., either everywhere positive or everywhere negative but still satisfying the uniform bound~$\cQ=O(\delta)$), then the estimate~\eqref{eq:DHF-6} holds in \emph{both} cases~$\cQ\ge 0$ and~$\cQ\le 0$. Since~$\phi_u\ge 0$ and $[\DS H_F]\ge 0$ (see Corollary~\ref{cor:s-positivity} for the latter), the sum in~\eqref{eq:DHF-6} can be made both positive or negative by such shifts. Therefore, this sum must have the same order as the error term~$O(\ldots)$. This completes the proof of the estimate~\eqref{eq:Delta-key-estimate}.

In particular, the proof runs through if the parameter~$\varepsilon$ used in the definition of the running mollifier~$\phi_u$ is chosen smaller than~$\min\{\frac{1}{9},\frac{\beta}{2(1+\beta)}\}\cdot\eta$, where~$\beta$ is the a priori H\"older exponent of s-holomorphic functions provided by Theorem~\ref{thm:F-Hol}\qed

\begin{remark} (i) One can try to avoid the use of the s-positivity property in Step~7 by proving an accurate estimate for the values~$[\DS H_F](v)$. According to Corollary~\ref{cor:s-positivity}, each of them is a quadratic form in the variables~$F(z)=O(d_u^{-1/2})$ with~$O(1)$ coefficients, vanishing on constants. If we were able to prove the estimate
\begin{equation}
\label{eq:x-Holder}
|\DS H_F|(v)\ =\ O((\delta/d_u)^{2\beta}d_u^{-1})\ \ \text{with}\ \ \beta>\tfrac{1}{2},
\end{equation}
this would give a total error~$O(\delta^{-2}\rho_u^2\cdot\delta\cdot\rho_u^{-2}\cdot\delta^{2\beta}d_u^{-1-2\beta})=O(\delta^{2\beta-1}d_u^{-1-2\beta})$, which fits the condition~\eqref{eq:Delta-key-estimate}. In other words, an a priori~$(\tfrac{1}{2}+\varepsilon)$-H\"older regularity of \mbox{s-holomorphic} functions would be enough to provide a more straightforward argument in Step~7. However, at the moment we are not aware of methods allowing to prove such an a priori regularity statement without periodicity assumptions (and beyond the isoradial context, where~$\beta=1$ is known~\cite[Theorem~3.12]{ChSmi2}).

\smallskip

\noindent (ii) Vice versa, one can try to use the s-positivity in order to prove the \mbox{$(\tfrac{1}{2}+\varepsilon)$-H\"older} regularity of s-holomorphic functions. Indeed, \emph{a posteriori} we know that sign-definite terms in~\eqref{eq:DHF-6} do not produce a big sum. Loosely speaking, this means that the estimate~\eqref{eq:x-Holder} holds `in average' because of the s-positivity. One could then imagine a bootstrap-type argument that improves the $\beta$-H\"older regularity (with a small exponent~$\beta$) using this observation.
\end{remark}

\subsection{Proof of Theorem~\ref{thm:FK-conv}} \label{sub:proof-FK}
Let discrete domains~$(\Od;a^\delta,b^\delta)$ converge to~$(\Omega;a,b)$ as $\delta \to 0$ in the Carath\'eodory sense and s-holomorphic functions~$F^\delta$ be constructed from the Kadanoff--Ceva fermionic observables~\eqref{eq:KC-Dob-def} as in Proposition~\ref{prop:shol=3term}. Recall that the corresponding functions~$H_{F^\delta}$ satisfy the~$0/1$ Dirichlet boundary conditions~\eqref{eq:HF2-bc} and that~$H_{F^\delta}\in[0,1]$ everywhere in~$\Omega^\delta$ due to the maximum principle (see Proposition~\ref{prop:HF-max}).

It follows from Theorem~\ref{thm:F-Hol} and Theorem~\ref{thm:F-via-HF}(ii) that the families~$\{F^\delta\}$ and $\{H_{F^\delta}\}$ are precompact in the bulk of~$\Omega$. Therefore, passing to a subsequence~$\delta\to 0$, we can assume that
\[
\textstyle F^\delta\to f\quad\text{and}\quad H_{F^\delta}\to h=\tfrac{1}{2}\int \Im[f^2dz]\quad \text{uniformly~on~compact~subsets~of}~\Omega.
\]
It immediately follows from Proposition~\ref{prop:f-hol} that~$h:\Omega\to [0,1]$ is a harmonic function. Thus, the main difficulty is to prove that the $0/1$ Dirichlet boundary conditions of~$H_{F^\delta}$ survive in the limit~$\delta\to 0$, i.e., that~$h=\hm_\Omega(\cdot,(ba))$.

Let~$p_0>0$ be the uniform lower bound on the probability of {open} circuits in annuli~$\boxbox(z,d)$ from Corollary~\ref{cor:circuits} and~$\eta<-\log (1-p_0)/\log 3$ be a sufficiently small positive constant. Denote by~$\Od_\intr\subset\C$ the principal connected component of the \mbox{$\delta^{1-\eta}$-interior} of~$\Od$. As in Theorem~\ref{thm:HF-almost-harm}, let~$h^\delta_\intr$ be the solution to the \emph{continuous} Dirichlet boundary value problem in~$\Od_\intr$ such that~$h^\delta_{\intr}=H_{F^\delta}$ at~$\pa\Od_\intr$. Since the functions~$H_{F^\delta}$ and~$h^\delta_{\intr}$ are uniformly (with respect to both~$\Od$ and~$\delta$) close to each other due to Theorem~\ref{thm:HF-almost-harm}, it remains to prove that
\[
h^\delta_{\intr}\to\hm_\Omega(\cdot,(ba))\ \ \text{as}\ \ \delta\to 0.
\]

Note that we now work with \emph{continuous} harmonic functions, which are determined by their boundary values in a very stable way.  By construction of the function~$H_{F^\delta}$ (see~\eqref{eq:HF-def}), its gradient is bounded by~$|F^\delta|^2$. {Passing to the random cluster (or Fortuin--Kasteleyn) representation of the Ising model it is easy to see from Corollary~\ref{cor:circuits} that for each dual vertex~$v^\bullet\in G^\bullet=G^{\bullet,\delta}$ we have
\begin{align*}
\E[\mu_{v^\bullet}\mu_{(b^\delta a^\delta)^\bullet}]\ &=\ \mathbb{P}[\,\text{$\cS^\delta(v^\bullet)$ is dual-connected to $(b^\delta a^\delta)^\bullet$}\,]\ \\
&\le\ (1\!-\!p_0)^{\lfloor(\log\dist(\cS^\delta(v^\bullet),(b^\delta a^\delta)^\bullet)-\log(L_0\delta))/\log 3\rfloor}
\end{align*}
and hence the following `crude' estimate holds:}
\[
|F^\delta(z)|\ \le\ O(\delta^{-\frac{1}{2}})\cdot\E[\,\mu_{v^\bullet(z)}\mu_{(b^\delta a^\delta)^\bullet}\,]\ \le\
O\biggl(\delta^{-\frac{1}{2}}\biggl[\frac{\delta}{\dist(\cS^\delta(z),(b^\delta a^\delta))}\biggr]^\eta\,\biggr)\,.
\]
Therefore, we have the following uniform estimate on~$\pa\Od_\intr$:
\[
h^\delta_{\intr}(u)\ =\ H_{F^\delta}(u)\ =\ O\biggl(\frac{\delta^\eta}{(\dist_{\Od}(u,(b^\delta a^\delta))^{2\eta}}\biggr)\quad \text{if}\ \ \dist_{\Od}(u,(a^\delta b^\delta))=\delta^{1-\eta}
\]
(note that~$\dist_{\Od}(u,(b^\delta a^\delta))$ is of order~$1$ if~$u$ is close to~$(a^\delta b^\delta)$) and similarly
\[
h^\delta_{\intr}(u)\ =\ H_{F^\delta}(u)\ =\ 1-O\biggl(\frac{\delta^\eta}{(\dist_{\Od}(u,(a^\delta b^\delta))^{2\eta}}\biggr)\quad \text{if}\ \ \dist_{\Od}(u,(b^\delta a^\delta))=\delta^{1-\eta}.
\]
The convergence~$h^\delta_{\intr}\to\hm_\Omega(\cdot,(ba))$ now follows from standard arguments: e.g., applying the Beurling estimate for \emph{continuous} harmonic functions one immediately sees that~$h^\delta_{\intr}$ must be uniformly (in~$\delta$) close to~$0$ near the boundary arc~$(ab)$ and close to~$1$ near the complementary arc~$(ba)$. The proof is complete. \qed

\begin{remark}\label{rem:on-the-strategy} In our proof of Theorem~\ref{thm:FK-conv}, the key output of Theorem~\ref{thm:HF-almost-harm} is \emph{not} the approximate harmonicity of functions~$\widetilde{H}_F$ in the bulk of~$\Od$ (this follows from Proposition~\ref{prop:f-hol}) but a quantitative control of their near-to-the-boundary values, which relies upon Lemma~\ref{lem:cont-estimate} and the estimate~$|\Delta \widetilde{H}_F(u)|=O(d_u^{-2+\alpha})$ for~$d_u\ge \delta^{1-\eta}$. In particular, a similar scheme of controlling boundary conditions of functions~$h$ in rough domains could be used for near-critical fermionic observables; cf.~\cite{park-iso}.
\end{remark}

\section{Crossing estimates for the FK-Ising model}\label{sec:RSW}
\setcounter{equation}{0}

This section is mostly devoted to the proof of Theorem~\ref{thm:RSW-selfdual}; for completeness, in Section~\ref{sub:circuits} we also recall a (standard) derivation of Corollary~\ref{cor:circuits} from Theorem~\ref{thm:RSW-selfdual}.
The proof is based on the classical idea of Smirnov to use the s-holomorphicity of fermionic observables in order to derive the required crossing estimate; the overall strategy of such a proof is explained in Section~\ref{sub:strategy-RSW}. However, technical details of our realization of this idea differ from, e.g., the proof of~\cite[Theorem~6.1]{ChSmi2} 
due to the lack of comparison with discrete harmonic functions on~$\Lambda(G)$. Section~\ref{sub:cuts} is devoted to the construction of `straight cuts' that are needed to define discrete rectangles~$\cR^\delta$. Properties of fermionic observables in~$\cR^\delta$ are further discussed in Section~\ref{sub:bc-in-R}. Section~\ref{sub:visibility} contains technical estimates of hitting probabilities for random walks in~$\cR^\delta$ required for the proof of Theorem~\ref{thm:RSW-selfdual}. This proof is completed in Section~\ref{sub:RSW-proof}. {In order to lighten the exposition,} we begin our analysis by considering fermionic observables in discrete rectangles with \emph{Dobrushin} boundary conditions and only then pass to more general observables in rectangles with wired/free/wired/free {ones}.

A reader familiar, e.g., with the square grid case can go directly to Section~\ref{sub:RSW-proof} after Section~\ref{sub:strategy-RSW} in order to grasp the main idea of our strategy before diving into technical details of its realization.

\subsection{Fermionic observable with wired/free/wired/free boundary conditions and the strategy of the proof of Theorem~\ref{thm:RSW-selfdual}}
\label{sub:strategy-RSW}
We begin this section with a reminder of Smirnov's idea on the analysis of crossing probabilities in domains with wired/free/wired/free boundary conditions, implemented in~\cite[Section~6]{ChSmi2} in the isoradial context. In what follows we keep using the Kadanoff--Ceva formalism instead of the loop representation of the FK-model used in~\cite[Section~6]{ChSmi2}.

 Let~$(\Omega;a_1,b_1,a_2,b_2)$ be a simply connected discrete domain on an \mbox{s-embedding}~$\cS$ with wired boundary arcs~$(a_1b_1)^\circ$, $(a_2b_2)^\circ$ and dual-wired boundary arcs~$(b_1a_2)^\bullet$ and~$(b_2a_1)^\bullet$, the latter considered as a \emph{single} dual vertex. (In other words, originally we set the interaction parameter~$x_\mathrm{out}:=1$ in Fig.~\ref{fig:SDdomain}.)

\begin{figure}
\centering
\includegraphics[clip, trim=5.8cm 10.7cm 2.7cm 5.7cm, width=0.8\textwidth]{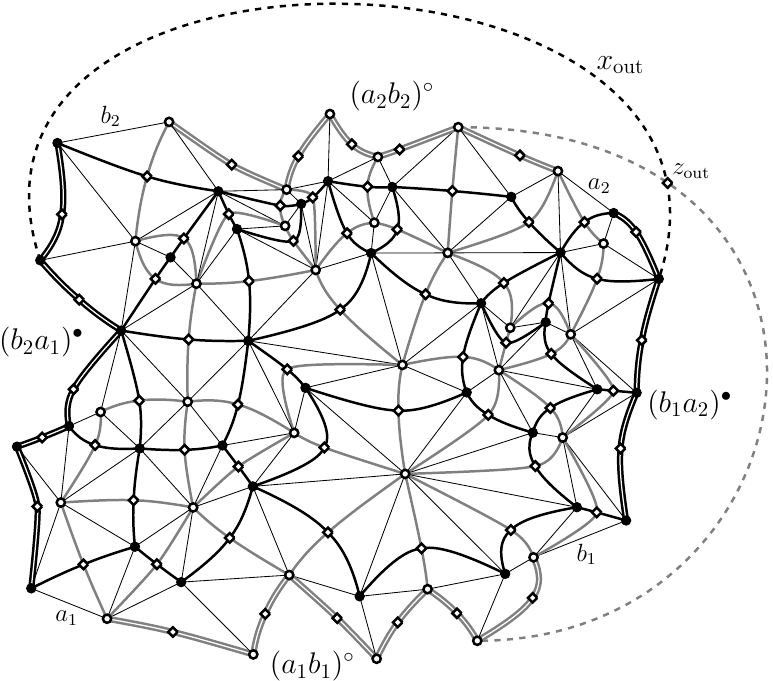}
\caption{An example of a domain with wired/free/wired/free boundary conditions drawn on an s-embedding. Note the external edge and external quad~$z_\mathrm{out}$ introduced so that the graphs $G^\bullet$ and~$G^\circ$ are dual to each other if viewed on the sphere. The parameter~$x_\mathrm{out}$ is originally set to~$1$ and then tuned in a special way.\label{fig:SDdomain}}
\end{figure}

 Denote~${\mathrm{p}}:=2\arctan \E[\sigma_{(a_1b_1)^\circ}\sigma_{(a_2b_2)^\circ}]$, where the expectation~$\E=\E^{(0)}$ is taken in the model described above. Now \emph{re-define}~$x_\mathrm{out}:=\tan\tfrac{1}{2}(\tfrac{\pi}{2}\!-\!{\mathrm{p}})$ and let~$\E^{({\mathrm{p}})}$ be the expectation in the new model with the nontrivial interaction constant between~$\sigma_{(a_1b_1)^\circ}$ and~$\sigma_{(a_2b_2)^\circ}$. It is easy to see that
 \begin{align*}
 \mathbb{P}^{({\mathrm{p}})}[\sigma_{(a_1b_1)^\circ}\!=\sigma_{(a_2b_2)^\circ}]\ &=\ \frac{\mathbb{P}[\sigma_{(a_1b_1)^\circ}\!=\sigma_{(a_2b_2)^\circ}]} {\mathbb{P}[\sigma_{(a_1b_1)^\circ}\!=\sigma_{(a_2b_2)^\circ}]+x_\mathrm{out}\cdot (1-\mathbb{P}[\sigma_{(a_1b_1)^\circ}\!=\sigma_{(a_2b_2)^\circ}])}\\
 &=\ \frac{\frac{1}{2}(1+\tan\frac{1}{2}{\mathrm{p}})}{\tfrac{1}{2}(1+\tan\frac{1}{2}{\mathrm{p}})+x_\mathrm{out}\cdot\frac{1}{2}(1-\tan\frac{1}{2}{\mathrm{p}})}\ =\ \tfrac{1}{2}(1+\sin{\mathrm{p}})\,.
 \end{align*}
 Let~$X(c):=\E^{({\mathrm{p}})}[\chi_c\mu_{(b_1a_2)^\bullet}\sigma_{(a_2b_2)^\circ}]$ be the Kadanoff--Ceva fermionic observable in the \emph{new} (i.e., with~$x_\mathrm{out}\ne {1}$ defined above) model. It is easy to see that
 \[
 X(a_2)=\pm 1\quad\text{and}\quad X(b_1)=\pm\E^{({\mathrm{p}})}[\sigma_{(a_1b_1)^\circ}\sigma_{(a_2b_2)^\circ}]=\pm\sin{\mathrm{p}}.
 \]
 The propagation equation~\eqref{eq:3-terms} applied around the \emph{external} quad~$z_\mathrm{out}$ yields
 \[
 X(b_2)=\pm\cos{\mathrm{p}}\qquad \text{and}\qquad X(a_1)=0,
 \]
 the latter fact is crucial for the further analysis {(see Remark~\ref{rem:4pt>0})} and justifies the concrete choice of the parameter~$x_\mathrm{out}$ made above. Now let a function~$H_X$ be constructed from~$X$ via~\eqref{eq:HX-def}. Provided that the global additive constant in its definition is chosen appropriately, we have
 \begin{equation}
 \label{eq:HF4-bc}
 \begin{array}{l}
 H_X((b_1 a_2)^\bullet)=1,\\[2pt]
 H_X((a_2 b_2)^\circ)=0,
 \end{array} \qquad H_X((b_2a_1)^\bullet)=H_X((a_1b_1)^\circ)=\cos^2{\mathrm{p}},
 \end{equation}
 and~$\E^{(0)}[\sigma_{(a_1b_1)^\circ}\sigma_{(a_2b_2)^\circ}]=\tan\frac{1}{2}{\mathrm{p}}$; cf.~\cite[Eq.~(6.5)]{ChSmi2}.

\smallskip

We now describe the \emph{general strategy of the proof of Theorem~\ref{thm:RSW-selfdual},} which details are presented below. Given~$x_1<x_2$, $y_1<y_2$, let \mbox{$\cR:=(x_1,x_2)\times(y_1,y_2)\subset\C$} and {$\cR^\delta$ be} `straight rectangles'
\[
\cR^\delta:=[\cR(x_1,x_2;y_1,y_2)]^{\circ\bullet\circ\bullet}_{\cS^\delta}
\]
with corners~$a_1^\delta,b_1^\delta,a_2^\delta,b_2^\delta$ drawn on s-embeddings $\cS^\delta$ with~$\delta\to 0$. Let~$X^\delta$ be the Kadanoff--Ceva fermionic observables discussed above and~$F^\delta$ be the corresponding s-holomorphic functions in~$\cR^\delta$; see Proposition~\ref{prop:shol=3term}.

The functions~$H_{F^\delta}$ are uniformly bounded on~${\cR}^\delta$ due to the maximum principle (see Proposition~\ref{prop:HF-max}) and to the boundary conditions~\eqref{eq:HF4-bc}. The a priori regularity results from Section~\ref{sub:regularity} and Proposition~\ref{prop:f-hol} imply that there exist subsequential limits
\[
F^\delta\to f,\quad H_{F^\delta}\to {\textstyle h=\frac{1}{2}\int\Im[(f(z))^2dz]}\quad \text{on~compact~subsets~of~}\cR,
\]
where the function~$f$ is holomorphic and the function~$h:\cR\to [0,1]$ is harmonic. Assume now that~$\E^{(0)}_{\cR^\delta}[\sigma_{(a_1^\delta b_1^\delta)^\circ}\sigma_{(a_2^\delta b_2^\delta)^\circ}]\to 0$ as~$\delta\to 0$, which means that~${\mathrm{p}}={\mathrm{p}}^\delta\to 0$ in the boundary conditions~\eqref{eq:HF4-bc}. The desired contradiction is obtained in two steps:
\begin{itemize}
\item We show that~$h$ must have boundary values~$0$ at the top side~$(a_2b_2)$ of~$\cR$ and boundary values~$1$ at all the three other sides, including~$(a_1b_1)$. Intuitively, this follows from~\eqref{eq:HF4-bc} and~${\mathrm{p}}^\delta\to 0$ but an accurate proof is rather involved.

\smallskip
\item To rule out the only remaining possibility~$h(\cdot)=1-\hm_\cR(\cdot,(a_2b_2))$ we analyze the behavior of~$f=\sqrt{i\partial h}$ near the bottom side~$(a_1b_1)$ of~$\cR$. Loosely speaking, the contradiction comes from the fact that, for a small~$\phi_0>0$, we have $\Re [e^{\pm i\phi_0} F^\delta]\ge 0$ at the boundary arcs~$(a_1^\delta b_1^\delta)^\circ$, which is not compatible with almost constant purely imaginary values of~$f$ near a point on~$(a_1b_1)$.
\end{itemize}
Let us emphasize that \emph{both} these steps heavily rely upon a very particular choice of `straight cuts' on s-embeddings~$\cS^\delta$ and  `straight rectangles'~$\cR^\delta$; {see Section~\ref{sub:cuts}.}

{
\subsubsection{More recent developments: an alternative strategy of the proof} \label{subsub:RSW-Mahfouf}
Since the first version of this paper was published, a more conceptual strategy of the proof of Theorem~\ref{thm:RSW-selfdual} appeared in a very recent work of Mahfouf~\cite[Section~5]{mahfouf-thesis}. The most technical parts of our proof of Theorem~\ref{thm:RSW-selfdual} presented in the forthcoming Sections~\ref{sub:cuts}--\ref{sub:RSW-proof} are the construction of special `straight' cuts on a given \mbox{s-embedding}~$\cS^\delta$ and the analysis of fermionic observables near such `straight' boundaries. As shown in~\cite[Section~5]{mahfouf-thesis}, these unpleasant technicalities can be bypassed if one first `glues' (staying within the s-embeddings framework) a piece of an auxiliary square grid to the piece of~$\cS^\delta$ under consideration. This procedure does \emph{not} rely upon any additional assumption on~$\cQ^\delta$ except~\LipKd\ and provides an extension~$\widetilde{\cS}^\delta$ of~$\cS^\delta$ such that the condition~$\LipKd$ remains essentially unchanged and one has~$r_{\widetilde{z}}\ge \delta^{100}$ for all quads~$\widetilde{z}$ in~$\widetilde{\cS}^\delta$. In particular, in this case one can  take~$\delta':=99\delta|\log\delta|$ in Assumption~\ref{assump:ExpFat} to treat subsequential limits of s-holomorphic functions on~$\widetilde{\cS}^\delta$ as $\delta\to 0$. Moreover, we believe that this approach should also lead to an affirmative answer to the open question (I) from Section~\ref{sub:further} in full generality, i.e., under the least restrictive `technical' assumption~\ExpFat\ on s-embeddings~$\cS^\delta$ for which the analytic tools discussed in Section~\ref{sub:shol-limits} can be applied as for now.

We decided to keep the original proof of Theorem~\ref{thm:RSW-selfdual} below instead of referring to~\cite[Section~5]{mahfouf-thesis} as we hope that the construction of such `straight' cuts can be also used for other purposes. In particular, we hope that these techniques, together with the strategy recently developed in~\cite{DCMT-fractal}, can lead to the proof of uniform `boundary-to-boundary' crossings estimates in general topological rectangles with free boundary conditions at least for all critical doubly periodic planar Ising models. In their turn, such `strong RSW' estimates allow one to deduce the convergence of full loop ensembles from that of interfaces; see~\cite[Section~5]{Ch-ICM18} and references therein.}

\subsection{Construction of discrete half-planes and discrete rectangles~${\cR^\delta}$}\label{sub:cuts} To simplify the notation, below we work in the full plane setup and focus on a construction of a discrete \emph{upper half-plane}~${\mathbb{H}}^\circ_\cS$ on an s-embedding~\mbox{$\cS=\cS^\delta$} satisfying the assumption~\Unif. Generalizations required to construct similar discrete half-planes \mbox{$[\overline{\alpha}^2(\H+is)]^\circ_\cS$} and \mbox{$[\overline{\alpha}^2(\H+is)]^\bullet_\cS$} with~$\alpha\in\mathbb{T},s\in\R$, are more-or-less straightforward and mentioned explicitly at the end of this section. The same constructions can be also applied locally for s-embeddings with the disc topology.

The construction of~$\H^\circ_\cS$ comes from the S-graph~$\cS-i\cQ$, recall that the image of a quad~$z=(v^\bullet_0(z)v^\circ_0(z)v^\bullet_1(z)v^\circ_1(z))$ in this S-graph is a \emph{non-convex} quad~$(\cS-i\cQ)^\dm(z)$, with both vertices $(\cS-i\cQ)(v^\circ_q(z))$ lying \emph{higher} than both  $(\cS-i\cQ)(v^\bullet_p(z))$. Let us note the following simple geometric property of these quads in the situation when the s-embedding~$\cS$ satisfies the assumption~\Unif, cf.~Remark~\ref{rem:Sgraph-deg}:
\begin{itemize}
\item There exists a constant~$\veps_0>0$ depending only on constants in~\Unif\ such that, for each quad~$z\in\Dm(G)$, the following estimate holds:
\[
\Im(\cS-i\cQ)(v_q^{\circ}(z))-\Im(\cS-i\cQ)(v_p^{\bullet}(z))\ \ge\ 4\veps_0\delta
\]
except maybe for \emph{one} of the four sides~$(v^\bullet_p(z)v^\circ_q(z))$, $p=0,1$, $q=0,1$. Moreover, if~$(\cS-i\cQ)(v)$ is the concave vertex of~$(\cS-i\cQ)^\dm(z)$, then
\[
|\Im(\cS-i\cQ)(v')-\Im(\cS-i\cQ)(v)|\ \ge\ 4\veps_0\delta,
\]
where~$v':=v^\bullet_{1-p}$ if~$v=v^\bullet_p$ and $v':=v^\circ_{1-q}$ if~$v=v^\circ_q$.
\end{itemize}
The following definition provides a technical notation, which we use in Definition~\ref{def:half-plane} to construct the discrete half-plane~$\H^\circ_\cS$; see also Fig.~\ref{fig:SgraphCut}.
\begin{definition}\label{def:doubleton} Let~$\veps_0>0$ be chosen as above.
\begin{itemize}
\item We say that a pair of adjacent vertices~$v^\circ\in G^\circ$ and~$v^\bullet\in G^\bullet$ form a \emph{doubleton} in the S-graph~$\cS-i\cQ$ if
\[
\Im(\cS-i\cQ)(v^\circ)-\Im(\cS-i\cQ)(v^\bullet)\ <\ 2\veps_0\delta.
\]
If a vertex~$v\in\Lambda(G)$ is not a part of a doubleton, we call it a \emph{singleton}.
\smallskip
\item We call~$v^\circ\!\in G^\circ$ \emph{positive} if~$\Im(\cS-i\cQ)(v^\circ)\ge \veps_0\delta$ and \emph{negative} otherwise.
\item We call~$v^\bullet\!\in G^\bullet$ \emph{positive} if~$\Im(\cS-i\cQ)(v^\bullet)>-\veps_0\delta$ and \emph{negative} otherwise.
\item We call a doubleton~$(v^\circ v^\bullet)$ \emph{neutral} if~$v^\circ$ is negative but~$v^\bullet$ is positive (in other words, if~$|\Im(\cS-i\cQ)(v)|<\veps_0\delta$ for both vertices of the doubleton).
\end{itemize}
\end{definition}
\begin{remark} The introduction of the cut-off~$\veps_0\delta$ and a special attention paid to neutral doubletons is motivated by the fact that below we consider backward random walks in the discrete half-plane~$\H^\circ_\cS$ associated with small \emph{perturbations} \mbox{$\cS-ie^{2i\phi}\cQ$} of the S-graph~$\cS-i\cQ$ {(below we always require that~$|\phi|\le\phi_0$ , where $\phi_0>0$ does not depend on~$\delta$ and is chosen small enough).} Under such perturbations, possible changes in the combinatorial structure of the S-graph are caused by doubletons.
\end{remark}
Let~$G^\bullet_+$ and~$G^\bullet_-$ denote the sets of positive and negative vertices~$v^\bullet\in G^\bullet$, respectively, and let~$G^\bullet_0\subset G^\bullet_+$ be the set of `black' vertices which are parts of neutral doubletons. It is easy to see that
\begin{itemize}
\item under the assumption~\Unif, there exists a constant~$C_0>0$ such that
\begin{itemize}
\item all vertices~$v^\bullet\in G^\bullet$ with~$\Im\cS(v^\bullet)\ge C_0\delta$ are positive and
\item all vertices~$v^\bullet\in G^\bullet$ with~$\Im\cS(v^\bullet)\le -C_0\delta$ are negative;
\end{itemize}
\smallskip
\item the set~$G^\bullet_-\cup G^\bullet_0$ is a \emph{connected} subgraph of~$G^\bullet$. Indeed, if~$(\cS-i\cQ)(v^\bullet)$ is the concave vertex of a quad~$(\cS-i\cQ)^\dm(z)$, then the opposite vertex of this quad lies at least~$4\veps_0\delta$ lower and one can iterate this consideration.
\end{itemize}

Since the graph~$G^\bullet_-\cup G^\bullet_0$ is connected, the set~$G^\bullet_+\setminus G^\bullet_0$ is a disjoint union of simply connected subgraphs of~$G^\bullet$. Let~$G^\bullet_\H$ be the connected component of $G^\bullet_+\setminus G^\bullet_0$ that contains all vertices~$v^\bullet$ with~$\Im(\cS-i\cQ)(v^\bullet)\ge C_0\delta$. We are now ready to give the central definition of this section.

\begin{figure}
\centering
\includegraphics[clip, trim=5cm 14cm 4cm 6cm, width=0.8\textwidth]{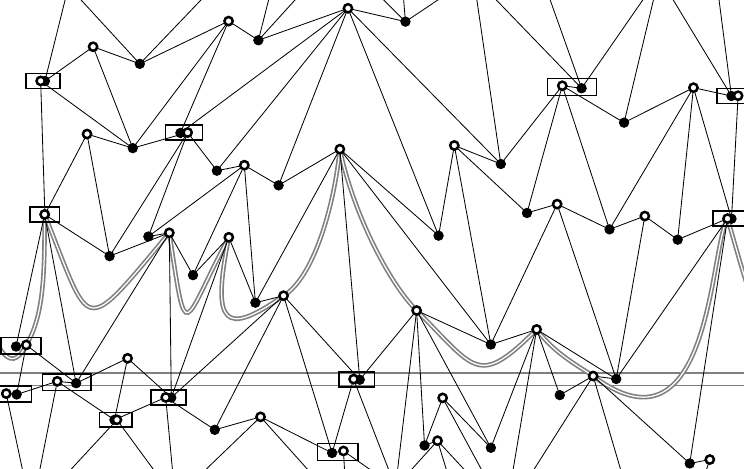}
\caption{An example of a `horizontal straight cut' {(the curve formed by double edges) on the S-graph} $\cS-i\cQ$ from the Definition~\ref{def:half-plane} of the discrete {upper} half-plane~$\H^\circ_{\cS}$. {A vertex $v^\bullet\in G^\bullet$ belongs to $\mathbb{H}^\circ_\cS$ if and only if $\Im \cS(v^\bullet)>-\varepsilon_0\delta$ and $v^\bullet$ does not form a doubleton  lying in the strip~$\R\times(-\varepsilon_0\delta,\varepsilon_0\delta)$.} The horizontal lines represent the levels~$\pm\veps_0\delta$, doubletons are marked by rectangles.
\label{fig:SgraphCut}}
\end{figure}

\begin{definition} \label{def:half-plane}
Let~$\pa\H^\circ_\cS$ be the simple path on~$G^\circ\cup\Dm(G)$ that separates~$G^\bullet_\H$ and its complement in $G^\bullet$, and let~$\H^\circ_\cS\subset \Lambda(G)\cup\Dm(G)$ be the set of vertices and quads lying above (or at) this path in the S-graph~$\cS-i\cQ$.
\end{definition}
Let us emphasize once again that~$\pa\H^\circ_\cS$ is always a \emph{simple} path: since~$G^\bullet_-\cup G^\bullet_0$ is a connected subgraph of~$G^\bullet$ and thus~$G^\bullet_\H$ is simply connected, no vertex of~$G^\circ$ can be visited by~$\pa\H^\circ_\cS$ more than once due to topological reasons.

Recall that the S-graph~$\cS-i\cQ$ defines a bijection~$z\leftrightarrow v(z)$ between~$\Dm(G)$ and~$\Lambda(G)$: by definition, $(\cS-i\cQ)(v(z))$ is the concave vertex of~$(\cS-i\cQ)^\dm(z)$. The following two properties of~$\H^\circ_\cS$ are straightforward from its definition (recall that we view the path~$\pa\H^\circ_\cS$ as a part of~$\H^\circ_\cS$):
\renewcommand{\theenumi}{\alph{enumi}}
\begin{enumerate}
\item $z\in \H^\circ_\cS$ if and only if~$v(z)\in \H^\circ_\cS$;
\item if $(v^\circ v^\bullet)$ is a doubleton in~$\cS-i\cQ$, then~$v^\circ\in\H^\circ_\cS$ if and only if~$v^\bullet\in\H^\circ_\cS$.
\end{enumerate}
In particular, if we replace the S-graph~$\cS-i\cQ$ by~$\cS-ie^{2i\phi}\cQ$ with a sufficiently small~$\phi\in[-\phi_0,\phi_0]$ (where~$\phi_0>0$ depends only on constants in~\Unif), then the equivalence (a) remains true for the bijection~$z\leftrightarrow v(z)$ defined by the \emph{new} S-graph.

\smallskip

Let us now discuss simple modifications required to construct discrete half-planes $[\overline{\alpha}^2(\H+is)]^\circ_\cS$ and~$[\overline{\alpha}^2(\H+is)]^\bullet_\cS$ with~$\alpha\in\mathbb{T}$ and~$s\in\R$.
\begin{itemize}
\item To construct a \emph{shifted} half-plane~$[\H+is]^\circ_\cS$ with~$s\in\R$, use the levels~$s\pm\veps_0\delta$ instead of~$\pm\veps_0\delta$ in Definition~\ref{def:half-plane} of positive and negative vertices.
\smallskip
\item To construct a \emph{rotated} half-plane~$[\overline{\alpha}^2(\H+is)]^\circ_\cS$ with~$\alpha\in\mathbb{T}$, use the S-graph $\cS-i\overline{\alpha}^2\cQ$ instead of~$\cS-i\cQ$ and the condition~$\Im[(\alpha^2\cS-i\cQ)(v)]\ge s\pm \veps_0\delta$ to define positive vertices in Definition~\ref{def:half-plane}.
\smallskip
\item Finally, to construct a discrete half-plane~$[\overline{\alpha}^2(\H+is)]^\bullet_\cS$ with boundary running along the set~$G^\bullet\cup\Dm(G)$, use the S-graph~$\cS+i\overline{\alpha}^2\cQ$ instead of~$\cS-i\overline{\alpha}^2\cQ$ and interchange the roles of $G^\circ$ and~$G^\bullet$ in Definition~\ref{def:doubleton} and Definition~\ref{def:half-plane}.
\end{itemize}

We conclude this section by defining discrete rectangles that we use in Theorem~\ref{thm:RSW-selfdual}. For~$x_1<x_2$ and~$y_1<y_2$, we set (up to a small modification near the corners discussed below)
\begin{align}
[\cR(x_1,x_2;&y_1,y_2)]^{\circ\bullet\circ\bullet}_\cS \notag \\
&:=\ [(\H+iy_1)]^\circ_\cS\ \cap\ [i\H+x_2]^\bullet_\cS\ \cap\ [-\H+iy_2]^\circ_\cS\ \cap\ [-i\H+x_1]^\bullet_\cS
\label{eq:Rdelta-def}
\end{align}
and similarly for other superscripts in place of~${}^{\circ\bullet\circ\bullet}$ where each symbol ($\circ$ or~$\bullet$) denotes the `type'of the discretization chosen on the corresponding (bottom, right, top, left) side of the rectangle.

Under the assumption~\Qflat, the intersection~\eqref{eq:Rdelta-def} defines a finite set, whose boundaries on the s-embedding~$\cS$ go within~$O(\delta)$ from the vertical lines~$x_{1,2}+i\R$ and horizontal lines~$y_{1,2}+\R$. In general, the boundaries of discrete half-planes can intersect more than once near the points $x_{1,2}+iy_{1,2}$. In this case, we always assume that, if necessary, the set~\eqref{eq:Rdelta-def} is slightly modified near the corners so that its boundary consists of four parts
\begin{equation}
\label{eq:Rdelta-ab-def}
\begin{array}{ll}
(a_1b_1)^\circ\in\pa[\H+iy_1]^\circ_\cS,\quad & (b_1a_2)^\bullet\in\pa[i\H+x_2]^\bullet_\cS,\\
(a_2b_2)^\circ\in\pa[-\H+iy_2]^\circ_\cS, & (b_2a_1)^\bullet\in\pa[-i\H+x_1]^\bullet_\cS,
\end{array}
\end{equation}
where~$a_1,b_1,a_2,b_2\in\Upsilon(G)$ as in Section~\ref{sub:strategy-RSW}; see also Fig.~\ref{fig:SDdomain}.

\subsection{Boundary values of fermionic observables in discrete rectangles and backward random walks stopped at~${\pa\H^\circ_\cS}$}\label{sub:bc-in-R}
Let~$F=F^\delta$ be the basic fermionic observable in a discrete rectangle~$\cR^\delta=[\cR(x_1,x_2;y_1,y_2)]^{\circ\bullet\bullet\bullet}_{\cS^\delta}$ with \emph{Dobrushin} boundary conditions (wired along the bottom side~$(ab)$, dual-wired along the three other sides); below we assume that~$y_1=0$ and skip~$\delta$ for shortness. (Recall that our strategy is to analyze the boundary behavior of Dobrushin observables in discrete rectangles first and then use them to analyze more general fermionic observables from Section~\ref{sub:strategy-RSW}.)

At inner quads~$z$ of~$\cR$, the function~$F(z)$ is related to the Kadanoff--Ceva fermionic observable~$X(c)=\E[\chi_c\mu_{(ba)^\bullet}\sigma_{(ab)^\circ}]$ by the identity~\eqref{eq:X-from-F}. The same identity can be used to define the values~$F(z)$ at \emph{boundary half-quads} of~$\cR$. Note that the function~$F$ is defined up to a global sign caused by the ambiguity in the choice of the global sign of the Dirac spinor
\begin{equation}
\label{eq:cF-sign-in-R}
\cX(c)=(\cS(v^\bullet(c))-\cS(v^\circ(c)))^{1/2}.
\end{equation}
In what follows, let us fix lifts of near-to-boundary corners~$c\in\Upsilon(G)$ to the double cover~$\Upsilon^\times(G)$ so that the values~\eqref{eq:cF-sign-in-R} change `in a continuous manner' along the boundary of~$\cR$ except at the corner~$a$, where the sign of the square root flips. (In other words, we require that $\arg(\cX(c_{-+})\overline{\cX(c_{--})})$, $\arg(\cX(c_{-+})\overline{\cX(c_{+-})})$ and~$\arg(\cX(c_{++})\overline{\cX(c_{+-})})$ belong to~$(0;\pi)$ in the notation of the left part of Fig.~\ref{fig:SembBdry}, and similarly along other boundaries.)

Below we focus on the bottom side~$\pa\H^\circ_\cS$ of the discrete rectangle~$\cR$ with wired boundary conditions, the other three sides (carrying dual-wired boundary conditions) can be treated similarly. Let~$z=(v^\bullet v^\circ_-zv^\circ_+)$ be a boundary half-quad of~$\cR$ and~$c_\pm\in\Upsilon^\times(G)$ correspond to the edges~$(v^\circ_\pm v^\bullet)$ as explained above; see Fig.~\ref{fig:SembBdry}.

By definition of the Kadanoff--Ceva fermionic observables~$X$ we have
\begin{equation}
\label{eq:2ptX>0}
X(c_-)=X(c_+)\ =\ \E_\cR[\chi_{c_\pm}\sigma_{(ab)^\circ}\mu_{(ba)^\bullet}]=
\E_\cR[\mu_{v^\bullet}\mu_{(ba)^\bullet}]>0.
\end{equation}

{
\begin{remark}\label{rem:4pt>0} It is worth noting that the same positivity condition
\begin{equation}
\label{eq:4pt>0}
X(c_-)=X(c_+)\ =\ \E_\cR^{(\mathrm{p})}[\mu_{v^\bullet}\mu_{(b_1a_2)^\bullet}\sigma_{(a_1b_1)^\circ}\sigma_{(a_2b_2)^\circ}]\ >\ 0
\end{equation}
holds for the observables in discrete rectangles with~${}^{\circ\bullet\circ\bullet}$ boundary conditions defined in Section~\ref{sub:strategy-RSW} (as above, we denote by~$c_\pm\in\Upsilon^\times(G)$ two nearby corners corresponding to the edges $(v^\circ_\pm v^\bullet)$ of a boundary half-quad adjacent to~$(a_1b_1)^\circ$; see Fig.~\ref{fig:SembBdry}). To prove~\eqref{eq:4pt>0}, note that
\[
\sigma_{(a_1b_1)^\circ}\sigma_{(a_2b_2)^\circ}\cdot \cos\mathrm{p}\ =\ 1 - \mu_{(b_2a_1)^\bullet}\mu_{(b_1a_2)^\bullet}\cdot \sin\mathrm{p}
\]
and hence
\begin{align*}
X(c_\pm)\cdot \cos\mathrm{p}\ &=\ \E_\cR^{(\mathrm{p})}[\mu_{v^\bullet}\mu_{(b_1a_2)^\bullet}]-\E_\cR^{(\mathrm{p})}[\mu_{v^\bullet}\mu_{(b_2a_1)^\bullet}]\cdot \sin\mathrm{p}\\
&=\ \E_\cR^{(\mathrm{p})}[\mu_{v^\bullet}\mu_{(b_1a_2)^\bullet}]-\E_\cR^{(\mathrm{p})}[\mu_{v^\bullet}\mu_{(b_2a_1)^\bullet}]\cdot \E_\cR^{(\mathrm{p})}[\mu_{v^\bullet}\mu_{(b_1a_2)^\bullet}]\ >\ 0,
\end{align*}
where the equality $\sin\mathrm{p}=\E_\cR^{(\mathrm{p})}[\mu_{v^\bullet}\mu_{(b_1a_2)^\bullet}]$ follows from the fact that~$X(a_1)=0$ and the positivity follows from the FKG inequality applied to the dual model.
\end{remark}}

\begin{lemma}\label{lem:argF=} In the setup described above and for a proper choice of the global sign of the observable~$F$ with Dobrushin boundary conditions, the following holds:
\begin{equation}
\label{eq:argF=}
\arg F(z)\ =\ \arg\big(i\varsigma(\overline{\cX(c_+)}-\overline{\cX(c_-)})\big),
\end{equation}
for all half-quads $(v^\bullet v^\circ_- z v^\circ_+)$ on the wired part~$(ab)^\circ$ of the boundary; see Fig.~\ref{fig:SembBdry}.

The same identity \emph{modulo~$\pi$} holds for all fermionic observables defined in a discrete domain whose boundary contains~$(ab)^\circ$.
\end{lemma}

\begin{proof} The identities~\eqref{eq:X-from-F} and~\eqref{eq:2ptX>0} imply that~$X:=X(c_\pm)=\Re[\,\overline{\varsigma}F(z)\cdot \cX(c_\pm)]$,
which gives the explicit formula
\[
\overline{\varsigma}F(z)\ =\ iX\cdot({\overline{\cX(c_+)}-\overline{\cX(c_-)}})\,/\,{\Im\big[\cX(c_+)\overline{\cX(c_-)}\big]}.
\]
The claim for fermionic observables with Dobrushin boundary conditions easily follows since the denominator is a positive number (recall that we chose the lifts of~$c_\pm$ to the double cover~$\Upsilon^\times(G)$ so that the values~$\cX(c)$ change `continuously' along the boudnary) and~$X>0$ for Dobrushin observables. {A} similar claim for other fermionic observables follows since the identity~$X(c_-)=X(c_+)\in\R$~(similar to~\eqref{eq:2ptX>0}) holds for all of them, though without the positivity condition.
\end{proof}
\begin{corollary} \label{cor:ReF>0}
There exists a constant~$\phi_0>0$ depending only on constants in \Unif\ such that, with a proper choice of the sign of the Dobrushin observable~$F$ in a discrete rectangle~$\cR=\cR^{\circ\bullet\bullet\bullet}$, one has
\[
\Re[e^{i\phi}F(z)]\ge 0\quad \text{for~all}\ \ \phi\in[-\phi_0,\phi_0]
\]
and all boundary half-quads $(v^\bullet v^\circ_- z v^\circ_+)$ at the bottom side~$(ab)^\circ\subset\pa\H^\circ_\cS$ of~$\cR$.
\end{corollary}
\begin{proof}
Let us start with the case~$\phi=0$. To simplify the notation, denote
\[
w:=\cS(v^\bullet)-\tfrac{1}{2}(\cS(v^\circ_-)+\cS(v^\circ_+))\quad \text{and}\quad u:=\tfrac{1}{2}(\cS(v^\circ_+)-\cS(v^\circ_-)).
\]
Due to Lemma~\ref{lem:argF=}, the condition~$\Re[F(z)]\ge 0$ can be equivalently rewritten as
\begin{equation}
\label{eq:arg-Zhuk}
\arg\big[w-(w^2\!-u^2)^{1/2}\big]\ =\ 2\arg\big[(w\!-\!u)^{1/2}-(w\!+\!u)^{1/2}\big]\ \in\ (\tfrac{1}{2}\pi;\tfrac{5}{2}\pi)
\end{equation}
Since~$w\mapsto w-(w^2\!-\!u^2)^{1/2}$ is nothing but the inverse Zhukovsky (or Joukowski) function, it is not hard to see that~\eqref{eq:arg-Zhuk} holds in the following three cases:
\begin{itemize}
\item $\arg u=\arg(\cS(v^\circ_+)-\cS(v^\circ_-))\in[-\frac{1}{2}\pi;\frac{1}{2}\pi]$,

\item $\arg u\in (\tfrac{1}{2}\pi;\frac{3}{2}\pi)$ and~$|w+u|-|w-u|<2\Im u$,

\item $\arg u\in (-\tfrac{3}{2}\pi;-\frac{1}{2}\pi)$ and~$|w+u|-|w-u|>2\Im u$.
\end{itemize}
As~$|w+u|-|w-u|=-\cQ(v^\circ_-)+\cQ(v^\circ_+)$, the preceding computation can be summarized in a surprisingly simple form:
\[
\Re[F(z)]\ge 0\quad \Leftrightarrow\quad \arg\big[(\cS-i\cQ)(v^\circ_+)-(\cS-i\cQ)(v^\circ_-)\big]\in (-\pi,
\pi)\ \ \text{along}\ \ \pa\H^\circ_\cS.
\]
In other words, the direction of the boundary of~$\pa\H^\circ_\cS$, \emph{drawn on the S-graph~$\cS-i\cQ$,} should not turn across the negative real line~$\R_-$. This condition easily follows from the construction of the discrete plane~$\H^\circ_\cS$ given in Section~\ref{sub:cuts}. (To give a formal proof, note that the oriented angle between diagonals of a non-self-intersecting quad~$(\cS-i\cQ)^\dm(z)$ is always in~$(0;\pi)$ and that the direction of the vector going from the positive to the negative (or neutral) `black' vertex of this quad is in~$(-\pi;0)$.)

Finally, note that a similar condition~$\Re[e^{i\phi}F(z)]\ge 0$ along~$\pa\H^\circ_\cS$ is equivalent to say that the direction of the boundary~$\pa\H^\circ_\cS$ \emph{drawn on the S-graph~$\cS-ie^{-2i\phi}\cQ$} also does not turn across the negative real line. For sufficiently small~$|\phi|\le\phi_0$ this follows from the construction of~$\H^\circ_\cS$ and the assumption~\Unif.
\end{proof}

\begin{figure}
\centering \includegraphics[clip, trim=4.6cm 12.2cm 1.4cm 5.6cm, width=0.92\textwidth]{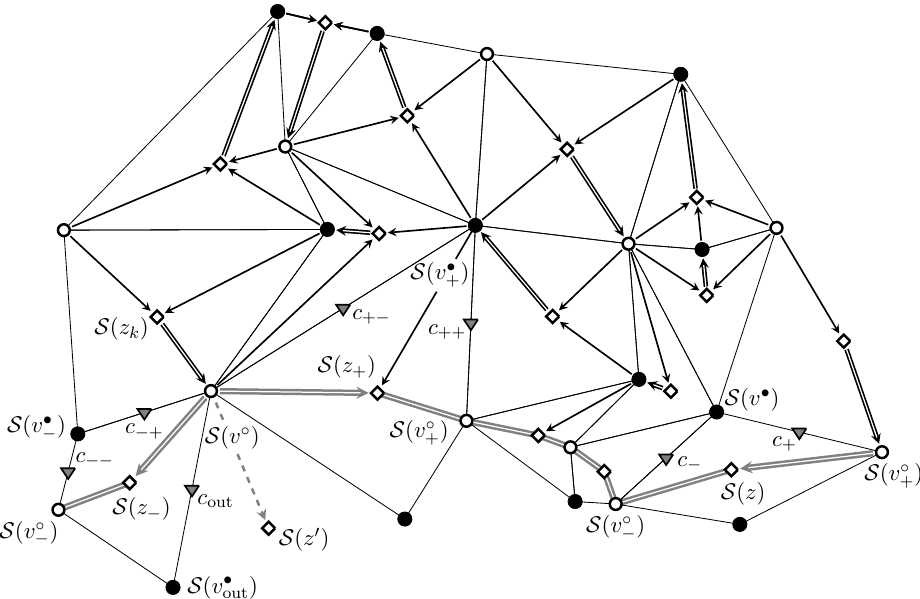}
\caption{An example of allowed jumps of the \emph{backward} random walk associated with the S-graph~$\cS-i\cQ$ stopped at the boundary of the discrete half-plane~$\H^\circ_\cS$. Double arrows indicate the bijection~$z\mapsto v(z)$ defined by this S-graph; note that not all boundary half-quads~$z\in\pa\H^\circ_\cS$ can be reached by this random walk.
\textsc{Right:} The notation for a boundary half-quad~$(v^\bullet v^\circ_- z v^\circ_+)$ used in Lemma~\ref{lem:argF=} and Corollary~\ref{cor:ReF>0}. \textsc{Left:} The notation for two consecutive boundary half-quads~$(v^\bullet_- v^\circ_- z_- v^\circ)$ and~$(v^\bullet_+ v^\circ z_+ v^\circ_+)$ used in Definition~\ref{def:RW-at-paH}. As explained in Lemma~\ref{lem:RW-at-paH} and Lemma~\ref{lem:kpm=O(1)}, the jumps from~$z_k$ to~$z'$ can be redistributed to~$z_+$ not breaking the martingale property of (the real parts of) fermionic observables. \label{fig:SembBdry}}
\end{figure}

Recall that real parts~$\Re[F(z)]$ of s-holomorphic functions are martingales with respect to a certain random walk~$\widetilde{Z}_t$ on~$\Dm(G)$, namely with respect to the backward random walk associated with the S-graph~$\cS-i\cQ$; see Proposition~\ref{prop:PrF-mart}, Remark~\ref{rem:associatedRW} and the equation~\eqref{eq:ReF-RW} for the values~$\Re(z_k)$ near a vertex~$v\in\Lambda(G)$. When approaching the boundary of the discrete half-plane~$\H^\circ_\cS$ (i.e., when~$v\in\pa\H^\circ_\cS$), this random walk can jump \emph{across} $\pa\H^\circ_\cS$ rather than simply hit a point~$z\in\pa\H^\circ_\cS$. This definition is not compatible with the analysis of fermionic observables near~$\pa\H^\circ_\cS$ (since these observables are not defined below the path~$\pa\H^\circ_\cS$) and should be modified if we wish to apply the optional stopping theorem for these martingales up to the hitting time to~$\pa\H^\circ_\cS$. We discuss such a modification below.

Recall that~$\pa\H^\circ_\cS$ is a \emph{simple} path on~$G^\circ\cup\Dm(G)$. Let~$(v^\bullet_-v^\circ_-zv^\circ)$ and~$(v^\bullet_+v^\circ zv^\circ_+)$ be two consecutive boundary half-quads at~$\pa\H^\circ_\cS$ and let~$v^\bullet_\mathrm{out}\not\in\H^\circ_\cS$ be one of the `outer' neighbors of~$v^\circ$, i.e.,~$v^\circ\sim v^\bullet_\mathrm{out}\in  G^\bullet\smallsetminus G^\bullet_\H$.

\begin{definition}\label{def:RW-at-paH} We say that the backward random walk associated with the \mbox{S-graph} $\cS-i\cQ$ is stopped at~$\pa\H^\circ_\cS$ if, for each~$v^\circ\in\pa\H^\circ_\cS$, all its transitions from $z\in\Int\H^\circ_\cS$ to $z'\not\in\H^\circ_\cS$ with~$z\sim v^\circ\sim z'$  are redirected either to~$z_+$ or to~$z_-$ depending on the position of~$z'$ with respect to the edge~$(v^\circ v^\bullet_\mathrm{out})$; see the left part of Fig.~\ref{fig:SembBdry}.

\end{definition}
\begin{remark} \label{rem:RWphi-at-paH}
 A similar construction of the backward random walk associated with the S-graph~$\cS-ie^{2i\phi}\cQ$ \emph{stopped at~$\pa\H^\circ_\cS$} applies for all~$\phi\in[-\phi_0,\phi_0]$, where~$\phi_0>0$ is fixed in Corollary~\ref{cor:ReF>0}. The condition~$|\phi|\le\phi_0$ is particularly important in the forthcoming discussion of the \emph{positivity} of coefficients~$k_\pm$ in the identity~\eqref{eq:RW-at-paH}.
\end{remark}

Certainly, a priori there is no guarantee that~$\Re F(z)$ remains a martingale with respect to this modified walk until its hitting time to~$\pa\H^\circ_\cS$. However, it turns out that boundary conditions~\eqref{eq:argF=} allows to save this martingale property at~$\pa\H^\circ_\cS$ by the cost of a controllable modification of values of~$\Re F(z)$ at the stopping time.

\begin{lemma}\label{lem:RW-at-paH}
Let $v^\circ\in\pa {\H}^\circ_\cS$ and~$z_+=z_0,z_1,\ldots,z_n=z_-\in\Dm(G)$ be neighbors of $v^\circ$ in the discrete upper half-plane~$\H^\circ_\cS$ listed counterclockwise. Assume that an s-holomorphic function~$F$ satisfies (modulo~$\pi$) boundary conditions~\eqref{eq:argF=} at both half-quads~$z_\pm$. Then, the following identity holds:
\begin{align}
\sum_{k=1}^{n-1}\biggl(\frac{\Re\eta_{k+1}}{\Im\eta_{k+1}}-\frac{\Re\eta_k}{\Im\eta_k}\biggr)\Re F(z_k)&+ \biggl(\frac{\Re\eta_1}{\Im\eta_1}-\frac{\Re\eta_\mathrm{out}}{\Im\eta_\mathrm{out}}\biggr)\cdot k_+\Re F(z_+) \notag\\
&+\ \biggl(\frac{\Re\eta_\mathrm{out}}{\Im\eta_\mathrm{out}}-\frac{\Re\eta_n}{\Im\eta_n}\biggr)\cdot k_-\Re F(z_-)=0,
\label{eq:RW-at-paH}
\end{align}
where the coefficients~$k_\pm\in\R$ are given by the formulas~\eqref{eq:k+=} and~\eqref{eq:k-=}.
\end{lemma}
\begin{proof} This is a straightforward computation. Using the s-holomorphicity of~$F$ as in the derivation {of}~\eqref{eq:ReF-RW}, we see that the identity~\eqref{eq:RW-at-paH} holds provided that
\[
k_+\cdot\biggl(\frac{\Re\eta_1}{\Im\eta_1}-\frac{\Re\eta_\mathrm{out}}{\Im\eta_\mathrm{out}}\biggr) =\ \frac{\Re\eta_1}{\Im\eta_1}+\frac{\Im F(z_+)}{\Re F(z_+)}
\]
and similarly for~$z_-$. Let~$c_{+\pm}:=c_\pm(z_+)=c_{(v^\circ v^\bullet_\pm)}$; see Fig.~\ref{fig:SembBdry}. It~is easy to see that
\[
\frac{\Re\eta_1}{\Im\eta_1}-\frac{\Re\eta_\mathrm{out}}{\Im\eta_\mathrm{out}}\ =\ \frac{\Re[\varsigma\overline{\cX(c_{+-})}]}{\Im[\varsigma\overline{\cX(c_{+-})}]}+ \frac{\Re[\overline{\varsigma}\cX(c_\mathrm{out})]}{\Im[\overline{\varsigma}\cX(c_\mathrm{out})]}\ =\ \frac{\Im[\cX(c_{+-})\overline{\cX(c_\mathrm{out})}]}{\Im[\varsigma\overline{\cX(c_{+-})}]\Im[\varsigma\overline{\cX(c_\mathrm{out})}]}.
\]
Also, it follows from Lemma~\ref{lem:argF=} that
\begin{align*}
\frac{\Re\eta_1}{\Im\eta_1}+\frac{\Im F(z_+)}{\Re F(z_+)}\ &=\ \frac{\Re[\varsigma\overline{\cX(c_{+-})}]}{\Im[\varsigma\overline{\cX(c_{+-})}]}- \frac{\Im[-i\overline{\varsigma}(\cX(c_{++})\!-\!\cX(c_{+-}))]}{\Re[-i\overline{\varsigma}(\cX(c_{++})\!-\!\cX(c_{+-}))]}\\
&=\ \frac{\Im[\cX(c_{++})\overline{\cX(c_{+-})}]}{\Im[\varsigma\overline{\cX(c_{+-})}]\cdot \Re[-i\overline{\varsigma}(\cX(c_{++})\!-\!\cX(c_{+-}))]}.
\end{align*}
Therefore, we arrive at the formula
\begin{equation}
\label{eq:k+=}
k_+\ =\ \frac{\Im[\cX(c_{++})\overline{\cX(c_{+-})}]}{\Im[\cX(c_{+-})\overline{\cX(c_\mathrm{out})}]}\cdot \frac{\Im[\varsigma\overline{\cX(c_\mathrm{out})}]}{\Re[-i\overline{\varsigma}(\cX(c_{++})\!-\!\cX(c_{+-}))]}\,.
\end{equation}
A similar computation (see Fig.~\ref{fig:SembBdry} for the notation) shows that
\begin{equation}
\label{eq:k-=}
k_-\ =\ -\ \frac{\Im[\cX(c_{--})\overline{\cX(c_{-+})}]}{\Im[\cX(c_{-+})\overline{\cX(c_\mathrm{out})}]}\cdot \frac{\Im[\varsigma\overline{\cX(c_\mathrm{out})}]}{\Re[-i\overline{\varsigma}(\cX(c_{-+})\!-\!\cX(c_{--}))]}\,.
\end{equation}
Note that these expressions do not depend on the choice of the sign of~$\cX(c_\mathrm{out})$.
\end{proof}
Assume now that there exists~$k=1,\ldots,n-1$ such that~$v^\circ=v(z_k)$ under the bijective correspondence {between}~$\Dm(G)$ and~$\Lambda(G)$ provided by the S-graph~$\cS-i\cQ$. (Recall that this happens if and only if the line~$\R$ lies in between~$\eta_k\R$ and~$\eta_{k+1}\R$ or, equivalently, if and only if the ray~$i\R_+$ lies in between~$(\cS(v^\bullet_k)-\cS(v^\circ))\R_+$ and~$(\cS(v^\bullet_{k+1})-\cS(v^\circ))\R_+$.) In this situation, the backward random walk associated with the S-graph~$\cS-i\cQ$ has non-zero jump rates from~$z_k$ to other neighbors of~$v^\circ$.

\begin{lemma} \label{lem:kpm=O(1)}
Let~$v^\circ\in\partial\H^\circ_\cS$ be such that~$i\R_+$ lies in between \mbox{$(\cS(v^\bullet_+)-\cS(v^\circ))\R_+$} and~$(\cS(v^\bullet_-)-\cS(v^\circ))\R_+$. Then, the coefficients $k_\pm$ given by~\eqref{eq:k+=}, \eqref{eq:k-=} are positive. Moreover, under the assumption~\Unif, both~$k_\pm$ are uniformly comparable to~$1$.
\end{lemma}
\begin{proof}
Let the choice of the sign of~$\cX(c_\mathrm{out})$ be fixed by a `continuous clockwise turn' starting from the value~$\cX(c_+)$. Then, the first factor in~\eqref{eq:k+=} is positive. Moreover, the numerator of the second factor is also positive due to the assumption made on the position of the ray~$i\R_+$. Finally, the denominator of the second factor is positive due to Corollary~\ref{cor:ReF>0} (and Lemma~\ref{lem:argF=}); note that the last fact heavily relies upon the special choice of the path~$\pa\H^\circ_\cS$ made in Section~\ref{sub:cuts}.

Assuming~\Unif, one easily sees that both the numerator and the denominator of the first factor in~\eqref{eq:k+=} are uniformly comparable to~$\delta$ while both the numerator and the denominator of the second factor are uniformly comparable to~$\delta^{1/2}$. The coefficient~$k_-$ can be treated similarly. 
\end{proof}

\subsection{Estimates of hitting probabilities in discrete rectangles} \label{sub:visibility} In this section we discuss technical estimates for the hitting probabilities of boundary points (in the sense of Definition~\ref{def:RW-at-paH}) in discrete rectangles~$\cR^\delta$ constructed above. The key estimate is given in Proposition~\ref{prop:hm>delta}: loosely speaking it says that, under the assumption~\Unif\ and for all~$|\phi|\le\phi_0$, the probability that the backward random walk associated with the S-graph~$\cS-ie^{2i\phi}\cQ$ exits the discrete half-plane~$\H^\circ_{\cS^\delta}$ through a given point~$z\in\pa\H^\circ_\cS$ (starting, e.g., at height one above this point) is uniformly comparable to~$\delta$ unless~$z$ is \emph{completely} screened by~$\pa\H^\circ_\cS$; cf. Definition~\ref{def:local-visible}.

\stoptocwriting

\subsubsection{Uniform crossing property and staircase paths on S-graphs} To begin with, let us recall that both forward and backward random walks associated with \mbox{S-graphs} satisfy the so-called \emph{uniform crossing property} (cf.~\cite{BLR}) on scales above~$\delta$ (if viewed on the s-embedding~$\cS$ rather then on~$\Dm(G)$). We refer the reader to~\cite[Sections~6.2,~6.3]{CLR1} where this fact is derived under much weaker assumption~\LipKd\ than~\Unif. In particular, this property implies that there exists a (big) constant~$C_0>0$ depending only on constants in~\Unif\ such that, for all~$w\in\C$, all~$\rho>C_0\delta$ and all~$z$ such that~$|\cS(z)-w|=\rho+O(\delta)$, the probability that the random walk started at~$z$ makes the full turn inside the annulus~$B(w,2\rho)\smallsetminus B(w,\frac{1}{2}\rho)$ before hitting its boundary is uniformly bounded from below. Moreover, one can even condition this random walk on any event outside this annulus (e.g., on exiting it through a given boundary point) and the same uniform lower bound still holds.

Among other consequences of the uniform crossing property let us mention the following one, which we frequently use in Section~\ref{sub:RSW-proof}. There exist constants~$k_0>0$ and~$\beta>0$ depending only on constants in~$\Unif$\ such that
\begin{itemize}
\item for all~$s>0$, the probability that the backward random walk associated with the \mbox{S-graph} $\cS^\delta-ie^{2i\phi}\cQ^\delta$, started at~$\cS^\delta(z)=ik_0s+O(\delta)$ and stopped when it hits~$\pa\H^\circ_{\cS^\delta}$ {or the lines}~$4k_0s-y=\pm k_0x$, { visits} the ball~$B(\pm 4s,\rho)$ is uniformly (in~$\delta\le C_0^{-1}s$) bounded by~$O(\max\{\rho,\delta\}^{1/2+\beta})$; {see Fig.~\ref{fig:technical}A.}
\end{itemize}

Recall now the notion of a \emph{doubleton} introduced in Definition~\ref{def:doubleton}. If \mbox{$w=(v^\circ v^\bullet)$} is a doubleton on the S-graph~$\cS^\delta-i\cQ^\delta$, then for all random walks associated with the S-graphs~$\cS^\delta-ie^{2i\phi}\cQ^\delta$, $|\phi|\le\phi_0$, the jump rates between~$v^\circ$ and~$v^\bullet$ are at least~$\cst\cdot\delta^{-2}$ (and possibly much higher if~$(\cS^\delta-i\cQ^\delta)(v^\circ)$ and~$(\cS^\delta-i\cQ^\delta)(v^\bullet)$ are very close to each other). All the other jump rates (outgoing from a singleton or leaving a doubleton) are uniformly comparable to~$\delta^{-2}$. Recall also that the bijection~$z\mapsto v(z)$ defined by the S-graph~$\cS^\delta-i\cQ^\delta$ remains unchanged when passing to S-graphs~$\cS^\delta-ie^{2i\phi}\cQ^\delta$ except possible swaps at doubletons.

\begin{figure}
\begin{minipage}{0.64\textwidth}

\includegraphics[clip, trim=4.55cm 21cm 8cm 4.2cm, width=\textwidth]{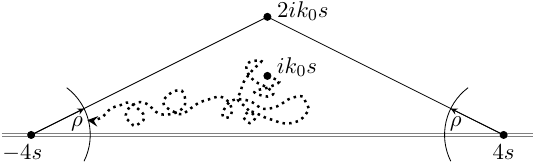}

\vskip 3pt

{\textsc{(A)} For small enough~$k_0$, the probability that a random walk satisfying the uniform crossing property hits~$B(\pm 4s,\rho)$ staying inside the triangle is~$O(\rho^{\frac12+\beta})$.}

\vskip 12pt

\includegraphics[clip, trim=4.55cm 20.8cm 8cm 4.2cm, width=\textwidth]{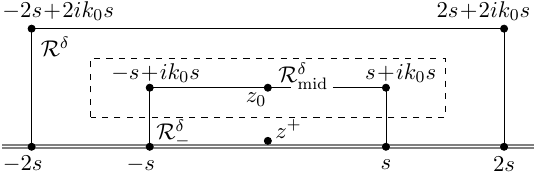}

\vskip 3pt

{\textsc{(B)} Notation from the proof of Proposition~\ref{prop:hm>delta}.}
\end{minipage}
\hskip 0.04\textwidth
\begin{minipage}{0.27\textwidth}
\includegraphics[clip, trim=11.6cm 14.7cm 5.6cm 4.8cm, width=\textwidth]{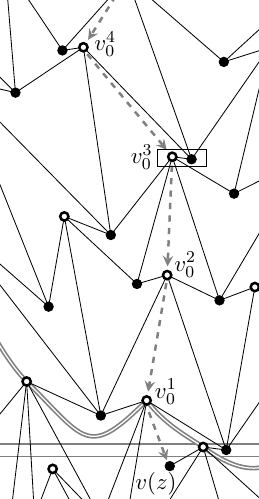}

\vskip 4pt

{\textsc{(C)} A staircase path to $\partial\mathbb{H}^\circ_{\cS^\delta}$ (Lemma~\ref{lem:staircase}).}
\end{minipage}

\caption{{`Technical' constructions used for estimates of hitting probabilities.}\label{fig:technical}}
\end{figure}

\begin{lemma} \label{lem:staircase}
Under the assumption~\Unif, for each~$z\in\Dm(G)$ there exists a sequence of vertices~$v_n^\circ\in G^\circ$, $n\ge 1$, such that
\[
\Im(\cS\!-\!i\cQ)(v^\circ_{n+1})\ \ge\ \Im(\cS\!-\!i\cQ)(v^\circ_n)+4\veps_0\delta,
\]
the jump rates from~$v^\circ_{n+1}$ to~$v_n^\circ$ of the backward random walks associated with \mbox{S-graphs} $\cS^\delta-ie^{2i\phi}\cQ^\delta$, $|\phi|\le\phi_0$, are non-zero, and the same holds for~$v_1^\circ$ and~$v(z)$.
\end{lemma}
\begin{proof}
This is a triviality: if~$v_n^\circ$ is already constructed and~$(\cS-i\cQ)(v_n^\circ)$ is the concave vertex of {the} quad~$(\cS^\delta-i\cQ^\delta)^\dm(z_n)$, then one can take the other `white' vertex of this quad as~$v_{n+1}^\circ$. The same choice works for~$v(z)$: even if~$v(z)\in G^\bullet$ and~$v(z)$ is a part of a doubleton, then the other `white' vertex of this quad lies at least~$4\veps_0\delta$ above in the S-graph~$\cS^\delta-i\cQ^\delta$.
\end{proof}
We call the paths~$v(z),v_1^\circ,v_2^\circ,\ldots $ from Lemma~\ref{lem:staircase} \emph{staircase paths}; {see Fig.~\ref{fig:technical}C.} The preceding discussion implies that
\begin{itemize}
\item there exists a constant~$p_0>0$ depending only on constants in~\Unif\ such that for all \emph{backward} random walks on S-graphs~$\cS^\delta-ie^{2i\phi}\cQ^\delta$, $|\phi|\le\phi_0$,
    \begin{itemize}
    \item if~$v_n^\circ$ is a singleton (on the S-graph~$\cS^\delta-i\cQ^\delta$), then the probability to jump from~$v_n^\circ$ to~$v_{n-1}^\circ$ is at least~$p_0$;
    \item if~$(v_n^\circ v^\bullet)$ is a doubleton, then the probability to leave this doubleton (starting, e.g., at~$v^\circ_n$) to~$v^\circ_{n-1}$ is at least~$p_0$;
    \end{itemize}
Moreover, the same property holds for~$v_1^\circ$ and~$v_0:=v(z)$ even if $v_0\in G^\bullet$.
\end{itemize}

Recall that, under the assumption~\Qflat, the discrete upper half-plane~$\H^\circ_{\cS^\delta}$ contains all~$z$ such that~$\Im(\cS^\delta-i\cQ^\delta)(z)\ge C_0\delta$, where~$C_0>0$ depends only on constants in~\Unif\ and \Qflat. Given a boundary half-quad~$z\in\pa\H^\circ_{\cS^\delta}$ it can happen that it is completely screened by the nearby parts of~$\pa\H^\circ_{\cS^\delta}$ from the viewpoint of the backward random walks in~$\H^\circ_{\cS^\delta}$ discussed above. (In other words, the irregular  boundary~$\pa\H^\circ_{\cS^\delta}$ can produce bottlenecks preventing the backward random walks to enter small fjords behind them and, in particular, preventing these random walks to hit points~$z\in\pa\H^\circ_{\cS^\delta}$ inside such fjords; see also~Fig.~\ref{fig:SembBdry}).

\begin{definition}\label{def:local-visible} We call a half-quad~$z\in\pa\H^\circ_{\cS^\delta}$ \emph{locally visible} if there exists~\mbox{$z^+\in\H^\circ_{\cS^\delta}$} with~$\Im(\cS^\delta-i\cQ^\delta)(z^+)\ge C_0\delta$ such that the backward random walk associated with the S-graph~$\cS^\delta-i\cQ^\delta$ started at~$z^+$ has non-zero probability to hit~$\pa\H^\circ_{\cS^\delta}$ at~$z$.
\end{definition}
It is worth noting that if~$z\in\pa\H^\circ_{\cS^\delta}$ is locally visible, then the same property holds for \emph{all}~$z^+\in\H^\circ_{\cS^\delta}$ with~$\Im(\cS^\delta-i\cQ^\delta)(z^+)\ge C_0\delta$; this is a consequence of the uniform crossing property on scales above~$\delta$ and of the existence of staircase paths on scale~$\delta$. Moreover, the same property holds for the backward random walks associated to \emph{each} of the S-graphs~$\cS^\delta-ie^{2i\phi}\cQ^\delta$ with~$|\phi|\le\phi_0$; this follows from the special consideration of doubletons in the construction of~$\H^\circ_{\cS^\delta}$ .

\subsubsection{Hitting probabilities in discrete rectangles} We are now ready to prove a required `uniform visibility' estimate for boundary half-quads of the half-plane~$\H^\circ_{\cS^\delta}$. Let rectangles~$\cR^\delta:=
[\cR(-2s,2s;0,2k_0s)]^{\circ\circ\circ\circ}_{\cS^\delta}$ be defined as in Section~\ref{sub:cuts} (the choice of discretizations of its right, top and left sides is irrelevant).

\begin{proposition} \label{prop:hm>delta} There exists a (small) constant~\mbox{$k_0>0$} and a (big) constant \mbox{$C_0>0$}, only depending on constants in the assumptions \Unif\ and \Qflat, such that the following holds for all ${s}\ge C_0\delta$, all~$\phi\in[-\phi_0,\phi_0]$, and all {locally visible} boundary half-quads~$z\in\pa\H^\circ_{\cS^\delta}$ such that~$|\Re(\cS^\delta-i\cQ^\delta)(z)|\le s$:

\smallskip

the probability that the backward random walk associated with \mbox{$\cS^\delta\!-\!ie^{2i\phi}\cQ^\delta$} and started {in an $O(\delta)$ vicinity of} the center of $\cR^\delta$ exits this rectangle through a given boundary half-quad~$z\in\pa\H^\circ_{\cS^\delta}$ is at least~$\cst\cdot s^{-1}\delta$, where~$\cst>0$ depends only on constants in \Unif\ and \Qflat.
\end{proposition}
\begin{remark}
Under the assumption~\Unif\ it is plausible to expect that the probability for a random walk to exit the rectangle of size~$(-2s,2s)\times(0,2k_0s)$ through a given point~$z$ in the middle of its boundary is uniformly comparable to~$s^{-1}\delta$. For our purposes it is enough to prove the lower bound only; with a caveat that~$z$ has to be \emph{locally visible} as described above, otherwise this probability simply vanishes.
\end{remark}
\begin{proof} Let 
$z_0$ be such that {its image in the S-graph~$\cS^\delta-i\cQ^\delta$} lies within~$O(\delta)$ from the center of~$\cR^\delta$; {see Fig.~\ref{fig:technical}B for the notation. Further, let} $Z_t=Z_t^{(\phi)}$ be the forward random walk associated with the S-graph \mbox{$\cS^\delta-ie^{2i\phi}\cQ^\delta$}, and~$\widetilde{Z}_t=\widetilde{Z}_t^{(\phi)}$ be the corresponding backward random walk. Also, let~$z\mapsto v^{(\phi)}(z)$ be the bijection of~$\Dm(G)$ and~$\Lambda(G)$ defined by the S-graph~$\cS^\delta-ie^{2i\phi}\cQ^\delta$.

Since we assume that the boundary half-quad~$z$ is \emph{locally visible}, for each~$C>0$ there exists~$z^+$ with~$\Im(\cS^\delta-i\cQ^\delta)(z^+)\ge C\delta$ such that, for all~$\phi\in[-\phi_0,\phi_0]$, the probability that the random walk~$\widetilde{Z}^{(\phi)}_t$ started at~$z_+$ exits~$\cR^\delta$ through~$z$ is uniformly bounded from below by~$p(C)>0$. Therefore, it is enough to prove that
\begin{equation}
\label{eq:z0->z+>delta}
\mathbb{P}^{(z_0)}[\,\widetilde{Z}^{(\phi)}_t\ \text{visits}\ z^+\ \text{before~exiting}\ \cR^\delta\,]\ \ge\ \cst\cdot s^{-1}\delta,
\end{equation}
where the superscript~$z_0$ indicates that the random walk~$\widetilde{Z}^{(\phi)}_t$ is started at~$z_0$.

Let us first consider the \emph{forward} random walk~$Z^{(\phi)}_t$ started at~$z_+$ and stopped when it exists the twice smaller rectangle~$\cR^\delta_-:=[\cR(-s,s;0,k_0s)]^\circ_{\cS^\delta}$; {see Fig.~\ref{fig:technical}B.} We claim that
\begin{equation}
\label{eq:z+->up>delta}
\mathbb{P}^{(z^+)}[\,Z^{(\phi)}_t\ \text{exits~$\cR^\delta_-$~through~its~top~side}\,]\ \ge \cst\cdot s^{-1}\delta
\end{equation}
provided that~$k_0$ is chosen small enough and~$C_0$ and~$C$ are big enough. To prove the uniform estimate~\eqref{eq:z+->up>delta}, note that the process~$\Im(\cS^\delta-ie^{2i\phi}\cQ^\delta)(v^{(\phi)}(Z^{(\phi)}_t))$ is a martingale. Under the assumption~\Qflat, the constant~$C>0$ can be chosen so as to guarantee that the starting value~$\Im(\cS^\delta-ie^{2i\phi}\cQ^\delta)(v^{(\phi)}(z_+))$ is at least~$\frac{3}{4}C\delta$ and all the values~$\Im(\cS^\delta-ie^{2i\phi}\cQ^\delta)(v^{(\phi)}(z))$ for~$z\in\pa\H^\circ_{\cS^\delta}$ are less than~$\frac{1}{4}C\delta$. Therefore, the probability to exit~$\cR^\delta_-$ through the top, left or right sides starting from~$z_+$ must be at least~$\frac{1}{4}C\delta s^{-1}$. Moreover, provided that the aspect ratio~$k_0$ is small enough, the probability to exit~$\cR^\delta_-$ through the left or the right side can be neglected compared to that for the top side. (E.g., one can easily show that the probability to hit these sides starting from~$z^+$ is~$O(\delta k_0s^{-1})$ by using quadratic {supermartingales $\varepsilon_0x^2+(y+C_0\delta)(2k_0s-y+C_0\delta)$} of the uniformly elliptic martingale process~$(\cS^\delta-ie^{2i\phi}\cQ^\delta)(v^{(\phi)}(Z_t^{(\phi)}))$ similarly to the proof of~\cite[Lemma~3.17]{ChSmi1}.)

\smallskip

Let~$G^{(\phi)}_{\cR^\delta}(z,z^+)$ be the Green function of the random walk~$Z^{(\phi)}_t$ in~$\cR^\delta$, i.e. the expected time spent at~$z$ by the walk started at~$z^+$ before it exits~$\cR^\delta$. Also, let~$\cR^\delta_\mathrm{mid}:=[\cR(-\frac{3}{2}s,\frac{3}{2}s;\frac{1}{2}k_0s,\frac{3}{2}k_0s)]^\circ_{\cS^\delta}\subset\cR^\delta$ be a smaller rectangle containing the top side of~$\cR^\delta_-$; {see Fig.~\ref{fig:technical}B.} It is easy to see that the estimate~\eqref{eq:z+->up>delta} {gives}
\begin{equation}
\label{eq:z+->middleR>delta}
\textstyle \sum_{z\in\cR^\delta_\mathrm{mid}} G^{(\phi)}_{\cR^\delta}(z,z^+)\ \ge\ \cst\cdot s \delta,
\end{equation}
where a (new) constant~$\cst>0$ depends only on constants in~\Unif\ and \Qflat. Indeed, the time parametrization of forward random walks on S-graphs is given by the trace of the variance and their instant increments are uniformly bounded by~$O(\delta)$. Therefore, once this random walk hits the top side of~$\cR^\delta_-$, the probability that it stays inside~$\cR^\delta_\mathrm{mid}$ for at least~$\cst\cdot s^2$ amount of time is uniformly bounded from below (e.g., see~\cite[Proposition~6.1]{CLR1} for a version of the Bernstein inequality which allows to control large deviations).

\smallskip

Let~$\widetilde{G}{}^{(\phi)}_{\cR^\delta}(z^+\!,z_0)$ stand for the Green function of the backward random walk in~$\cR^\delta$, i.e., denote the expected time spent at~$z^+$ by the backward random walk associated with the S-graph~$\cS^\delta-ie^{2i\phi}\cQ^\delta$ started at the point~$z_0$ near the center of~$\cR^\delta$. We now benefit from the fact that the invariant measure of the random walk~$Z_t^{(\phi)}$ is explicitly given by the area of the quads~$\cS^\dm(z)$ (see Section~\ref{sub:rwalks}) and thus is uniformly comparable to~$\delta^2$ under the assumption~\Unif. This implies the uniform estimate~$\widetilde{G}{}^{(\phi)}_{\cR^\delta}(z^+\!,z)\asymp G^{(\phi)}_{\cR^\delta}(z,z^+)$. Also, note that {the function} $\widetilde{G}{}^{(\phi)}_{\cR^\delta}(z^+\!,\,\cdot\,)$ is positive and harmonic with respect to the backward random walk and thus satisfies the Harnack principle (e.g., see~\cite[Corollary~6.12]{CLR1}). Therefore, the estimate~\eqref{eq:z+->middleR>delta} implies that
\begin{equation}
\label{eq:G(z0->z+)>delta3}
\textstyle \widetilde{G}{}^{(\phi)}_{\cR^\delta}(z^+\!,z_0)\ \ge\ \cst\cdot(\delta^{-2}s^2)^{-1}\sum_{z\in\cR^\delta_\mathrm{mid}} \widetilde{G}^{(\phi)}_{\cR^\delta}(z^+\!,z)\ \ge\ \cst\cdot s^{-1}\delta^3
\end{equation}
(where the factor~$\delta^{-2}s^2$ comes from the number of vertices in~$\cR^\delta_\mathrm{mid}$).

\smallskip

The last step to derive the estimate~\eqref{eq:z0->z+>delta} from~\eqref{eq:G(z0->z+)>delta3} is to use the identity
\[
\widetilde{G}{}^{(\phi)}_{\cR^\delta}(z^+\!,z_0)\ =\ \mathbb{P}^{(z_0)}[\,\widetilde{Z}^{(\phi)}_t\ \text{visits}\ z^+\ \text{before~exiting}\ \cR^\delta\,]\cdot \widetilde{G}{}^{(\phi)}_{\cR^\delta}(z^+\!,z^+)
\]
and to note that the expected time spent by the process~$\widetilde{Z}^{(\phi)}_t$ at~$z^+$ is~$O(\delta^2)$. Indeed, due to the uniform crossing property, the backward random walk started at~$z^+$ has a positive chance to reach one of staircase paths from Lemma~\ref{lem:staircase} { (see Fig.~\ref{fig:technical}C)} and to descend to~$\pa\H^\circ_{\cS^\delta}$ following this path in~$O(\delta^2)$ amount of time.
\end{proof}

\newcommand\Lik[1]{L^{(ik)}(\H^\circ_{#1})}
\newcommand\UpaH[1]{U(\pa\H^\circ_{#1})}

\subsubsection{The sets~$\Lik{\cS}$ and~$\UpaH{\cS}$} In our proof of Theorem~\ref{thm:RSW-selfdual} given in the forthcoming Section~\ref{sub:RSW-proof} we also need to `slice' the discrete upper half-plane into `horizontal lines' at level~$ik$ in a way compatible with estimates similar to those given in Proposition~\ref{prop:hm>delta}; see Proposition~\ref{prop:hm>delta-Lik}.

\smallskip

Let a constant~$C_0>0$ be fixed so that all quads~$z$ with~$\Im(\cS^\delta-i\cQ^\delta)(z)\ge C_0\delta$ belong to~$\H^\circ_{\cS^\delta}$. For~$k\ge C_0\delta$, we say that~$v\in\Lambda(G)$ \emph{lies above above the level~$ik$} in the S-graph~$\cS^\delta-i\cQ^\delta$ if
\begin{itemize}
\item $v$ is a singleton and~$\Im(\cS^\delta-i\cQ^\delta)(v)\ge k$;
\item $v$ is a part of a doubleton~$(v^\circ v^\bullet)$ and the inequality~$\Im(\cS^\delta-i\cQ^\delta)(v)\ge k$ for \emph{both}~$v=v^\circ$ and~$v=v^\bullet$.
\end{itemize}
In the opposite case we say that~$v$ \emph{lies below the level~$ik$}; note the similarity of the construction with that from Definition~\ref{def:doubleton}.

\smallskip

Let $z\mapsto v(z)$ be the bijection of~$\Dm(G)$ and~$\Lambda(G)$ defined by the S-graph~$\cS^\delta-i\cQ^\delta$.
\begin{definition}\label{def:Lik}
Given~$k\ge C_0\delta$, we denote by~$\Lik{\cS}$ the set of all~$z$ such that
\begin{itemize}
\item either $v(z)$ is a singleton lying below the level~$ik$ such that all the vertices~$v_1^\circ,v_2^\circ,\ldots$ of a staircase path from Lemma~\ref{lem:staircase} lie above the level~$ik$
\item or $v(z)$ is a part of a doubleton~$(v^\circ v^\bullet)$ lying below the level~$ik$ and both staircase paths started at~$v^\circ$ and~$v^\bullet$ go along vertices lying above the level~$ik$.
\end{itemize}
\end{definition}
As usual, the special attention paid to doubletons is caused by the fact that we need to work with backward random walks associated to the S-graphs~$\cS^\delta-ie^{2i\phi}\cQ^\delta$ with~$\phi\in[-\phi_0,\phi_0]$ and not with a single S-graph~$\cS^\delta-i\cQ^\delta$. It is easy to see that, if~$\Im(\cS^\delta-i\cQ^\delta)(v(z))\ge C_0\delta$, then
\begin{equation}
\label{eq:in-delta-Lik}
\lambda\big(\{k\ge C_0\delta:z\in\Lik{\cS^\delta}\}\big)\ \ge\ 2\veps_0\delta,
\end{equation}
where~$\lambda$ stands for the Lebesgue measure.

\smallskip

Note that we do \emph{not} pretend that~$\Lik{\cS^\delta}$ is a boundary of any discrete domain and view it simply as a subset of~$\H^\circ_{\cS^\delta}$. This slightly simplifies the considerations since, by construction, all points~$z\in\Lik{\cS^\delta}$ are `locally visible' from the set of points lying above the level~$ik$; the property which is a priori not guaranteed, e.g., for the boundary of the shifted discrete half-plane~$[\H+ik]^\circ_{\cS^\delta}$.

\begin{proposition}
\label{prop:hm>delta-Lik} In the setup of Proposition~\ref{prop:hm>delta}, the following holds for all \mbox{$k\in [C_0\delta,\tfrac{1}{2}k_0]$}, \mbox{$\phi\in[-\phi_0,\phi_0]$} and all~$z\in\Lik{\cS^\delta}$ such that~$\Re(\cS^\delta-i\cQ^\delta)(z)\le s$:

\smallskip

the probability that the backward random walk associated with \mbox{$\cS^\delta\!-\!ie^{2i\phi}\cQ^\delta$} and started {in an $O(\delta)$ vicinity of} the center of $\cR^\delta$ hits the set of vertices lying below the level~$ik$ (as defined above) at the point~$z$ and before it exits $\cR^\delta$ is at least~$\cst\cdot s^{-1}\delta$, where~$\cst>0$ depends only on \Unif\ and \Qflat.
\end{proposition}
\begin{proof}
The same proof as that of Proposition~\ref{prop:hm>delta} applies (with minor changes in the definitions of the domains~$\cR^\delta_-$ and~$\cR^\delta_\mathrm{mid}$).
\end{proof}
The last technical definition that we need is the following notation for a tiny (of width~$O(\delta)$) vicinity of the boundary~$\pa\H^\circ_{\cS^\delta}$.
\begin{definition}
\label{def:UpaH} In the same setup as above, denote
\[
\UpaH{\cS^\delta}\ :=\ \{z\in \Int H^\circ_{\cS^\delta}\,:\ \text{the~estimate~\eqref{eq:in-delta-Lik}~fails}\}.
\]
\end{definition}
The following simple fact is listed here for reference purposes: there exists a constant~$P_0>0$ depending only on constants in~\Unif\ and \Qflat\ such that for all~$z'\not\in\UpaH{\cS^\delta}$ and all~$\phi\in[-\phi_0,\phi_0]$,
\begin{equation}
\label{eq:P0-leave-UpaH}
\textstyle \sum_{z\in\UpaH{\cS^\delta}}\mathbb{P}^{(z)}[\,\widetilde{Z}^{(\phi)}_t~\text{first~leaves the set}~\UpaH{\cS^\delta}~\text{through}~z'\,]\ \le\ P_0.
\end{equation}
Indeed, the points~$z$ with~$|(\cS^\delta-ie^{2i\phi}\cQ^\delta)(z)-(\cS^\delta-ie^{2i\phi}\cQ^\delta)(z')|\ge C\delta$ give exponentially small (in~$C$) contributions to~$P_0$ because of the uniform crossing property, and there are only~$O(1)$ points with \mbox{$|(\cS^\delta\!-ie^{2i\phi}\cQ^\delta)(z)-(\cS^\delta\!-ie^{2i\phi}\cQ^\delta)(z')|=O(\delta)$}.

\resumetocwriting

\subsection{Scaling limits of fermionic observables in discrete rectangles} \label{sub:RSW-proof}
We begin with proving the convergence of fermionic observables~$F^\delta$ in discrete rectangles~$\cR^\delta=[\cR(x_1,x_2;y_1,y_2)]^{\circ\bullet\bullet\bullet}_\cS$ with \emph{Dobrushin} boundary conditions: wired at the bottom side~$(a^\delta b^\delta)^\circ$ of~$\cR^\delta$ and dual-wired along the three other sides. Though we are not interested in this result itself (\emph{a posteriori}, this is a very particular case of Theorem~\ref{thm:FK-conv}), it is useful to start with a simpler setup before moving to the analysis of fermionic observables from Section~\ref{sub:strategy-RSW} and to the proof of Theorem~\ref{thm:RSW-selfdual}.

\smallskip

Recall that the functions~$H_{F^\delta}$ associated to~$F^\delta$ satisfy the boundary conditions
\begin{equation}
\label{eq:x-H2pt-bv}
H_{F^\delta}(v)=0\ \ \text{for}\ \ v\in (a^\delta b^\delta)^\circ,\qquad H_{F^\delta}(v)=1\ \ \text{for}\ \ v\in (b^\delta a^\delta)^\bullet.
\end{equation}
Due to the maximum principle (see Proposition~\ref{prop:HF-max}), we have~$H_{F^\delta}(v)\in [0,1]$ for all~$v$ in~${\cR}^\delta$. Thus, the a priori regularity results discussed in Section~\ref{sub:regularity} {allow us} to find a subsequence~$\delta=\delta_k\to 0$ such that
\begin{equation}
\label{eq:HF-to-h-2pt}
F^\delta\to f,\quad H_{F^\delta}\to {\textstyle h=\frac{1}{2}\int\Im[(f(z))^2dz]}\quad \text{in}\ \ \cR=(x_1,x_2)\times(y_1,y_2),
\end{equation}
uniformly on compact subsets. Moreover, the function~$f$ is holomorphic, the function~$h$ is harmonic, so it only remains to prove that the \emph{boundary conditions}~\eqref{eq:x-H2pt-bv} of functions~$H_F$ imply the same boundary conditions for the continuous function~$h$. Below we focus on the behavior of~$h$ near the segment~$(ab)^\circ$, the three other sides of~$\cR$ can be {treated} similarly.

Let~$x_0\in (x_1,x_2)$ and~$s>0$ be chosen so that~$(x_0-8s,x_0+8s)\subset (x_1,x_2)$ and~$8k_0s<y_2-y_1$, where~$k_0>0$ is fixed in Section~\ref{sub:visibility}. By scaling, in what follows we assume that~$y_1=0$, $x_0=0$ and~$s=1$ without true loss of generality. The analysis of boundary conditions of~$h$ goes through a sequence of lemmas.
\begin{lemma} \label{lem:x:>-O(1)}
Let~$F=F^\delta$ be a fermionic observable in the rectangle~$\cR=\cR^\delta$ with Dobrushin boundary conditions. For all~$\phi\in[-\phi_0,\phi_0]$ (where~$\phi_0$ is given in Corollary~\ref{cor:ReF>0}), the following uniform (in~$\delta$ and in the position of~$z\in\cR$, including points near~$\pa\H^\circ_\cS$) estimate holds:
\[
\Re[e^{i\phi}F^\delta(z)]\ge-O(1)\ \ \text{if}\ \ |\Re \cS^\delta(z)|\le 2\ \ \text{and}\ \ \Im \cS^\delta(z)\le 2k_0.
\]
\end{lemma}
\begin{proof} Recall that, by Corollary~\ref{cor:H-Lip}, the functions~$F^\delta$ admit an a priori bound $|F^\delta(z)|=O(\max\{\Im \cS^\delta(z),\delta\}^{-1/2})$ and that~$\Re[e^{i\phi}F^\delta(z)]\ge 0$ on~$\pa\H^\circ_\cS$. Let \mbox{$\cT\subset\cR$} be a triangular domain bounded by $\pa\H^\circ_\cS$ and two lines~\mbox{$4k_0-y=\pm k_0x$}; {see Fig.~\ref{fig:technical}A with~$s=1$} (the choice of discretizations of {the sides of~$\cT$} is irrelevant). The claim follows by applying the optional stopping theorem for the martingale~$\Re[e^{i\phi}F(\widetilde{Z}_t)]$, where the backward random walk~$\widetilde{Z}_t$ associated with the S-graph~$\cS^\delta-ie^{2i\phi}\cQ^\delta$ is started at the point~$z$ with~$|\Re z|\le 2$,~$\Im z\le 2k_0$ and stopped when it hits~$\pa\H^\circ_\cS$ (see Definition~\ref{def:RW-at-paH}) or crosses one of the two sides of~$\cT$. The base of the triangular domain~$\cT$ gives a positive contribution to~$\Re[e^{i\phi}F(z)]$ and the (possibly, negative) contribution of its sides is~$O(1)$ provided that the slope \mbox{$k_0>0$} is chosen small enough to guarantee that the probability to hit these sides at height~$\rho\ll 1$ decays as~$O(\rho^{\frac{1}{2}+\beta})$ with~$\beta>0$. (Recall that such~$k_0$ can be also found due to the uniform crossing property of the random walk.)
\end{proof}

\begin{lemma} \label{lem:x:paH=O(1)}
In the same setup, the following uniform (in~$\delta$) estimate holds:
\[
\textstyle \delta\sum_{z\in\pa\H^\circ_{\cS^\delta}:\,|\Re \cS^\delta(z)|\le 1}|F^\delta(z)|\ =\ O(1).
\]
\end{lemma}
\begin{proof} As in the proof of the previous lemma, we apply the optional stopping theorem for~$\Re[e^{i\phi}F(\widetilde{Z}_t)]$, where the random walk~$\widetilde{Z}_t$ is started at a point~\mbox{$z\in\Dm(G)$} with~$\cS^\delta(z)=ik_0+O(\delta)$ and stopped when it hits the boundary of the same triangular domain~$\cT$. The starting value~$\Re[e^{i\phi}F(z)]$ is~$O(1)$ and the contribution of the sides of~$\cT$ is also~$O(1)$ provided that~$k_0$ is chosen small enough. We now invoke the harmonic measure estimates from Proposition~\ref{prop:hm>delta}, which imply that
\[
\textstyle\delta\sum_{z\in\pa\H^\circ_{\cS^\delta}:\,|\Re z|\le 1,\ z\ \text{is~locally~visible}}\Re[e^{i\phi}F^\delta(z)]\ \le\ O(1).
\]
Indeed, if this sum were too big, the positive contribution of the base of~$\cT$ to the starting value could not have been compensated by diagonal sides. Note that the values~\eqref{eq:2ptX>0} of the fermionic observables with Dobrushin boundary conditions at nearby boundary points are uniformly comparable to each other (as all the Ising interaction parameters are uniformly bounded from below under the assumption~\Unif). This allows to include \emph{all} the boundary points to the last sum by the cost of an additional multiplicative constant {on} the right-hand side. Thus,
\[
\textstyle\delta\sum_{z\in\pa\H^\circ_{\cS^\delta}:\,|\Re z|\le 1}\Re[e^{i\phi}F^\delta(z)]\ \le\ O(1).
\]
Since we have this estimate for both~$\phi=\pm\phi_0$ and $\Re[e^{i\phi}F^\delta(z)]\ge 0$ at~$\pa\H^\circ_{\cS^\delta}$, it also holds with~$|F^\delta(z)|$ instead of~$\Re[e^{i\phi}F(z)]$.
\end{proof}

Recall that the sets~$\Lik{\cS^\delta}$ with~$C_0\delta\le k\le\tfrac{1}{2}k_0$ and~$\UpaH{\cS^\delta}$ are introduced in Definition~\ref{def:Lik} and Definition~\ref{def:UpaH}, respectively.

\begin{lemma}\label{lem:x:Lik=O(1)}
In the same setup, the following uniform (in~$\delta$) estimates hold:
\[
\textstyle \delta\sum_{z\in\Lik{\cS^\delta}:\,|\Re \cS^\delta(z)|\le 1}|F^\delta(z)|\ =\ O(1)\ \ \text{for~all}\ \ k\in [C_0\delta,\frac{1}{2}k_0].
\]
Moreover, one also has $\delta\sum_{z\in\UpaH{\cS^\delta}:\,|\Re \cS^\delta(z)|\le 1}|F^\delta(z)|\ =\ O(1)$.
\end{lemma}
\begin{proof} Recall that~$\Re[e^{i\phi}F^\delta(z)]\ge -O(1)$ due to Lemma~\ref{lem:x:>-O(1)}, for all~$\phi\in[-\phi_0,\phi_0]$. Repeating the proof of Lemma~\ref{lem:x:paH=O(1)} with the uniform bounds on hitting probabilities from Proposition~\ref{prop:hm>delta} replaced by those from Proposition~\ref{prop:hm>delta-Lik}, it is easy to see that
\[
\textstyle \delta\sum_{z\in\Lik{\cS^\delta}:\,|\Re \cS^\delta(z)|\le 1}\Re[e^{i\phi}F^\delta(z)]\ =\ O(1),
\]
for both~$\phi=\pm\phi_0$, which proves the desired uniform estimate for the set~$\Lik{\cS^\delta}$.

The similar estimates for the sum over a tiny strip~$\UpaH{\cS^\delta}$  follow from the already obtained estimates on the boundaries of this strip and from~\eqref{eq:P0-leave-UpaH}.
\end{proof}

Let us now prove the convergence of fermionic observables in discrete rectangles with Dobrushin boundary conditions basing upon the uniform estimates from Lemmas~\ref{lem:x:>-O(1)}--\ref{lem:x:Lik=O(1)}. From the perspective of the proof of Theorem~\ref{thm:RSW-selfdual}, we do \emph{not} need the following proposition, it is included to illustrate the strategy of this proof in a simpler situation.
\begin{proposition}\label{prop:Dobrushin-rectangles} All subsequential limits~\eqref{eq:HF-to-h-2pt} of functions~$H_{F^\delta}$ associated with fermionic observables in discrete rectangles with Dobrushin boundary conditions satisfy boundary conditions~$h|_{(ab)}=0$ and~$h|_{(ba)}=1$ (and hence~$h(\cdot)=\hm_\cR(\cdot,(ba))$).
\end{proposition}
\begin{proof} We work in the same setup as above and analyze the boundary values of the function~$h$ near the bottom (wired) side of~$\cR$, the dual-wired sides can be handled similarly. Since~$|F^\delta(z)|=O((\max\{\Im\cS^\delta(z),\delta\}^{-1/2})$ due to the a priori regularity of fermionic observables and to the trivial estimate~$|H_{F^\delta}|\le 1$, it easily follows from Lemma~\ref{lem:x:Lik=O(1)} that
\[
\textstyle \delta\sum_{z\in\Lik{\cS^\delta}:\,|\Re \cS^\delta(z)|\le 1}|F^\delta(z)|^2\ =\ O({k}^{-1/2})\ \ \text{for~all}\ \ k\in [C_0\delta,\frac{1}{2}k_0]
\]
and~$\phi\in[-\phi_0,\phi_0]$, and that $\delta\sum_{z\in\UpaH{\cS^\delta}:\,|\Re \cS^\delta(z)|\le 1}|F^\delta(z)|^2\ =\ O(\delta^{-1/2})$.

Recall that~$H_{F^\delta}$ and~$F^\delta$ are linked by~\eqref{eq:HF-def}. As the function~$H_{F^\delta}$ has boundary values~$0$ at~$\pa\H^\circ_{\cS^\delta}$, its values near the segment~$[-1,1]+iy$ can be represented by integrals~\eqref{eq:HF-def} computed over \emph{vertical} segments starting at~$\pa\H^\circ_{\cS^\delta}$. Taking into account the property~\eqref{eq:in-delta-Lik},
this allows us to conclude that
\begin{equation}
\label{eq:HF-near-paH}
\textstyle \delta\sum_{v:\,|\Re \cS^\delta(v)|\le 1,\,\Im\cS^\delta(v)=y+O(\delta)}H_{F^\delta}(v)\ =\ O(y^{1/2}+\delta^{1/2})
\end{equation}
for all~$y\le \frac{1}{2}k_0$. Passing to the limit~$\delta\to 0$, we see that
\begin{equation}
\label{eq:h-bc}
\int_{-1}^{1}h(x+iy)dx\ =\ O(y^{1/2})\ \ \text{uniformly~for}\ \ y\in (0,\tfrac{1}{2}k_0].
\end{equation}
Since~$h$ is a \emph{non-negative} harmonic function, sending $y\to 0$ in~\eqref{eq:h-bc} implies that~$h$ satisfies Dirichlet boundary conditions on the segment~$[-1,1]$ (this easily follows, e.g., from the Herglotz representation theorem).
\end{proof}

We are now ready to give a proof of Theorem~\ref{thm:RSW-selfdual}. Let~$\cR=(x_1,x_2)\times(y_1,y_2)\subset\C$, $x_1<x_2$,~$y_1<y_2$. Following the strategy described in Section~\ref{sub:strategy-RSW}, assume that
\begin{equation}
\label{eq:HF-to-h-4pt}
F^\delta\to f,\quad H_{F^\delta}\to {\textstyle h=\frac{1}{2}\int\Im[(f(z))^2dz]}\quad \text{on~compact~subsets~of~}\cR,
\end{equation}
where~$F^\delta$ are the fermionic observables in rectangles~$\cR^\delta=[\cR(x_1,x_2;y_1,y_2)]^{\circ\bullet\circ\bullet}_{\cS^\delta}$
leading to the special boundary conditions~\eqref{eq:HF4-bc} for the associated functions~$H_{F^\delta}$.
Recall that the function~$f:\cR\to\C$ is holomorphic and the function~$h:\cR\to [0,1]$ is harmonic.
Arguing by contradiction, assume that
\begin{equation}
\label{eq:RSW-contrary}
\cos^2{\mathrm{p}}=\cos^2({\mathrm{p}}^\delta)\to 1\ \ \text{as}\ \ \delta\to 0.
\end{equation}
The first step in the proof of Theorem~\ref{thm:RSW-selfdual} is to prove that all subsequential limits~\eqref{eq:HF-to-h-4pt} of functions~$H_{F^\delta}$ inherit the boundary conditions~\eqref{eq:HF4-bc}.

\begin{proposition} \label{prop:hbc-0111}
In the setup of Theorem~\ref{thm:RSW-selfdual} and under the assumption~\eqref{eq:RSW-contrary}, each subsequential limit~$h$ of functions~$H_{F^\delta}$ defined above has boundary values~$0$ at the top side~$(a_2b_2)$ of~$\cR$ and boundary values~$1$ at the three other sides of~$\cR$.
\end{proposition}
\begin{proof} We focus our attention on the analysis near the bottom sides~$(a_1^\delta b_1^\delta)^\circ$ of discrete rectangles~$\cR^\delta$, the three other sides can be {treated} in a similar way. As above, let~$x_0$ and~$s>0$ be such that~$(x_0-8s,x_0+8s)\subset(x_1,x_2)$, $8k_0s\le y_2-y_1$, and assume that~$x_0=y_1=0$ and~$s=1$ by shifts and by the scaling.

Denote by~$F^\delta_0$ the fermionic observable in the rectangle~$[\cR(x_1,x_2;y_1,y_2)]^{\circ\bullet\bullet\bullet}_{\cS^\delta}$ with \emph{Dobrushin} boundary conditions. By definition, if~$z=(v^\bullet v^\circ_- z v^\circ_+)$ is a boundary half-quad on~$(a_1^\delta b_1^\delta)$ and~$c_\pm:=c_{(v^\circ_\pm v^\bullet)}$, then
\begin{align*}
|X(c_-)|=|X(c_+)|\ &=\ \big|\E^{({\mathrm{p}})}_\cR[\mu_{v^\bullet}\sigma_{(a_1b_1)^\circ}\mu_{(b_1a_2)^\bullet}\sigma_{(a_2b_2)^\circ}]\big|\\
&\le\ \E^{({\mathrm{p}})}_\cR[\mu_{v^\bullet}\mu_{(b_1a_2)^\bullet}]\ \le\ \E_{\cR^{\circ\bullet\bullet\bullet}}[\mu_{v^\bullet}\mu_{(a_2b_1)^\bullet}]\ =\ X_0(c_\pm),
\end{align*}
where the last expectation is taken with Dobrushin boundary conditions (i.e., dual-wired along \emph{all} three sides of~$\cR$ except~$(a_1b_1)^\circ$) and the Kadanoff--Ceva fermionic observables~$X$, $X_0$ correspond to~$F$, $F_0$ in a standard way (see Proposition~\ref{prop:shol=3term}).

Due to Lemma~\ref{lem:argF=} and Corollary~\ref{cor:ReF>0}, this inequality implies the a priori bound
\[
|\Re [e^{i\phi}F^\delta(z)]|\ \le\ \Re[e^{i\phi}F^\delta_0(z)],\qquad z\in(a_1^\delta b_1^\delta)^\circ,
\]
for all~$\phi\in[-\phi_0,\phi_0]$ with~$\phi_0>0$. Recall now the triangular domain~$\cT$ from the proof of Lemma~\ref{lem:x:>-O(1)}, which is bounded by~$\pa\H^\circ_\cS$ and the lines \mbox{$4k_0-y=\pm k_0x$}, with a sufficiently small~$k_0>0$; {see Fig.~\ref{fig:technical}A with~$s=1$.} The optional stopping theorem for the martingales~$\Re[e^{i\phi}(F^\delta\pm F^\delta_0)(\widetilde{Z}_t)]$ started at a point with~$|\Re\cS^\delta(z)|\le 2$, $\Im\cS^\delta(z)\le 2k_0\delta$, implies the uniform estimate
\[
|\Re [e^{i\phi}F^\delta(z)]|\ \le\ |\Re[e^{i\phi}F^\delta_0(z)]|+O(1)\quad \text{if}\ \ |\Re\cS^\delta(z)|\le 2,\ \Im\cS^\delta(z)\le 2k_0\delta.
\]
Thus, the uniform estimates from Lemma~\ref{lem:x:Lik=O(1)} -- originally proved for the fermionic observables~$F^\delta_0$ with Dobrushin boundary conditions -- also hold for the observables~$F^\delta$. Therefore, one can repeat the proof of Proposition~\ref{prop:Dobrushin-rectangles} with~$F^\delta_0$ replaced by~$F^\delta$. In this setup, the analogue of the estimate~\eqref{eq:HF-near-paH} reads as
\[
\textstyle \delta\sum_{v:\,|\Re \cS^\delta(v)|\le 1,\,\Im\cS^\delta(v)=y+O(\delta)}\big|\cos^2({\mathrm{p}}^\delta)-H_{F^\delta}(v)\big|\ =\ O(y^{1/2}+\delta^{1/2})
\]
and hence
\[
\int_{-1}^{1}(1-h(x+iy))dx\ =\ O(y^{1/2})\ \ \text{uniformly~for}\ \ y\in (0,\tfrac{1}{2}k_0].
\]
This is enough to conclude that the harmonic function~$1-h\ge 0$ has boundary values~$0$ at the segment~$[-1,1]$.
\end{proof}

Proposition~\ref{prop:hbc-0111} implies that the only possible subsequential limit~\eqref{eq:HF-to-h-4pt} of functions~$H_{F^\delta}$ under the assumption~\eqref{eq:RSW-contrary} is~$h(\cdot)=1-\hm_\cR(\cdot,(a_2b_2))=\hm_\cR(\cdot,(b_2a_2))$. To complete the proof of Theorem~\ref{thm:RSW-selfdual} it remains to rule out this scenario.

\begin{proof}[{\bf Proof of Theorem~\ref{thm:RSW-selfdual}}] Assume, by contradiction, that the convergence~\eqref{eq:HF-to-h-4pt} holds with \mbox{$h(\cdot)=\hm_\cR(\cdot,(b_2a_2))$} and~$f^2=i\pa h=\tfrac{i}{2}(\pa_xh-i\pa_yh)$. As above, we assume that~$y_1=0$ and work in a vicinity~$(x_0-8s,x_0+8s)\subset(x_1,x_2)$ of a point \mbox{$x_0=0$}. Note that now we do not assume that~$s=1$ (and actually will choose it sufficiently small to get a contradiction below).

From the Dirichlet boundary conditions of~$h$ near~$0$ it is easy that there exists a \emph{real} constant~$f_0\ne 0$ such that
\[
(f(w))^2\ =\ -f_0^2+O(|w|)\ \ \text{for}\ \ w\in\cR\ \text{near}\ 0.
\]
Let~$\cT$ be the same triangular subdomain of~$\cR$ as above {(see Fig.~\ref{fig:technical}A),} with the base~$\pa\H^\circ_{\cS^\delta}$ and the sides~${4k_0s}-y=\pm k_0x$. Note that the convergence~\eqref{eq:HF-to-h-4pt} yields
\begin{align}
\Re[e^{\pm i\phi_0}F^\delta(z)]\ &=\ \Re[e^{\pm i\phi_0}(f(\cS^\delta(z))+o_{\delta\to 0}(1))] \notag \\
&=\ -|f_0|\sin\phi_0+O(s)+o_{\delta\to 0}(1),
\label{eq:ReF-near-R<0}
\end{align}
uniformly for~$\cS^\delta(z)$ away from~$\R$, provided that the sign in~$\pm\phi_0$ is chosen appropriately. {We claim that, provided that~$s>0$ is small enough, the asymptotics~\eqref{eq:ReF-near-R<0} in not compatible with the positive boundary values \mbox{$\Re[e^{\pm i\phi_0}F^\delta(z)]\ge 0$} at~$\pa H^\circ_{\cS^\delta}$; see Remark~\ref{rem:4pt>0} and Corollary~\ref{cor:ReF>0} for this positivity property.}

To get a contradiction, consider the martingale~$\Re[e^{\pm i\phi_0}F^\delta(\widetilde{Z}_t)]$ started at a point~$z$ such that~$\cS^\delta(z)=ik_0s+O(\delta)$ and stopped at the boundary of~$\cT$. The probability to hit the base of~$\cT$ is uniformly (in~$\delta$) bounded from below by a constant~$p_0>0$ due to the uniform crossing property and this event produces a \emph{non-negative} contribution. In view of~\eqref{eq:ReF-near-R<0}, the two sides of~$\cT$ -- except small vicinities of radius~$\rho\ll s$ near the corner points~$\pm 4s$ of~$\cT$ -- contribute~$(1-p_0)(-|f_0|\sin\phi_0+O(s)+o_{\delta\to 0}(1))$. Finally, the contribution of these tiny vicinities of~$\pm 4s$ can be uniformly (in~$\delta$) estimated by~$O((\rho/s)^\beta)$ with~$\beta>0$ as in the proof of Lemma~\ref{lem:x:>-O(1)}. To summarize, the expectation at the hitting time is greater than~$-(1-\frac{1}{2}p_0)|f_0|\sin\phi_0$ provided that $s>0$, $\rho>0$, and then~$\delta>0$ are chosen small enough. 
\end{proof}

\begin{remark}
Let us emphasize that our proof of Theorem~\ref{thm:RSW-selfdual} heavily relies upon the fact that we consider very \emph{special} discrete rectangles~$\cR^\delta$ (and not \emph{generic} discrete domains~$(\Od;a_1,b_1,a_2,b_2)$ as in~\cite[Theorem~6.1]{ChSmi2}). A posteriori, i.e. once Theorem~\ref{thm:FK-conv} (which relies upon Corollary~\ref{cor:circuits} of Theorem~\ref{thm:RSW-selfdual}) is proven, one can generalize~\cite[Theorem~6.1]{ChSmi2} to the setup of s-embeddings satisfying the same assumptions as in Theorem~\ref{thm:FK-conv}. However, this requires an additional work and we do not include such a derivation {into} this paper.
\end{remark}

\subsection{Proof of Corollary~\ref{cor:circuits}}\label{sub:circuits} The argument given below mimics the proof of~\cite[Proposition~2.10]{duminil-garban-pete}, see also~\cite[Section~5.4]{duminil-parafermions} for more references.

Let~$\cR:=(\Re u-3d,\Re u-d)\times (\Im u-4d,\Im u+4d)\subset\C$. By the FKG inequality, it is enough to prove that
\begin{equation}
\label{eq:middle-to-middle}
\mathbb{P}^{\operatorname{free}}_{\cR}
\biggl[\begin{array}{l}\text{there exists an open path in $\cR$ from $v_1$ to $v_2$ such that}\\
\Im\cS^\delta(v_1)\le\Im u-3d\ \text{and}\ \Im\cS^\delta(v_2)\ge \Im u+3d\end{array}\biggr]\,\ge\ p_0^{1/4}
\end{equation}
and similarly for the rectangle~$(\Re u+d,\Re u+3d)\times (\Im u-4d,\Im u+4d)$ and for two horizontal rectangles $(\Re u-4d,\Re u+4d)\times (\Im u\pm3d,\Im u\pm d)$. Applying shifts and the scaling, we can assume that~$\cR=(-1,1)\times(-4,4)$ and~$\delta\le L_0^{-1}$.

Assume, by contradiction, that~\eqref{eq:middle-to-middle} does not hold. Then, it should exist a sequence of s-embeddings~$\cS^\delta$ with~$\delta\to0$ (note that we use \Unif\ here) such that
\[
\liminf_{\delta\to 0}\mathbb{P}^{\operatorname{free}}_\cR\Biggl[\begin{array}{l}\text{there~exists~an~open~path~in}\ \cR\ \text{from}\ v_1\ \text{to}\ v_2\\
\text{such~that}\ \Im\cS^\delta(v_1)\le -3\ \text{and}\ \Im\cS^\delta(v_2)\ge 3\end{array}\Biggr]\ =\ 0.
\]

Denote \mbox{$\cR_-\!:=(-1,1)\!\times\!(-4,-3)$}, \mbox{$\cR_\mathrm{mid}\!:=(-1,1)\!\times\!(-3,3)$}, \mbox{$\cR_+\!:=(-1,1)\!\times\!(3,4)$}. We~call the event considered above an \emph{open vertical crossing of~$\cR_\mathrm{mid}$}, and use similar terminology for crossings in~$\cR_\pm$. Using the monotonicity of the probability to have open/closed crossings with respect to boundary conditions (which is again a corollary of the FKG inequality), one sees that
\begin{align}
\notag 
&\mathbb{P}^{\operatorname{free}}_\cR\big[\text{there exists an open vertical crossing of $\cR_\mathrm{mid}$}\big]\\
&\hskip 24pt \ge\
\mathbb{P}^{\operatorname{free}}_\cR\Biggl[\begin{array}{l}
\text{there exists an open vertical crossing of $\cR_\mathrm{mid}$}\\
\hskip 48pt \&\ \text{a closed horizontal crossing of $\cR_-$}\\
\hskip 48pt \&\ \text{a closed horizontal crossing of $\cR_+$}
\end{array}\Biggr]\notag\\
&\hskip 24pt \ge\ \mathbb{P}^{\operatorname{wfwf}}_\cR\Biggl[\begin{array}{l}
\text{there exists an open vertical crossing of $\cR_\mathrm{mid}$}\\
\hskip 48pt \&\ \text{a closed horizontal crossing of $\cR_-$}\\
\hskip 48pt \&\ \text{a closed horizontal crossing of $\cR_+$}
\end{array}\Biggr],\notag
\end{align}
where we denote by `wfwf' the wired boundary conditions (all edges are open) on the horizontal sides of a rectangle and free boundary conditions (all edges are closed) on the vertical sides. Using the FKG inequality again, one obtains the estimates
\begin{align}
&\mathbb{P}^{\operatorname{wfwf}}_\cR\Biggl[\begin{array}{l}
\text{there exists an open vertical crossing of $\cR_\mathrm{mid}$}\\
\hskip 24pt \&\ \text{a closed horizontal crossing of $\cR_-$}\\
\hskip 24pt \&\ \text{a closed horizontal crossing of $\cR_+$}
\end{array}\Biggr],\notag\\
&\hskip 24pt \ge\ \mathbb{P}^{\operatorname{wfwf}}_{\cR}\big[\text{there exists an open vertical crossing of $\cR_\mathrm{mid}$}\big] \notag\\
&\hskip 40pt \times\ \mathbb{P}^{\operatorname{wfwf}}_{\cR_-}\big[\text{there exists a closed horizontal crossing of $\cR_\mathrm{-}$}\big] \label{eq:wfwf-wfwf-wfwf}\\
&\hskip 40pt \times\ \mathbb{P}^{\operatorname{wfwf}}_{\cR_+}\big[\text{there exists a closed horizontal crossing of $\cR_\mathrm{+}$}\big], \notag
\end{align}
Also, one trivially has
\begin{align*}
&\mathbb{P}^{\operatorname{wfwf}}_{\cR}\big[\text{there exists an open vertical crossing of $\cR_\mathrm{mid}$}\big]\\
&\hskip 24pt \ge\ \mathbb{P}^{\operatorname{wfwf}}_{\cR}\big[\text{there exists an open vertical crossing of $\cR$}\big].
\end{align*}

We are now (almost) in the setup of Theorem~\ref{thm:RSW-selfdual}: it remains to note, e.g., that
\begin{align*}
&\mathbb{P}^{\operatorname{wfwf}}_{\cR}\big[\text{there exists an open vertical crossing of $\cR$}\big]\\
&\hskip 24pt \ge\ \E_{\widetilde{\cR}{}^\delta}[\sigma_{(a_1^\delta b_1^\delta)^\circ}\sigma_{(a_2^\delta b_2^\delta)^\circ}],\qquad \widetilde{\cR}{}^\delta:=[\cR(-\tfrac{1}{2},\tfrac{1}{2};-\tfrac{9}{2},\tfrac{9}{2})]^{\circ\bullet\circ\bullet}_{\cS^\delta}
\end{align*}
(and similarly for~$\cR_\pm$) where we again used the monotonicity with respect to boundary conditions and the fact that the boundaries of discrete rectangles~$\widetilde{\cR}{}^\delta$ constructed in Section~\ref{sub:cuts} stay within~$O(\delta)$ from the corresponding horizontal and vertical lines under the assumption~\Qflat. Thus, all the three probabilities in~\eqref{eq:wfwf-wfwf-wfwf} are uniformly bounded from below due to Theorem~\ref{thm:RSW-selfdual}, which leads to a contradiction.

\renewcommand\thesection{A}
\section{Appendix}
\setcounter{equation}{0}
For the sake of completeness of the presentation, below we give the proofs of two simple estimates `in {the} continuum' that were used in Section~\ref{sec:convergence}.

Recall that, given a simply connected domain~$\Omega\subset\C$ and a point~$u\in\Omega$, we denote by~$d_u=d(u,\Omega)$ the distance from~$u$ to the boundary of~$\Omega$, and by~$\crad(u,\Omega)$ the \emph{conformal radius} of~$u\in\C$: if~$\varphi_u:\mathbb{D}\to \Omega$ is conformal and~$\varphi_u(0)=u$, then~$\crad(u,\Omega)=|\varphi'_u(0)|$. The Koebe~$\frac{1}{4}$-theorem says that~$\frac{1}{4}d_u\le\crad(u,\Omega)\le 4d_u$.
\begin{lemma}\label{lem:crad-estimates} The function $u\mapsto\crad(u,\Omega)$ is smooth, its gradient is bounded by~$4$ and all its second derivatives are uniformly (both in~$u$ and~$\Omega$) bounded by~$18d_u^{-1}$.
\end{lemma}
\begin{proof}
Without loss of generality (by translating and scaling~$\Omega$ appropriately), we can assume that~$u=0$ and that~$\crad(u,\Omega)=1$. Under these assumptions, we have $\varphi_u(z)=z+a_2z^2+a_3z^3+O(z^4)$, $z\to 0$, and
$\varphi^{-1}(w)=w-a_2w^2+O(w^3)$, $w\to 0$. If~$w\in \Omega$, then
\[
\crad(w,\Omega)\ =\ \varphi'_u(\varphi_u^{-1}(w))\cdot(1-|\varphi^{-1}_{u}(w)|^2)
\]
since~$\varphi_u^{-1}\circ\varphi_w$ is the M\"obius mapping~$z\mapsto (z+\varphi_u^{-1}(w))/(1+\overline{\varphi_u^{-1}(w)}z)$. Therefore,
\[
\crad(w,\Omega)\ =\ 1+2a_2w+(3a_3-2a_2^2)w^2-|w|^2+O(|w|^3),\quad w\to 0.
\]
The desired estimates for derivatives of the function~$\crad(u,\Omega)$ at~$u=0$ follow from the classical bounds~$|a_2|\le 2$ and~$|a_3|\le 3$ for the coefficients of univalent maps.
\end{proof}

The second lemma is slightly more involved.

\begin{lemma} \label{lem:cont-estimate} Let~$\Omega\subset\C$ be a bounded simply connected domain, $\alpha\in(0,2]$, and a function~$h\in C^2(\Omega)\cap C(\overline\Omega)$ satisfy the Dirichlet boundary conditions~$h=0$ at~$\partial\Omega$. If~$|\Delta h(u)|\le d_u^{-2+\alpha}$ for all~$u\in\Omega$, then
\[
|h(w)|\le \cst(\alpha)\cdot d_w^{\beta(\alpha)}\cdot (\diam(\Omega))^{\alpha-\beta(\alpha)}\ \ \text{for all}\ \ w\in\Omega,
\]
where~$\cst(\alpha),\beta(\alpha)>0$ depend only on~$\alpha$. In particular, one can choose~$\beta(\alpha):=\frac{1}{5}\alpha$.
\end{lemma}
\begin{proof} Applying the scaling~$\Omega\mapsto k\Omega$, $h(u) \mapsto k^\alpha h(k^{-1}u)$, we can assume (without loss of generality) that~$\diam(\Omega)=1$. Let~$G_\Omega(w,u)$ be the Green function in~$\Omega$. Clearly,
\[
|h(w)|\le\int_\Omega |G_\Omega(w,u)|d_u^{-2+\alpha}dA(u).
\]
Let~$\varphi_w:\mathbb{D}\to\Omega$ be a conformal mapping such that~$\varphi_w(0)=w$. The Koebe $\frac{1}{4}$-theorem and the conformal invariance of the Green function yield the estimate
\begin{equation}
\label{eq:A-x}
|h(w)|\le 16\int_\mathbb{D}|G_\mathbb{D}(z,0)|(1\!-\!|z|)^{-2+\alpha}|\varphi_w'(z)|^\alpha dA(z);
\end{equation}
recall that
\[
G_\mathbb{D}(z,0)(1\!-\!|z|)^{-2+\alpha}=\tfrac{1}{2\pi}\log|z|\cdot (1\!-\!|z|)^{-2+\alpha}\sim-\tfrac{1}{2\pi}(1\!-\!|z|)^{-1+\alpha},\quad |z|\to 1.
\]

We now choose a small parameter~$\gamma>0$ (e.g.,~$\gamma:=\frac{2}{5}$), split the unit disc~$\mathbb{D}$ into the regions~$|z|\ge 1-2d_w^\gamma$ and~$|z|\le 1-2d_w^\gamma$, and estimate each of the corresponding terms in~\eqref{eq:A-x} via the H\"older inequality with the exponents~$2/(2-\alpha)$ and~$2/\alpha$.

Since~$\int_\mathbb{D}|\varphi_w'(z)|^2dA(z)=\Area(\Omega)\le \pi$ (recall that~$\diam(\Omega)=1$), we have
\begin{align}
\int_{z\in\mathbb{D}:|z|\ge 1-2d_w^\gamma}&|G_\mathbb{D}(z,0)|(1\!-\!|z|)^{-2+\alpha}|\varphi_w'(z)|^\alpha dA(z) \notag\\ &\le\
\mathrm{cst}\cdot \biggl[\int_0^{2d_w^\gamma}t^{\frac{2(-1+\alpha)}{2-\alpha}}dt\biggr]^{\frac{2-\alpha}{2}}\ =\ \mathrm{cst}(\alpha)\cdot d_w^{\frac{\alpha\gamma}{2}}.
\label{eq:A-x1}
\end{align}
Similarly, the contribution of the second region can be estimated as follows:
\begin{align}
\int_{z\in\mathbb{D}:|z|\le 1-2d_w^\gamma}&|G_\mathbb{D}(z,0)|(1\!-\!|z|)^{-2+\alpha}|\varphi_w'(z)|^\alpha dA(z) \notag\\ &\le\
\mathrm{cst}(\alpha)\cdot (\Area(\varphi_w(\{z:|z|\le 1-2d_w^\gamma\})))^{\frac{\alpha}{2}}.
\label{eq:A-x2}
\end{align}
Note that
\[
|z|\le 1-2d_w^\gamma\ \ \Rightarrow\ \ \mathrm{hdist}_\Omega(w,\varphi_w(z))=\mathrm{hdist}_\mathbb{D}(0,z)=\int_0^{|z|}\!\frac{dr}{1-r^2}\ \le\ -\frac{\gamma}{2}\log d_w,
\]
where~$\mathrm{hdist}_\Omega$ stands for the hyperbolic distance in~$\Omega$ (thus, $\mathrm{hdist}_\mathbb{D}$ is the distance in the Poincar\'e metric in~$\mathbb{D}$). Due to the Koebe theorem, the hyperbolic metric element at a point~$u\in\Omega$ is bounded from below by~$\frac{1}{4}d_u^{-1}$. Therefore,
\[
\mathrm{hdist}_\Omega(w,\varphi_w(z))\ \ge\ \frac{1}{4}\int_0^{|\varphi_w(z)-w|}\!\frac{dt}{d_w+t}\ =\ \frac{1}{4}\log\biggl(1+\frac{|\varphi_w(z)\!-\!w|}{d_w}\biggr).
\]
In particular, the condition~$|z|\le 1\!-\!2d_w^\gamma$ implies that~$|\varphi_w(z)\!-\!w|\le d_w^{1-2\gamma}$. This allows us to estimate~\eqref{eq:A-x2} from above by~$\cst(\alpha)\cdot d_w^{\alpha(1-2\gamma)}$\!. The proof is complete.
\end{proof}


\end{document}